\documentclass[11pt]{article}

\usepackage{color}
\usepackage[margin=1in]{geometry}
\usepackage{amsmath,amssymb,amsthm,mathtools}
\usepackage{microtype}
\usepackage[round,authoryear]{natbib}
\usepackage[hidelinks]{hyperref}
\usepackage{tikz}
\usepackage{subcaption}
\usepackage{enumerate}
\usetikzlibrary{arrows.meta}
\makeatletter
\pgfmathdeclarefunction{Phinorm}{1}{%
	\begingroup
	\pgfmathsetmacro\ax{abs(#1)}%
	\pgfmathsetmacro\tq{1/(1+0.2316419*\ax)}%
	\pgfmathsetmacro\dens{0.3989422804014327*exp(-0.5*\ax*\ax)}%
	\pgfmathsetmacro\poly{\tq*(0.319381530+\tq*(-0.356563782+\tq*(1.781477937%
		+\tq*(-1.821255978+1.330274429*\tq))))}%
	\pgfmathsetmacro\res{ifthenelse(#1<0,\dens*\poly,1-\dens*\poly)}%
	\pgfmathsmuggle\res
	\endgroup
	
}
\makeatother

\usepackage{setspace}
\newtheorem{proposition}{Proposition}
\newtheorem{lemma}{Lemma}
\newtheorem{corollary}{Corollary}
\newtheorem{remark}{Remark}
\newtheorem{assumption}{Assumption}

\newcommand{\E}{\mathbb{E}}
\newcommand{\Prb}{\mathbb{P}}
\newcommand{\ind}{\mathbf{1}}
\newcommand{\NR}{\mathrm{NR}}

\newcommand{\Cov}{\mathrm{Cov}}
\newcommand{\pw}{\mathrm{pw}}

\title{	Midterm Review	}

\author{Doruk Cetemen\thanks{Department of Economics and Financial Markets, Luiss Guido Carli University and EIEF: dcetemen@luiss.it}
	\and Yonggyun Kim\thanks{Department of Economics, Florida State University: ykim22@fsu.edu}
	\and Fei Li\thanks{Department of Economics, University of North Carolina: lifei@email.unc.edu}
	\and Curtis R. Taylor\thanks{Department of Economics, Duke University: curtis.taylor@duke.edu}}

\date{\today}

\begin{document}
	\hypersetup{pageanchor=false}
	
	\maketitle
	\bigskip
	\begin{abstract}
		We study why organizations conduct interim performance reviews when monetary rewards are limited. An interim review creates incentive capacity by allowing future work and career opportunities to serve as rewards for past performance. Optimal review policies map a continuum of performance outcomes into a simple incentive ladder: termination, tough or easy continuation, and, for exceptional performance, an early maximal reward with no further work. Review can even sustain high effort when terminal compensation alone cannot. Its timing balances two forces: waiting improves the information revealed by performance, but leaves less future work available to motivate the agent.
	\end{abstract}

	\thispagestyle{empty}
	
	\clearpage
	\setcounter{page}{1}
	\hypersetup{pageanchor=true}
	\onehalfspacing
	\clearpage

	\section{Introduction}
	\label{sec:introduction}
	
	Performance reviews are ubiquitous in organizations and often have substantial
	consequences for an employee's career.  Evaluations affect bonuses and merit pay, but also
	promotion, retention, and dismissal \citep{cappelliconyon2018}.  This raises a basic question.
	If an organization can ultimately observe an informative measure of performance and reward
	the employee accordingly, why evaluate performance before the end of the current phase of employment?
	In a standard moral-hazard model with unrestricted monetary transfers, there need be no
	reason to do so.  Indeed, it is optimal to bundle incentives from different stages into a sufficiently powerful
 terminal compensation scheme.
	
	This paper shows that the answer changes fundamentally when monetary rewards are
	bounded.  Let $M$ denote the largest performance-contingent payment that the principal can
	make.  The cap may be literal, but more broadly it captures organizations in which large
	financial bonuses are difficult, costly, or institutionally constrained.  Such restrictions are
	particularly natural in public-sector employment, where pay is relatively compressed and
	promotion is an important source of incentives \citep{mas2017,deserranno2025}, and in
	nonprofit organizations, which make less use of output-contingent incentive pay than
	for-profit firms \citep{devaro2007}.  More generally, many careers provide substantial
	incentives through advancement rather than immediate cash compensation.  Tenure and
	promotion in academia and promotion to partnership in law, accounting, consulting, and
	other professional-service firms are prominent examples
	\citep{bakerjensenmurphy1988,barlevyneal2019,levintadelis2005}.  In these environments,
	the opportunity to continue on a favorable career path is itself a reward.
	
	This observation gives interim review a role that is absent when monetary incentives are
	unbounded.  A review allows the principal to allocate future opportunities according to past
	performance.  An agent who performs well can be offered a more attractive continuation
	contract, while one who performs poorly can face a demanding continuation standard or
	termination.  Rents that must be conceded to motivate future effort can therefore do double
	duty: besides inducing future effort, they can be used as rewards for past effort.  We refer
	to these future rents and continuation opportunities as \emph{incentive capital}.  When $M$
	is large, the principal can instead provide incentives directly with money and incentive
	capital has little value.  When $M$ is tight, continuation opportunities become valuable
	precisely because monetary rewards are scarce.
	
	We study this mechanism in a canonical dynamic moral-hazard environment.  A
	risk-neutral agent produces cumulative Gaussian output over a finite horizon and privately
	chooses costly effort.\footnote{For ease of exposition, we assume that effort in the baseline model is chosen only twice, before the performance review and after it.  We relax this assumption and allow the agent to choose effort continuously in Section~\ref{sec:pointwise-effort}.}
	A risk-neutral principal receives output net of compensation.
	Compensation satisfies both limited liability and is bounded above by $M$.  The principal may commit to one public interim review, choose when it occurs, and condition both compensation and
	the continuation opportunity on the review score.  Past performance contains information
	about past effort but does not change the technology, productivity, or difficulty of future
	work.  Thus review has no sorting, learning, or technological role in the baseline: any value
	it creates comes purely from incentives.
	
	Our first main result is that a continuum of possible review scores is optimally mapped
	into at most four discrete actions.  These actions form an \emph{incentive ladder}, ordered
	by the rent they deliver to the agent.  At the bottom is a \emph{bad stop}: sufficiently poor
	performance leads to termination with no payment.  The next rung is continuation under a
	difficult terminal standard, which induces future effort while leaving the agent relatively
	little rent.  Better review performance earns continuation under an easier terminal
	standard and hence a larger rent.  Finally, when first-phase incentives are sufficiently
	scarce, exceptional performance reaches a \emph{good stop}: the agent receives the maximal
	payment $M$ immediately and is released from the remaining workload.  Thus termination
	can optimally occur at both tails of the performance distribution.  The bad stop is the
	strongest punishment available; the good stop is the strongest reward available once money
	alone can no longer reward performance.  Depending on parameters, two, three, or all four
	rungs are used.
	
	The discreteness of this ladder is notable because both performance and the feasible set of
	continuation contracts are continuous.  Its structure comes from the geometry of bounded
	incentive provision.  Among contracts that induce future effort, the principal uses only the
	minimum- and maximum-rent continuation contracts.  Together with the two extreme
	no-effort actions---termination without pay and termination with the maximal payment---these
	are the only exposed actions in the principal's statewise assignment problem.  Monotone
	likelihood ratios then order the four actions by review performance.  In this sense, a
	continuous performance evaluation optimally produces a coarse rating system with only a
	small number of economically distinct consequences.
	
	Our second set of results establishes when review is valuable.  The payment bound $M$ is
	again the key comparative static.  When $M$ is sufficiently large, no genuine review can
	improve on the terminal contract: uniformly over all review dates, the principal strictly
	prefers to wait until the end.  As the cap tightens, however, the terminal contract exhausts
	its ability to provide incentives and continuation rents become increasingly valuable.
	Near the smallest cap at which full effort can be implemented without review, an interior
	review is always strictly profitable and all four rungs of the incentive ladder are used.
	
	Indeed, review can do more than reduce the cost of incentives.  It can make high effort
	implementable when terminal evaluation alone cannot.  At the no-review implementability
	boundary, splitting the project into two shorter incentive problems creates strict incentive
	slack in each phase.  For payment caps slightly below that boundary, full effort is impossible
	under any terminal-only contract but can be induced with an interim review.  Review therefore
	expands the organization's effective incentive capacity: future rents that would otherwise
	be an unavoidable cost of motivating later effort become an additional currency with which
	to motivate earlier effort.
	
	Our third main result concerns \emph{when} the review should occur.  Delay has an obvious
	informational advantage.  With more performance accumulated before the review, the
	distributions generated by working and shirking are easier to distinguish.  But waiting also
	uses up the very future opportunities that make review valuable.  A later review leaves less
	remaining work, a smaller range of continuation rents, and therefore less incentive capital
	with which to reward first-phase performance.  Optimal review timing balances better
	information against more time left to motivate.
	
	The endpoint results make this tradeoff particularly sharp.  An arbitrarily early review is
	strictly harmful: in the Brownian model the right derivative of review value at date zero is
	minus infinity.  A very short probationary period generates almost no useful information,
	while bounded rewards prevent the principal from cheaply magnifying that information into
	incentives.  Every strictly optimal review is therefore bounded away from the initial date.
	At the other endpoint, when $M$ is tight, moving the review slightly before the terminal date
	has first-order value because even a short continuation phase creates scarce incentive
	capital.  Thus the model predicts neither ``review as soon as possible'' nor ``wait until
	performance is most informative.''  The optimum is governed by the tension between
	learning about past effort and preserving future opportunities with which to reward it.
	
	The optimal mechanism that ultimately awards a \textit{prize} of either $0$ or $M$ also suggests a life-cycle interpretation.  Interim evaluation should be
	especially valuable when an agent's incentives depend heavily on future advancement---for
	example, for junior academics approaching tenure or young professionals competing for
	partnership or promotion.  As careers mature and direct monetary compensation becomes a
	larger part of the feasible reward package, the incentive value of allocating future career
	opportunities should diminish.  More generally, the model predicts that reviews should be
	most consequential precisely where organizations have limited scope to use large
	performance-contingent monetary payments.
	
	Finally, we examine three departures from the baseline.  To express the key mechanism in the
	simplest setup, the baseline lets the agent choose effort only once in each phase.  Our first
	extension allows him to choose effort at every instant; the four-rung ladder continues to apply.
	Next, replacing Gaussian output with a general statistical experiment shows that the
	ladder does not depend on normality.  Under suitable likelihood-ratio conditions, the same
	four exposed actions emerge, while the timing results depend on how rapidly useful evidence
	accumulates over short horizons.  Finally, we show that random review timing can be valuable because it relaxes the sharp tradeoff between collecting more pre-review information and preserving post-review incentive capital.
Together these extensions indicate that the four-rung incentive
 ladder is more robust than the baseline analysis alone suggests.
	
	\paragraph{Related Literature}
	
	This paper contributes to the literature on interim performance evaluation and dynamic incentive provision by asking three questions: Why review? How should review information be used? And when to review?
	
	First, we identify a new rationale for interim review that arises from limited
	incentive provision. As in \citet{lizzeri2002}, we consider a setting in which
	past performance does not change the primitive continuation agency problem. With
	unrestricted monetary rewards, the principal can bundle incentives across the
	two phases into a single terminal contract, making review unnecessary. However, when
	payments are bounded, as in the static setting of \citet{jewitt2008}, review allows
	the principal to condition otherwise unavoidable continuation rents and the
	remaining workload on earlier performance, thereby creating history dependence
	and expanding incentive capacity. In contrast, much of the literature considers environments in which earlier actions or outcomes affect the difficulty or return of subsequent work. In such settings, an important role of interim review is to reveal the severity of the continuation agency problem and allow the continuation contract to respond accordingly \citep[see, e.g.,][]{lukas2010,ChenChiu2013,elyszydlowski2020}.
	A separate rationale is efficient sorting:
	\citet{Ray2007} uses interim evaluation to identify and terminate projects with
	poor continuation prospects. Relatedly, \citet{liu2026} studies resource allocation
	when the agent privately knows a project's fixed quality and always prefers a larger
	allocation. A public midterm signal correlated with project quality allows the principal
	to discipline the agent's initial report, effectively ``levering'' information from the
	review into the first-period allocation. Our setting has neither private information about
	project quality to elicit nor heterogeneous continuation prospects to uncover. Instead,
	review provides information about the agent's hidden productive effort, and its value comes from
	using continuation decisions to provide incentives when monetary rewards are limited.

	Second, we show that interim performance can be used through a simple,
	cutoff-based contract. Depending on interim performance, the agent is
	terminated or retained under a
	continuation contract that pays only following sufficiently good final
	performance. Termination can therefore occur at both tails: a bad stop punishes
	poor performance, while a good stop rewards strong performance by combining
	the maximal monetary payment with relief from future effort. Whereas
	\citet{ely2025} jointly design dynamic incentives and the disclosure of
	performance feedback, our cutoff structure concerns how revealed interim
	performance is translated into termination and continuation incentives.
	
	Good-performance termination also arises in other dynamic moral hazard settings
	(see, e.g., \citet{spearwang2005} and \citet{sannikov2008}), but for a different
	reason. In those models, favorable performance raises the promised value of a
	risk-averse agent until maintaining incentives becomes too costly, so termination
	occurs when continuation value reaches a sufficiently high level. Our good stop is
	instead a direct response to the payment cap. Once the maximal monetary prize has
	been reached, reducing the agent's remaining workload provides the additional
	reward that money cannot.
	
	Third, we identify a new tradeoff governing the timing of review. Delaying review
	has a familiar benefit: more accumulated performance makes the review signal more
	informative. But delay has a novel cost in our setting: it shortens the remaining
	production horizon and thereby reduces the future output that can be used to create
	continuation rents. Review timing therefore determines not only the quality of
	information available to the principal, but also the amount of incentive capacity
	remaining after the review.
	
	This question connects our paper to a recent literature on the design and timing
	of dynamic monitoring. \citet{varas2020} show that optimal monitoring can take the
	form of deterministic periodic reviews when information acquisition is important,
	but random inspections when monitoring primarily serves to provide incentives.
	Similarly, \citet{ballknoepfle2026} study whether inspections should be predictable
	or random in a long-term moral-hazard relationship, with the answer depending on
	whether effort primarily promotes breakthroughs or prevents breakdowns.
	\citet{wong2026} goes further by allowing the principal to design the monitoring
	technology itself, obtaining a nonstationary sequence of pass/fail tests whose
	informativeness and consequences evolve over the relationship. Relatedly,
	\citet{DaiWangYang2026} jointly design monitoring and compensation, with the
	principal dynamically reallocating limited monitoring capacity between evidence
	confirming effort and evidence revealing shirking. These papers endogenize the
	pattern, composition, or technology of monitoring, whereas we take a single
	interim review as given and ask where it should be placed within a finite
	production horizon.

	Closest to our timing question, \citet{rodivilov2022} studies when to monitor an
	agent engaged in experimentation and shows that monitoring has both a static
	effect on contemporaneous moral hazard and a dynamic effect on the rents that
	must be provided in earlier periods. Our mechanism is different. In our model,
	the benefit of delay comes from the increasing informativeness of accumulated
	performance, while its cost comes from exhausting the very continuation
	opportunities through which review provides incentives.
	
	This distinction also separates our mechanism from \citet{hoffmann2021} and
	\citet{georgiadis2020}, who study the timing of performance measurement following
	a one-time hidden action rather than an interim review during ongoing production.
	In \citet{hoffmann2021}, delaying payout improves the information on which
	compensation is based but entails a direct cost of deferral; in
	\citet{georgiadis2020}, continued monitoring improves precision but incurs a
	direct monitoring cost. In our model, by contrast, the cost of delay arises from
	the contracting technology itself: waiting for a more informative review leaves
	less future work with which to reward favorable performance and provide
	incentives.

	\paragraph{Organization}
	The analysis proceeds as follows.  Section~\ref{sec:single-phase} derives the
	no-review benchmark and constructs the post-review continuation frontier.
	Section~\ref{sec:exogenous-review} holds the review date fixed, formulates pre-review incentives, and characterizes the ordered menu.
	Section~\ref{sec:endogenous-review} endogenizes the date, deriving the review value, its behavior near the boundary dates, and the tight-cap result.  Section~\ref{sec:extensions} studies pointwise effort, non-Gaussian monitoring, and random review.  Section~\ref{sec:conclusion} concludes with two directions requiring a richer institutional setting.
	
	\section{Environment} \label{sec:environment}
	
	A risk-neutral principal (she) contracts with a risk-neutral agent (he) over a finite horizon $[0,T]$. Output $X_t$ evolves as
	\begin{equation*}
		dX_t=e_t\,dt+dW_t, \qquad X_0=0,
	\end{equation*}
	where $W$ is a standard Brownian motion and $e_t\in\{0,1\}$ is agent effort.  We write $\Phi(\cdot)$ for the standard normal distribution and $\phi(\cdot)$ for its density.
	
	The contract fixes at date $0$ a single review date $\tau\in(0,T]$.\footnote{The limiting case $\tau=0$ is used only as a boundary limit in Section~\ref{sec:endogenous-review}.}
	The review date is ultimately a choice variable, and we solve for it in two steps: Section~\ref{sec:exogenous-review} takes $\tau$ as \emph{exogenous}, and Section~\ref{sec:endogenous-review} then \emph{endogenizes} it.
	
	\begin{assumption}[Observability of the output]
		\label{as:observability}
		Cumulative output is publicly observed at $\tau$ and at $T$. No one, including the agent, observes $X_t$ at any other time.
	\end{assumption}
	
	The review splits the horizon into two phases, $P_1=(0,\tau]$ and $P_2=(\tau,T]$, of lengths $\tau$ and $T-\tau$; we call $X_\tau$ the \emph{review score}.  The case $\tau=T$	is the \emph{no-review benchmark}: the second phase is empty and $X_T$ is the only signal.
	
	\begin{assumption}[Phase-level private effort]
		\label{as:phase-level}
		Effort is constant within each phase.  The agent privately chooses $a_1\in\{0,1\}$ at date $0$ and, having observed $X_\tau$, privately chooses $a_2\in\{0,1\}$ at date $\tau$; effort is $e_t=a_1$ on $P_1$ and $e_t=a_2$ on $P_2$.  A strategy is a pair $(a_1,{S})$ with ${S}:\mathbb{R}\to\{0,1\}$ giving $a_2={S}(X_\tau)$.\footnotemark
	\end{assumption}
	\footnotetext{Assumption~\ref{as:phase-level} is relaxed in Section~\ref{sec:pointwise-effort}, where the agent chooses $e_t\in\{0,1\}$ at	each instant, so that a phase of length $h$ delivers any total effort $\eta=\int e_t\,dt\in[0,h]$ rather than only $\{0,h\}$.}

	The agent commits at the start of a phase to work throughout it or shirk throughout it.  Consequently the two phase increments are independent given $(a_1,a_2)$, with
	\begin{equation*}
		X_\tau\sim N(a_1\tau,\tau), \qquad X_T-X_\tau\sim N\bigl(a_2(T-\tau),\,T-\tau\bigr).
	\end{equation*}
	
	Working through a phase of length $h$ costs the agent $ch$, so total effort cost is $c\bigl[a_1\tau+a_2 (T-\tau) \bigr]$ with $c \in (0,1)$.
	The same flow cost applies for the entire horizon, and hence in both phases.  After observing $X_\tau$ and $X_T$ the principal makes a measurable payment $w(X_\tau,X_T)$ to the agent.
	
	\begin{assumption}[Limited liability and finite compensation]
		The payment function satisfies $w:\mathbb{R}^2\to[0,M]$ with $M\in(0,\infty)$.
	\end{assumption}
	The lower bound of $0$ is the agent's limited liability, and the upper bound $M$ is a finite cap: the maximum payment the principal can make.\footnote{All results generalize to a setting with finite lower and upper bounds on payments $-\infty < \underline{w} < \overline{w} < +\infty$, where $M \equiv \overline{w}-\underline{w}$.  We set $\underline{w}=0$ as is standard in settings of limited liability.}
	
	The principal receives ex post payoff $X_T-w$, and the agent receives $w-c\bigl[a_1\tau+a_2(T-\tau)\bigr]$. At $t=0$ the principal chooses a contract that consists of a review date $\tau$ and a payment function $w$ to maximize her expected payoff.  Under the no-review benchmark $\tau=T$ a contract is simply $w(X_T) \in [0,M]$.  The principal has all the bargaining power and full commitment, and the agent's outside option is zero.

	\section{Contracting over a single phase}
	\label{sec:single-phase}
	
	We begin the analysis with the single-phase contract.  Fix a phase of length $h>0$ at the end of which output is observed and payment is made, with no further review.  Two cases of the model have this form, $h=T$ (no-review) and $h=T-\tau$ (post-review); in the latter, translation invariance of the technology lets a continuation contract be written as a function of the second phase output $Y=X_T-X_\tau$.  Solving the phase problem once therefore delivers both cases.
	
	Under work $Y\sim N(h,h)$ and under shirking $Y\sim N(0,h)$, with densities $f_j^h(y)=h^{-1/2}\phi\bigl((y-jh)/\sqrt h\bigr)$ for $j\in\{0,1\}$; working costs the agent $ch$.  A contract is a measurable $w:\mathbb R\to[0,M]$; it implements work if
	\begin{equation*}
		\E_1[w]-\E_0[w]\ \geq\ ch,	\tag{IC$_h$}	\label{eq:IC-h}
	\end{equation*}
	where $\E_j[w]=\int wf_j^h$.  Let $\mathcal F_h$ denote the set of contracts satisfying	\eqref{eq:IC-h}.  Throughout the section $c$ and $h$ are fixed, and we suppress the dependence of derived objects on $(h,c,M)$ except where an argument varies.

	\subsection{Threshold contracts}
	\label{sec:threshold-contracts}
	
	First, we show that it is without loss to focus on \emph{threshold contracts}: pay the cap $M$ when output clears a bar $b$ and nothing below it.  The reason is that the likelihood ratio $\ell^h(y)=f_0^h(y)/f_1^h(y)=\exp(h/2-y)$ is strictly decreasing, so a contract that shifts payment toward high output --- rewarding evidence of work and withholding payment after evidence of shirking --- provides incentives at the lowest expected cost.
	
	\begin{lemma}
		\label{lem:threshold-reduction}
		Let $w:\mathbb R\to[0,M]$ be any contract and let $\tilde w=\ind_{\{y\geq b\}} M$ be the threshold contract with the same expected cost under work, $\E_1[\tilde w]=\E_1[w]$. Then $\E_0[\tilde w]\leq\E_0[w]$, with equality only if $w=\tilde w$ almost everywhere.	Hence $w\in\mathcal F_h$ implies $\tilde w\in\mathcal F_h$.
	\end{lemma}
	
	For any contract that implements work, there is a threshold contract that also implements work and leaves the principal the same expected payoff. Thus, nothing is lost by restricting attention to threshold contracts.  Moreover, because every threshold contract pays either $0$ or $M$, it can equivalently be interpreted as awarding a fixed prize $M$ (e.g., a promotion) whenever performance exceeds a sufficiently high threshold.

	Now, index the bar $b$ by the \emph{normalized} value $z=(h-b)/\sqrt h$, so that a larger $z$ is an easier bar, cleared with probability $\Phi(z)$ under work.  For $w=\ind_{\{y\geq b\}} M$,
	\begin{equation}
		\E_1[w]=M\Phi(z),	\qquad	\E_0[w]=M\Phi\bigl(z-\sqrt h\bigr),	\qquad	\E_1[w]-\E_0[w]=M\,\Delta(h,z),
		\label{eq:threshold-moments}
	\end{equation}
	where $\Delta(h,z)=\Phi(z)-\Phi(z-\sqrt h)$ is the \emph{incentive index}: the increase in the probability of clearing the bar when the agent works rather than shirks.  By	\eqref{eq:threshold-moments} a threshold contract implements work exactly when
	\begin{equation}
		\Delta(h,z)\ \geq\ \frac{ch}{M}.
		\label{eq:bar-condition}
	\end{equation}
	
	Figure~\ref{fig:incentive-index} illustrates the admissible bars.  The index is single-peaked	at $z=\sqrt h/2$, where it attains its maximum $\overline\Delta(h)=2\Phi\bigl(\sqrt h/2\bigr)-1$, symmetric about that peak, and vanishes as $z\to\pm\infty$.  Work is therefore implementable if and only if $M\ \geq\ M_I(h,c)\equiv\frac{ch}{\overline\Delta(h)}$, since below that cap \eqref{eq:bar-condition} fails at every $z$.  When it holds, the admissible bars correspond to the closed interval $[z_D,z_E]$ with
	\begin{equation}
		z_D=\min\left\{z:\Delta(h,z)\geq\frac{ch}{M}\right\},
		\qquad
		z_E=\max\left\{z:\Delta(h,z)\geq\frac{ch}{M}\right\},
		\label{eq:zD-zE}
	\end{equation}
	the two crossings of the horizontal line at height $ch/M$, which merge at $\sqrt h/2$ when $M=M_I(h,c)$.  The constraint binds at $z_D$ and $z_E$ and is slack strictly between them.  The lower crossing $z_D$ corresponds to a high-output bar and sets a	\emph{difficult} standard $b_D$; the upper crossing $z_E$ corresponds to a low bar and sets an \emph{easy} one $b_E$.
	
	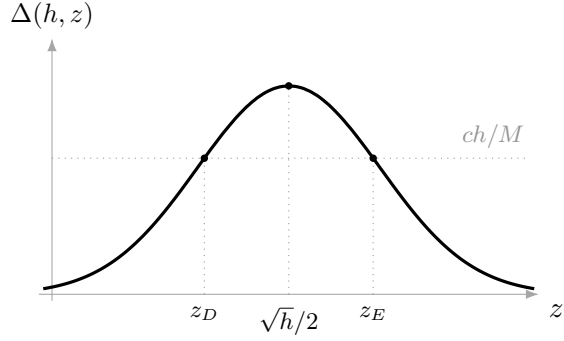
\begin{figure}
		\centering
		\begin{tikzpicture}[x=1.16cm,y=7.2cm,>=Latex]
			\def\hh{1}          
			\def\lev{0.25}      
			\def\zD{-0.4628}    
			\def\zE{1.4628}     
			\def\pk{0.5}        
			\def\Dbar{0.38292}  
			\draw[->,gray!70] (-2.35,0) -- (3.35,0) node[below right,black] {$z$};
			\draw[->,gray!70] (-2.2,-0.012) -- (-2.2,0.470) node[above,black]
			{\small$\Delta(h,z)$};
			\draw[dotted,gray!85] (-2.2,\lev) -- (3.2,\lev);
			\draw[dotted,gray!85] (\pk,0) -- (\pk,\Dbar);
			\draw[dotted,gray!85] (\zD,0) -- (\zD,\lev);
			\draw[dotted,gray!85] (\zE,0) -- (\zE,\lev);
			\draw[very thick,domain=-2.3:3.3,samples=180,smooth]
			plot (\x,{Phinorm(\x)-Phinorm(\x-sqrt(\hh))});
			\fill (\zD,\lev) circle (1.4pt);
			\fill (\zE,\lev) circle (1.4pt);
			\fill (\pk,\Dbar) circle (1.4pt);
			\node[font=\footnotesize] at (\zD,-0.004) [below] {$z_D$};
			\node[font=\footnotesize] at (\zE,-0.004) [below] {$z_E$};
			\node[font=\footnotesize] at (\pk,-0.004) [below] {${\sqrt h}/{2}$};
			\node[font=\footnotesize,gray!85] at (3.3,\lev) [above left] {$ch/M$};
		\end{tikzpicture}
		\caption{The incentive index $\Delta(h,\cdot)$}
		\label{fig:incentive-index}
	\end{figure}
	
	Implementability does not guarantee that the principal wants the agent to work: the rent she must concede to do so may exceed what the phase produces.  A second threshold therefore governs profitability.
	
	\begin{lemma}
		\label{lem:Mbar}
		There is a threshold $\underline M(h,c)\geq M_I(h,c)$ such that work is implementable and the principal's payoff from implementing it is nonnegative if and only if $M\geq\underline M(h,c)$.  Additionally, there exists $\bar c(h)\in(0,1)$ such that $\underline M(h,c)=M_I(h,c)$ for all $c\leq\bar c(h)$; otherwise $\underline M(h,c)>M_I(h,c)$.
	\end{lemma}
	
	Appendix~\ref{app:single-phase} gives $\bar c(h)$ and $\underline M(h,c)$ in closed form.	
	We impose the following assumption to ensure that contracting without a review is strictly more profitable than shutting down.
	
	\begin{assumption}[Implementable and profitable full effort]
		\label{as:maintained}
		The payment cap satisfies $M>\underline M(T,c)$.
	\end{assumption}
	
	\subsection{The no-review benchmark}
	\label{sec:no-review}
	
	Setting $h=T$ gives the benchmark against which a review must be evaluated: the principal writes a single contract on terminal output and never intervenes.  By Lemma~\ref{lem:threshold-reduction} she may take it to be a threshold contract, and among the admissible bars she chooses the cheapest, which is the difficult one.
	
	\begin{proposition}
		\label{prop:no-review}
		Under Assumption~\ref{as:maintained}, the principal's optimal no-review payoff is
		\begin{equation*}
			V^{\NR}(T;c,M)=T-M\Phi\bigl(z_D(T,c,M)\bigr)\ >\ 0,
		\end{equation*}
		attained by the difficult bar.  Moreover $z_D(T,c,M)\to-\infty$ as $M\to\infty$, and
		\begin{equation*}
			\lim_{M\to\infty}V^{\NR}(T;c,M)=(1-c)T,
		\end{equation*}
		so the first-best surplus accrues to the principal as the payment cap is relaxed.
	\end{proposition}
	
	The limit is what makes a review worth studying.  Since $z_D=(T-b_D)/\sqrt T$, $z_D$ falling to $-\infty$ corresponds to a bar $b_D$ rising to $+\infty$ in output: as $M$ grows the principal promises an ever larger prize on an ever less likely outcome.  Because the Gaussian likelihood ratio is unbounded, that trade is favorable for the principal.  The agent's rent is what he could secure by shirking, $\E_0[w]=M\Phi\bigl(z_D-\sqrt T\bigr)$ by \eqref{eq:threshold-moments}, and as the bar rises this vanishes even though $M$ diverges: the reward event becomes overwhelming evidence of work.  A principal who can promise arbitrarily large payments therefore captures the whole surplus without ever consulting interim performance, and has nothing to gain from a review.
	
	When $M$ is finite the agent retains a strictly positive rent, and the question is whether the principal can do better by making the continuation depend on an interim signal.  The benchmark is generous to her in one further respect: under Assumptions~\ref{as:observability} and~\ref{as:phase-level} the agent commits at date $0$ to work through the entire horizon and never observes his own progress, so he cannot condition effort on how the project is going.
	
	\subsection{The review action set}
	\label{sec:extreme-points}
	
	Now take $h=T-\tau$ to be a continuation phase.  After a review score $x$ the principal's choice is a continuation contract $w(x,\cdot):\mathbb R\to[0,M]$ on the second-phase output $Y$, and the agent's best response to it determines whether second-phase work occurs.  For everything that follows, such a contract is summarized by two numbers: the expected continuation utility $r$ it leaves the agent net of effort cost, and the \emph{total} expected continuation surplus $s$ it generates---the principal's own expected continuation payoff is then $u=s-r$.  We call the pair $(r,s)$ an \emph{action} and describe the review by the set of actions available at it.  (We drop the qualifier \textit{expected} from payoffs and surplus below whenever doing so creates no confusion.)
	
	When effort is implemented, total surplus is $(1-c)h$, so the only remaining question is which rents the agent can be promised.
	
	\begin{lemma}
		\label{lem:attainable-rents}
		The rents attainable while implementing work are exactly $[r_D,r_E]$, where
		\begin{equation*}
			r_j=M\Phi(z_j)-ch=M\Phi\bigl(z_j-\sqrt h\bigr),
			\quad j\in\{D,E\},
			\qquad
			0<r_D<r_E<M-ch,
		\end{equation*}
		and $r_D+r_E=M-ch$.
	\end{lemma}

	Here $r_D$ and $r_E$ are the rents left by the difficult and the easy standard respectively, and the lemma says that every attainable rent lies between them.  The principal's payoff from the difficult standard, $V^{\NR}=(1-c)h-r_D$, is the no-review benchmark value when $h=T$.
	The set of rent--surplus pairs available under work is therefore $\mathcal A_1=[r_D,r_E]\times\{(1-c)h\}$.
	
	If work is not implemented, surplus is zero, and since feasible payments range over
	$[0,M]$ the corresponding set is $\mathcal A_0=[0,M]\times\{0\}$.  The feasible action
	set is therefore $\mathcal A=\mathcal A_0\cup\mathcal A_1$.
	A review menu assigns one such action to every review score.  Thus, after each score, the
	principal chooses both a continuation profile---no work from $\mathcal A_0$ or work from
	$\mathcal A_1$---and the rent delivered within that profile.  The next section determines
	how these choices should vary with the review score in order to implement first-phase effort.

	\section{The exogenous review date}
	\label{sec:exogenous-review}
	
	Having established what can be offered after the review, we now move backward to the	incentives that precede it.  Throughout this section the review date $\tau\in(0,T)$ is fixed exogenously, and we continue to write $h=T-\tau$.
	
	\subsection{The optimal incentive ladder}
	\label{sec:fixed-review-optimal-menu}
	
	The review performance is $X_\tau\sim N(a\tau,\tau)$ if the agent chose action $e_t=a\in\{0,1\}$ throughout the first phase. Let $g_1^\tau$ and $g_0^\tau$ be the corresponding densities.  Their likelihood ratio	${g_1^\tau(x)}/{g_0^\tau(x)}=\exp(x-\tau/2)$ is strictly increasing.
	
	Let $a(x)=\bigl(r(x),s(x)\bigr)\in\mathcal A$ be the action assigned after review score $x$, delivering the agent rent $r(x)$ and the principal payoff $s(x)-r(x)$.  Working through the first phase yields the agent $\int r\,g_1^\tau-c\tau$ and shirking yields $\int r\,g_0^\tau$, so pre-review obedience requires
	\begin{equation*}
		\int_{\mathbb R} r(x)\left[g_1^\tau(x)-g_0^\tau(x)\right]dx \geq c\tau ,
		\tag{OB$_\tau$}
		\label{eq:pre-review-obedience}
	\end{equation*}
	which also guarantees participation, since $r\geq0$ makes the utility from shirking nonnegative. Conditional on inducing first-phase work, expected output accumulated by the exogenous review is $\tau$, so the principal's problem at $t=0$ is
	\begin{equation*}
		\begin{aligned}
			V^R(\tau)
			=
			\sup_{a(\cdot)}\quad
			&
			\tau + \int_{\mathbb R}\bigl[s(x)-r(x)\bigr]g_1^\tau(x)\,dx
			\\[3pt]
			\text{subject to}\quad
			&
			\eqref{eq:pre-review-obedience},\quad a(\cdot)\ \text{measurable},
			\quad a(x)\in\mathcal A\quad\text{for a.e. }x .
		\end{aligned}
		\tag{P$_\tau$}
		\label{eq:fixed-review-primal}
	\end{equation*}

	Now we characterize how the continuation profile varies with the midterm review performance.
	
	\begin{proposition}
		\label{prop:four-region-review-menu}
		Fix $\tau\in(0,T)$ and suppose first-phase work is implemented.  An optimal menu takes
		the ordered form
		\begin{equation*}
			a^*(x)
			=
			\begin{cases}
				A_B, & x<k_0^*,\\
				A_D, & k_0^*\leq x<k_1^*,\\
				A_E, & k_1^*\leq x<k_2^*,\\
				A_G, & x\geq k_2^*,
			\end{cases}
		\end{equation*}
		where $A_B=(0,0)$, $A_D=\bigl(r_D,(1-c)h\bigr)$,
		$A_E=\bigl(r_E,(1-c)h\bigr), A_G=(M,0)$ are four endpoints of $\mathcal A$, and
		$
		k_0^*\leq k_1^*\leq k_2^*.
		$
		Moreover, whenever both $k_j^*$ and $k_{j+1}^*$ are finite, the corresponding
		inequality is strict: $k_j^*<k_{j+1}^*$ for $j=0,1$.
	\end{proposition}
	
	Proposition~\ref{prop:four-region-review-menu} is our first main result. It shows that the optimal contract partitions the midterm-performance space into at most four ordered regions, assigning continuation profiles with progressively higher rents as performance improves.  After the lowest performance realizations, $A_B$ supplies the
	strongest punishment: the relationship ends with no payment and no second-phase production.
	At intermediate performance realizations the principal preserves the continuation surplus $(1-c)h$.  She first
	uses the difficult standard $A_D$, which delivers the smallest rent consistent with
	continuation effort, and then the easy standard $A_E$, which raises the agent's rent from
	$r_D$ to $r_E$ without sacrificing production.  Only after the highest performance realizations can the
	principal move to $A_G$, paying the maximal rent $M$ but giving up the continuation surplus
	altogether.  This progression linking better review outcomes with higher discrete rent levels is what we refer to as the \textit{incentive ladder.}
	
	The optimal ladder therefore uses continuation contracts as intermediate incentive
	instruments.  Before resorting to the maximal monetary reward, it increases the agent's
	continuation rent while keeping second-phase effort in place.  The weak ordering of the
	cutoffs allows some of the upper profiles (rungs of the ladder) to be absent.
	
	\subsection{Deriving the optimal ladder}
	\label{sec:deriving-optimal-menu}
	
	To understand why the optimal menu has this form, return to
	Problem~\eqref{eq:fixed-review-primal}.  This is an optimization over the entire menu $a(\cdot)$,
	and its only constraint linking choices across review scores is pre-review obedience. Despite its infinite-dimensional appearance, the problem admits an elementary geometric solution. The reduction also reveals the economic forces governing the choice among stopping and continuation arrangements.
	
	Let $\lambda^*\geq0$ be a
	Lagrange multiplier on \eqref{eq:pre-review-obedience} at the optimum and define the
	\emph{net shadow value of continuation rent}
	\begin{equation}
		Q_\lambda(x) \equiv	\lambda\left(1-\frac{g_0^\tau(x)}{g_1^\tau(x)}\right)-1
		=	\lambda\left[1-\exp\left(\frac{\tau}{2}-x\right)\right]-1 .
		\label{eq:Q-lambda-definition}
	\end{equation}
	Adjoining \eqref{eq:pre-review-obedience} and factoring out $g_1^\tau(x)$ gives the	Lagrangian
	\begin{equation}
		\mathcal L(a;\lambda,\tau)
		=
		\tau-\lambda c\tau
		+\int_{\mathbb R}
		g_1^\tau(x)\bigl\{s(x)+Q_\lambda(x)\,r(x)\bigr\}dx .
		\label{eq:frontier-substituted-lagrangian}
	\end{equation}
	For each $\lambda$, the Lagrangian is maximized over measurable menus satisfying
	$a(x)\in\mathcal A$ for almost every $x$; only pre-review obedience has been dualized.
	The optimal menu maximizes \eqref{eq:frontier-substituted-lagrangian} at
	$\lambda=\lambda^*$, together with primal feasibility and complementary
	slackness. This reformulation separates the two agency problems.  Post-review incentive compatibility
	and payment feasibility are encoded in the continuation-action set $\mathcal A$, while the
	multiplier incorporates pre-review obedience into the objective.  Because the work-to-shirking
	likelihood ratio, and hence $Q_\lambda(x)$, increases with observed performance, stronger
	performance places greater shadow value on the agent's continuation rent and therefore tilts
	the principal's choice toward higher-rent continuation profiles, subject to the
	surplus tradeoffs encoded by $\mathcal A$.\footnote{This is the familiar likelihood-ratio
		formulation of a static moral-hazard problem \citep{Holmstrom1979}, except that the reward is a
		continuation profile rather than a wage.}
	
	For a fixed $\lambda>0$, the Lagrangian objective is additively separable	across review-score realizations. Choosing a measurable menu therefore reduces to solving, for almost every review score $x$, the linear assignment problem $s+Q_\lambda(x)r$.
	The coefficient $Q_\lambda(x)$ captures the weight placed on the agent's	post-review continuation rent following score $x$. By	Lemma~\ref{lem:supporting-multiplier}\textup{(iii)} in	Appendix~\ref{app:review-menu}, the optimal supporting multiplier is positive outside a knife-edge case, so this characterization applies at the optimum.
	
	Although the original action choice at a given review score is from $\mathcal A=\mathcal A_0\cup\mathcal A_1$, it is convenient to solve the linear program
	\begin{equation*}
		a_\lambda(x)\in
		\mathop{\operatorname{arg\max}}\limits_{(r,s)\in\operatorname{co}(\mathcal A)}
		\bigl\{s+Q_\lambda(x)r\bigr\}.
	\end{equation*}
	Because the objective is linear in $(r,s)$, convexification does not change the maximal value, and an optimizer can always be selected from the original set $\mathcal A$.
	
	A linear objective over the compact polygon $\operatorname{co}(\mathcal A)$ can always be maximized at an extreme point.  Its four extreme points are $A_B$, $A_D$, $A_E$ and $A_G$,	all of which are defined in Proposition \ref{prop:four-region-review-menu} and belong to $\mathcal A$.  Hence, for each $x$, a solution to the above maximization problem can always be chosen from the original feasible set; public randomization cannot improve the principal's payoff.  At a breakpoint, an entire edge of the convex hull may also be optimal, but one of its feasible endpoints attains the same value.

	Figure~\ref{fig:continuation-trapezoid} depicts this review-score assignment problem.  The two bold segments form the original feasible set $\mathcal A$, while the shaded trapezoid is	$\operatorname{co}(\mathcal A)$.  The dashed lines are level sets of $s+qr$.  Maximizing
	the linear objective amounts to moving such a line in the direction $(q,1)$ until it supports
	the convex hull.

	\begin{figure}
		\centering
		\begin{tikzpicture}[x=4.4cm,y=4.0cm,>=Latex]
			\draw[->,gray!70] (-0.13,0) -- (1.66,0) node[right,black] {$r$};
			\draw[->,gray!70] (0,-0.17) -- (0,1.17) node[above,black] {$s$};
			\fill[gray!10] (0,0) -- (1.2,0) -- (0.814,0.7) -- (0.086,0.7) -- cycle;
			\draw[gray!45] (0,0) -- (0.086,0.7);
			\draw[gray!45] (1.2,0) -- (0.814,0.7);
			\draw[dashed,gray!65] (0,0) -- (0.0742,0.89);
			\node[gray!80,font=\footnotesize] at (0.1,0.89) [above] {$q<q_0$};
			\draw[dashed,gray!65] (0,0.614) -- (0.40,1.014);
			\node[gray!80,font=\footnotesize] at (0.17,1.08) [right] {$q_0<q<0$};
			\draw[dashed,gray!65] (0.52,0.934) -- (1.27,0.334);
			\node[gray!80,font=\footnotesize] at (1.27,0.334) [right] {$0<q<q_2$};
			\draw[dashed,gray!65] (0.92,0.84) -- (1.30,-0.30);
			\node[gray!80,font=\footnotesize] at (0.92,0.84) [above] {$q>q_2$};
			\draw[->,gray!85,semithick] (0,0) -- (-0.115,0.010);
			\draw[->,gray!85,semithick] (0.086,0.7) -- (0.004,0.782);
			\draw[->,gray!85,semithick] (0.814,0.7) -- (0.887,0.791);
			\draw[->,gray!85,semithick] (1.2,0) -- (1.310,0.037);
			\draw[very thick,gray!60] (0,0) -- (1.2,0);
			\node[gray!60,font=\small] at (0.60,-0.05) [below] {$\mathcal A_0$};
			\draw[very thick] (0.086,0.7) -- (0.814,0.7);
			\node[font=\small] at (0.45,0.64) [below] {$\mathcal A_1$};
			\draw[dotted,gray!85] (0,0.7) -- (0.086,0.7);
			\node[gray!80,font=\footnotesize] at (0,0.7) [left] {$(1-c)h$};
			\fill (0,0) circle (1.5pt);        \node[font=\small] at (0.004,-0.012) [below right] {$A_B$};
			\fill (1.2,0) circle (1.5pt);      \node[font=\small] at (1.2,-0.012) [below] {$A_G$};
			\fill (0.086,0.7) circle (1.5pt);  \node[font=\small] at (0.098,0.706) [above right] {$A_D$};
			\fill (0.814,0.7) circle (1.5pt);  \node[font=\small] at (0.802,0.706) [above left] {$A_E$};
			\node[gray!80,font=\footnotesize] at (0.094,0.694) [below right] {$r_D$};
			\node[gray!80,font=\footnotesize] at (0.806,0.694) [below left] {$r_E$};
		\end{tikzpicture}
		\caption{The review action set and the supporting lines of the assignment objective }
		\label{fig:continuation-trapezoid}
	\end{figure}
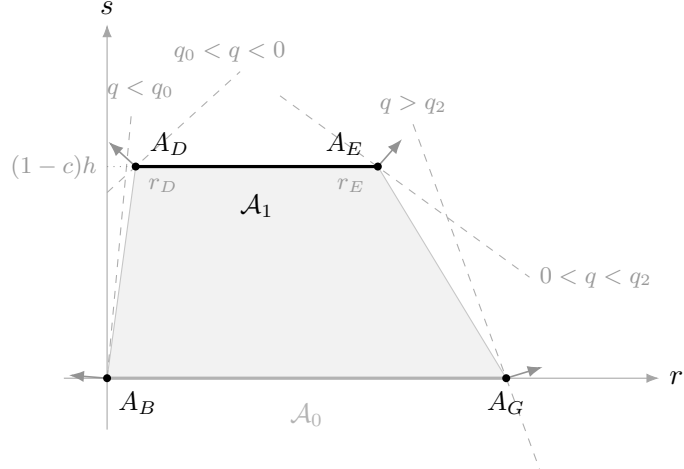

	To solve the linear program, fix the scalar $q=Q_\lambda(x)$.  The four vertices deliver
	\begin{equation*}
		v_B(q)=0,		\qquad	v_D(q)=(1-c)h+q\,r_D,	\qquad	v_E(q)=(1-c)h+q\,r_E,
		\qquad	v_G(q)=q\,M .
	\end{equation*}
	Their slopes are ordered by $0<r_D<r_E<M$.  The upper envelope therefore moves toward
	higher-rent actions as $q$ rises.  The adjacent crossings are
	\begin{equation}
		q_0=-\frac{(1-c)h}{r_D}
		\;<\;
		q_1=0
		\;<\;
		q_2=\frac{(1-c)h}{M-r_E}.
		\label{eq:three-breakpoints}
	\end{equation}
	Each breakpoint is an edge of Figure~\ref{fig:continuation-trapezoid}.  A level set of
	$s+qr$ has slope $-q$ in $(r,s)$ space, so at $q=q_j$ the supporting line is parallel to
	the edge joining the two actions that tie there: $-q_0=(1-c)h/r_D$ is the slope of
	$A_BA_D$, $-q_1=0$ that of the flat top edge $A_DA_E$, along which rent rises at no cost
	in surplus, and $-q_2=-(1-c)h/(M-r_E)$ that of $A_EA_G$.  Since the upper boundary of
	$\operatorname{co}(\mathcal A)$ is concave, these three slopes fall from left to right,
	which is the ordering in \eqref{eq:three-breakpoints}.
	It follows that the maximizing action is $A_B$ for $q<q_0$, $A_D$ for
	$q_0<q<q_1$, $A_E$ for $q_1<q<q_2$, and $A_G$ for $q>q_2$, with adjacent
	actions tied at each breakpoint. Since $\lambda^*>0$,
	$Q_{\lambda^*}(x)$ is strictly increasing in $x$. As performance rises,
	$q=Q_{\lambda^*}(x)$ therefore passes through these intervals in the same
	order. This converts the review-score ranking into the ordered performance
	regions stated in Proposition~\ref{prop:four-region-review-menu}.

	\subsection{On the form of the optimal ladder}
	\label{sec:form-optimal-menu}

	This subsection 	examines the economic forces that shape the specific form of the optimal menu---when termination is
	used following particularly poor or particularly strong performance, and when continuation is
	differentiated through more or less demanding standards.

	\begin{corollary}
		\label{cor:menu-form}
		Suppose $(1-c)h\neq r_D$.\footnote{The knife-edge case			$(1-c)h=r_D$, in which a zero supporting multiplier may arise, is discussed in Appendix~\ref{app:review-menu}.}
		Whenever first-phase work is implemented, $k_0^*$ is finite and the bad-stop and	difficult-standard regions are nonempty.  Moreover:
		\begin{enumerate}[(i)]
			\item If $1+q_0<\lambda^*\leq1$, then	$k_1^*=k_2^*=+\infty$, and only the bad-stop and difficult-standard regions are	nonempty.
			\item If $1<\lambda^*\leq1+q_2$, then
			$k_0^*<k_1^*<+\infty$ and $k_2^*=+\infty$, so the bad-stop,
			difficult-standard, and easy-standard regions are nonempty.
			\item If $\lambda^*>1+q_2$, then
			$k_0^*<k_1^*<k_2^*<+\infty$, and all four regions are nonempty.
		\end{enumerate}
	\end{corollary}
	
	The multiplier $\lambda^*$ is the shadow value of the first-phase obedience constraint.  It
	measures the principal's marginal gain from relaxing the incentive requirement for
	first-phase work.  A large multiplier therefore means that this requirement is particularly
	costly to satisfy: first-phase incentives are scarce, and the principal is willing to use
	stronger rewards following favorable performance.
	
	This interpretation is reflected directly in the menu.  The index $Q_{\lambda^*}(x)$ sweeps
	the interval $(-\infty,\lambda^*-1)$ as performance $x$ ranges over $\mathbb R$, and it
	crosses breakpoint $q_j$ if and only if $\lambda^*>1+q_j$.  Implementing first-phase work
	requires $\lambda^*>1+q_0$, so the bad stop and the difficult standard are always used.  Once
	$\lambda^*>1$, the incentive value of additional rent is high enough to justify the easy
	standard, which rewards stronger performance while preserving second-phase production.  Once
	$\lambda^*>1+q_2$, even the easy standard provides too little rent at the top: the principal
	then uses the good stop, paying the maximal reward despite sacrificing the continuation
	surplus.
	
	The range of $Q_{\lambda^*}$ also explains the asymmetry between the lower and upper regions.
	It is unbounded below, so sufficiently poor performance always receives the strongest
	punishment.  Its upper bound is $\lambda^*-1$, however, so the stronger rewards become active
	only when the shadow value of first-phase incentives is sufficiently large.  Whenever a
	breakpoint is reached, the corresponding cutoff is the unique solution to
	$Q_{\lambda^*}(k_j^*)=q_j$.
	
	An especially sharp implication emerges when the review occurs late. As $\tau$ approaches $T$, the remaining continuation horizon becomes short. The continuation rungs do not disappear, however.  Instead, the threshold for activating the highest reward converges to its lower bound, while the	shadow value of first-phase incentives remains strictly above that bound. Consequently, every rung of the incentive ladder is used sufficiently close to the terminal date.
	
	\begin{proposition}
		\label{prop:terminal-full-ladder}
		Suppose Assumption~\ref{as:maintained} holds. Then
		$k_2^*(\tau)<+\infty$ for every review date $\tau$ sufficiently
		close to $T$ from below.
	\end{proposition}
	
	For any finite admissible cap, a sufficiently late review uses all four continuation profiles.
	To see the force behind this result, recall that the good stop becomes active when $\lambda^*>1+q_2$.  As the continuation horizon $h=T-\tau$ vanishes, the surplus sacrificed by replacing the easy standard with the good stop shrinks faster than the additional rent thereby created.  Hence $q_2(h)\to0$.  At the same time, the fixed-review shadow value converges to the no-review shadow value, which is strictly greater than one under	Assumption~\ref{as:maintained}.  Thus the good-stop threshold is eventually	crossed, and all four regions are active.
	
	\section{The timing and value of review}
	\label{sec:endogenous-review}
	
	In this section, the principal chooses the review date endogenously from
	$\tau\in(0,T]$.  A date $\tau\in(0,T)$ represents a genuine interim review, whereas	$\tau=T$ represents no interim review.  A review at $\tau=0$ is equivalent to one at $\tau=T$, so we use $\tau=T$ to represent no review and do not consider $\tau=0$
	separately.  We first study how the value of a genuine review varies with its date and then
	determine when such a review dominates no review.
	
	\subsection{The optimal genuine review date}

	For the remainder of the timing analysis, fix a genuine review date
	$\tau\in(0,T)$ and let $h=T-\tau>0$.  Using the fixed-date characterization from
	Section~\ref{sec:exogenous-review}, define the
	continuation support function for a remaining horizon $h$ and shadow value $q$ by
	\begin{equation*}
		\mathcal V(q,h)
		\equiv
		\max\left\{0,\ (1-c)h+q\,r_D,\ (1-c)h+q\,r_E,\ qM\right\}.
	\end{equation*}
	Complementary slackness implies that the optimized value equals the Lagrangian at the
	constrained optimum.  Hence
	\begin{equation*}
		V^R(\tau)
		=
		\tau-\lambda^*c\tau
		+
		\int_{\mathbb R}g_1^\tau(x)
		\mathcal V\bigl(Q_{\lambda^*}(x),T-\tau\bigr)\,dx,
	\end{equation*}
	where $Q_{\lambda^*}(x)$ is the net shadow value of continuation rent defined in
	Equation~\eqref{eq:Q-lambda-definition}.
	
	\begin{lemma}
		\label{lem:review-boundary-values}
		Suppose Assumption~\ref{as:maintained} holds.  Then
		\begin{equation*}
			\lim_{\tau\downarrow0}V^R(\tau)		=	\lim_{\tau\uparrow T}V^R(\tau)	=	V^{\NR}(T;c,M).
		\end{equation*}
	\end{lemma}
	
	Both boundaries degenerate, for opposite reasons.  As $\tau\downarrow0$ the review score carries almost no information about first-phase effort, so there is nothing to condition on; as $\tau\uparrow T$ the continuation phase vanishes, so there is no continuation opportunity left to allocate.  Either way the principal is left with the terminal contract.  The no-review benchmark is therefore the right normalization across the whole interval, and any gain from review must be an interior phenomenon.
	
	\paragraph{The tradeoff among genuine review dates.} Away from knife-edge parameters, the constrained
	envelope theorem applies while holding the multiplier and the selected action fixed.  The
	derivative decomposes as
	\begin{align}
		\frac{dV^R(\tau)}{d\tau}
		={}&
		\underbrace{1-\lambda^*c}_{\text{first-phase production}}+\underbrace{\int_{\mathbb R}
			\left[
			\frac{\partial g_1^\tau(x)}{\partial\tau}
			\mathcal V\bigl(Q_{\lambda^*}(x),h\bigr)
			+g_1^\tau(x)
			\frac{\partial\mathcal V}{\partial q}
			\bigl(Q_{\lambda^*}(x),h\bigr)
			\frac{\partial Q_{\lambda^*}(x)}{\partial\tau}
			\right]dx}_{\text{information}}\notag \\
		&-\underbrace{\int_{\mathbb R}g_1^\tau(x)
			\frac{\partial\mathcal V}{\partial h}
			\bigl(Q_{\lambda^*}(x),h\bigr)dx}_{\text{incentive capital}}.
		\label{eq:review-value-envelope}
	\end{align}
	The first term captures first-phase production.  Delaying the review lengthens the first phase, adding a unit of production while tightening first-phase obedience, whose shadow value is charged
	at rate $c$.  The second term is the information effect.  It values a sharper performance measure: the review-score
	means under working and shirking differ by $\tau$ against a standard deviation $\sqrt\tau$,
	so their standardized separation grows with the review date; since the marginal value of
	the index is exactly the rent the menu assigns, precision is worth most at scores that
	already carry a large rent.  The third term is the opposing loss of incentive capital.  A longer continuation phase is \emph{incentive capital}: a family of work-inducing contracts that produce
	second-phase output while delivering different rents, so that rents conceded for
	second-phase incentives can be recycled into rewards for first-phase effort.  As $h$ falls
	this family shrinks, and strong first-phase rewards increasingly require sacrificing
	second-phase production, culminating in $A_G$.
	
	An interior optimum balances these three terms, setting
	\eqref{eq:review-value-envelope} equal to zero.  The optimal review date therefore does not merely trade an early noisy signal against a
	late precise one; it balances the precision of the review against the productive
	contracting capacity that survives it.
	
	\paragraph{The suboptimality of extreme review dates.}
	\label{subsec:endpoint-derivatives}
	The next result distinguishes the implications of the review-date tradeoff near the two boundaries.
	
	\begin{proposition}
		\label{prop:extreme-review-dates}
		Review dates sufficiently close to the initial date are strictly suboptimal.  Review dates
		sufficiently close to the terminal date are also strictly suboptimal when the payment cap is
		sufficiently tight.  Formally,
		\begin{enumerate}[(i)]
			\item for every $\tau$ sufficiently close to the initial date, there exists another genuine review date
			$\tau'$ such that $V^R(\tau')>V^R(\tau)$;
			\item suppose $c\leq\bar c(T)$, so that $\underline M(T,c)=M_I(T,c)$.  If $M$ is sufficiently close to $\underline M(T,c)$ from above, then every review date $\tau$ sufficiently close to the terminal date is strictly dominated by an earlier genuine review date $\tau'<\tau$.
		\end{enumerate}
	\end{proposition}
	
	The first conclusion reflects the information term in \eqref{eq:review-value-envelope}.  A short first phase
	adds only a small amount of production and generates an extremely weak signal of effort.  To
	motivate first-phase effort, the principal must nevertheless create enough dispersion in
	continuation rents across review outcomes.  Because those rents are bounded, extracting the
	required incentives from an almost uninformative signal is disproportionately costly.  The
	first-phase production gain therefore cannot compensate for the incentive loss generated by a
	review sufficiently close to the initial date.
	
	The second conclusion reflects the opposing incentive-capital term.  Near the terminal date, delaying the
	review yields a more informative score but leaves almost no continuation phase over which rents
	can be varied.  A tight payment cap makes this loss especially costly: the terminal contract is
	already close to exhausting its incentive capacity, so continuation rents have a high shadow
	value.  The marginal loss of incentive capital then dominates the production and information
	benefits of further delay, and an earlier review raises the principal's payoff.
	
	As a result, under the tight-cap conditions of Proposition~\ref{prop:extreme-review-dates}, if an
	optimal genuine review exists, it cannot lie near either the beginning or the end of the project.
	It is genuinely a \emph{midterm} review.  An implication is that short
	probationary periods observed in practice must therefore reflect motives absent from the
	model, such as screening an unknown type, learning match quality, or preserving an option to
	abandon a bad project.

	\subsection{When is a genuine review valuable?}
	\label{sec:review-vs-no-review}
	We now compare genuine review with the outside option of no review.  Define the review gain
	for $\tau\in(0,T)$ by
	\begin{equation*}
		G(\tau)=V^R(\tau)-V^{\NR}(T;c,M).
	\end{equation*}
	
	\paragraph{Loose payment caps.}
	When the payment cap is loose, the no-review contract can concentrate compensation on an
	increasingly rare terminal outcome and implement effort with vanishing rent.  A review cannot
	replicate this economy without either splitting the long experiment or relying on an almost
	uninformative early signal.
	
	\begin{proposition}
		\label{prop:loose-cap-no-review}
		Fix $T>0$ and $c\in(0,1)$.  There exists a finite $\overline M(T,c)$ such that
		\begin{equation*}
			G(\tau)<0
			\qquad\text{for every }\tau\in(0,T)
		\end{equation*}
		whenever $M>\overline M(T,c)$.
	\end{proposition}

	The result is uniform over review dates.  A pointwise comparison would still allow the principal
	to move the review date as $M$ rises and stay ahead of the benchmark.  Proposition~\ref{prop:loose-cap-no-review}
	rules this out.  The rent compression established in Proposition~\ref{prop:no-review} eventually
	dominates uniformly: no allocation of production between a review phase and a continuation phase
	can preserve both the precision and the production horizon of the terminal contract.
	
	\paragraph{Tight payment caps.}
	The comparison can reverse when the payment cap is sufficiently close to the implementability
	boundary.  The next result records the resulting local value improvement.
	
	\begin{proposition}
		\label{prop:tight-cap-review-gain}
		Suppose $c\leq\bar c(T)$, so that $\underline M(T,c)=M_I(T,c)$.  If $M$ is sufficiently
		close to $\underline M(T,c)$ from above, there exists $\delta>0$ such that
		\[
		V^R(\tau)>V^{\NR}(T;c,M)
		\qquad\text{for every }\tau\in(T-\delta,T).
		\]
	\end{proposition}
	
	The logic combines the boundary value with the local comparison in
	Proposition~\ref{prop:extreme-review-dates}.  By Lemma~\ref{lem:review-boundary-values},
	$G(\tau)$ converges to zero as $\tau$ approaches $T$ from below.  Under the stated tight-cap conditions,
	Proposition~\ref{prop:extreme-review-dates} shows that moving a terminal-adjacent review earlier
	strictly raises $G(\tau)$.  Together, these observations imply that $G(\tau)>0$ for review dates
	sufficiently close to $T$, as stated in Proposition~\ref{prop:tight-cap-review-gain}.
	
	Economically, moving the review slightly earlier creates a short continuation phase. When the
	cap is tight, the no-review contract is close to exhausting its incentive capacity, so the rents
	needed to induce continuation effort have a high shadow value and can also reward first-phase
	performance. This is what makes a sufficiently late review valuable. Separately,
	Proposition~\ref{prop:terminal-full-ladder} shows that, under the maintained assumptions, the
	optimal fixed-date contract uses the full incentive ladder whenever the review is sufficiently
	close to $T$, regardless of how tight the payment cap is. Hence, in the tight-cap region covered
	by Proposition~\ref{prop:tight-cap-review-gain}, the value improvement is implemented through
	all four incentive margins characterized in Section~\ref{sec:form-optimal-menu}. The resulting
	gain in incentive capacity outweighs the small losses in information and production.
	
	\paragraph{Payment caps and the value of review.}
	The preceding results show how the payment cap governs the value of review.  A review expands
	incentive capacity, but relies on a shorter performance record and may forgo continuation output
	after a stop.  A tight cap raises the value of the additional incentive margins; a loose cap lets
	the terminal contract provide incentives at negligible rent, leaving the information and
	production costs of review to dominate.
	
	Figure~\ref{fig:review-gain} illustrates the two regimes.  Both panels hold $T=1.5$ and	$c=0.6$ fixed and vary only the payment cap, so the contrast is a comparative static in $M$ alone.  Under the tight cap $M=2$, just above the implementability floor $M_I(T,c)=1.96$, the gain is negative at the earliest review dates, turns positive, and reaches a strictly positive interior maximum at $\tau^*=1.09$, where $V^R=0.239$ against	$V^{\NR}=0.195$.  The supporting multiplier there is $\lambda^*=2.19>1+q_2$, so by Corollary~\ref{cor:menu-form} the optimal menu uses all four rungs of the ladder.  Under the loose cap $M=5$ the gain is negative at every review date, and the principal optimally waits until $T$.  In both panels $G$ vanishes at the two boundary dates, as Lemma~\ref{lem:review-boundary-values} requires.  The short initial	interval on which $G<0$ in the left panel is the finite-$\tau$ counterpart of the infinitely steep right derivative at zero in	Proposition~\ref{prop:extreme-review-dates}(i): an almost uninformative probation cannot pay for the rent dispersion it requires.
	
	\begin{figure}[t]
		\centering
		\begin{subfigure}[b]{0.48\textwidth}
			\centering
			\begin{tikzpicture}[x=3.3cm,y=68cm,>=Latex]
				\draw[->,gray!70] (-0.06,0) -- (1.66,0) node[below right,black] {$\tau$};
				\draw[->,gray!70] (-0.06,-0.0042) -- (-0.06,0.0500) node[above,black]
				{\small$G(\tau)$};
				\draw[dotted,gray!85] (1.08931,0) -- (1.08931,0.044045);
				\draw[very thick] plot coordinates {
					(0.00020,-0.001141) (0.00172,-0.002298) (0.00325,-0.002244) (0.00477,-0.001868) (0.00629,-0.001339)
					(0.00782,-0.000725) (0.00934,-0.000063) (0.01086,0.000624) (0.01238,0.001324) (0.01391,0.002027)
					(0.01543,0.002725) (0.01695,0.003414) (0.01848,0.004090) (0.02000,0.004751) (0.02889,0.008266)
					(0.03778,0.011211) (0.04667,0.013700) (0.05556,0.015843) (0.06444,0.017720) (0.07333,0.019387)
					(0.08222,0.020882) (0.09111,0.022235) (0.10000,0.023468) (0.12889,0.026820) (0.15778,0.029433)
					(0.18667,0.031519) (0.21556,0.033211) (0.24444,0.034600) (0.27333,0.035750) (0.30222,0.036709)
					(0.33111,0.037512) (0.36000,0.038189) (0.38889,0.038763) (0.41778,0.039253) (0.44667,0.039673)
					(0.47556,0.040038) (0.50444,0.040357) (0.53333,0.040640) (0.56222,0.040895) (0.59111,0.041127)
					(0.62000,0.041344) (0.64889,0.041549) (0.67778,0.041746) (0.70667,0.041939) (0.73556,0.042129)
					(0.76444,0.042319) (0.79333,0.042509) (0.82222,0.042700) (0.85111,0.042892) (0.88000,0.043083)
					(0.90889,0.043271) (0.93778,0.043452) (0.96667,0.043623) (0.99556,0.043776) (1.02444,0.043905)
					(1.05333,0.043998) (1.08222,0.044043) (1.11111,0.044024) (1.14000,0.043920) (1.16889,0.043704)
					(1.19778,0.043342) (1.22667,0.042791) (1.25556,0.041994) (1.28444,0.040875) (1.31333,0.039335)
					(1.34222,0.037237) (1.37111,0.034388) (1.40000,0.030510) (1.40768,0.029256) (1.41537,0.027890)
					(1.42305,0.026401) (1.43074,0.024773) (1.43842,0.022993) (1.44611,0.021042) (1.45379,0.018897)
					(1.46148,0.016532) (1.46916,0.013918) (1.47685,0.011016) (1.48453,0.007780) (1.49222,0.004150)
					(1.49990,0.000056)};
				\fill (1.08931,0.044045) circle (0.9pt);
				\node[font=\footnotesize] at (1.08931,-0.0012) [below] {$\tau^{*}$};
				\node[font=\footnotesize] at (1.5,-0.0012) [below] {$T$};
			\end{tikzpicture}
			\caption{Tight cap: $M=2$.}
			\label{fig:review-gain-tight}
		\end{subfigure}
		\hfill
		\begin{subfigure}[b]{0.48\textwidth}
			\centering
			\begin{tikzpicture}[x=3.3cm,y=39cm,>=Latex]
				\draw[->,gray!70] (-0.06,0) -- (1.66,0) node[above right,black] {$\tau$};
				\draw[->,gray!70] (-0.06,-0.0870) -- (-0.06,0.0060) node[above,black]
				{\small$G(\tau)$};
				\draw[very thick] plot coordinates {
					(0.00020,-0.002366) (0.00172,-0.007344) (0.00325,-0.010188) (0.00477,-0.012399) (0.00629,-0.014265)
					(0.00782,-0.015903) (0.00934,-0.017376) (0.01086,-0.018723) (0.01238,-0.019968) (0.01391,-0.021130)
					(0.01543,-0.022221) (0.01695,-0.023253) (0.01848,-0.024232) (0.02000,-0.025164) (0.02889,-0.029904)
					(0.03778,-0.033806) (0.04667,-0.037153) (0.05556,-0.040096) (0.06444,-0.042730) (0.07333,-0.045116)
					(0.08222,-0.047299) (0.09111,-0.049310) (0.10000,-0.051175) (0.12889,-0.056430) (0.15778,-0.060750)
					(0.18667,-0.064379) (0.21556,-0.067464) (0.24444,-0.070103) (0.27333,-0.072363) (0.30222,-0.074296)
					(0.33111,-0.075938) (0.36000,-0.077320) (0.38889,-0.078463) (0.41778,-0.079389) (0.44667,-0.080112)
					(0.47556,-0.080647) (0.50444,-0.081004) (0.53333,-0.081194) (0.56222,-0.081224) (0.59111,-0.081103)
					(0.62000,-0.080836) (0.64889,-0.080429) (0.67778,-0.079886) (0.70667,-0.079211) (0.73556,-0.078408)
					(0.76444,-0.077479) (0.79333,-0.076426) (0.82222,-0.075251) (0.85111,-0.073955) (0.88000,-0.072538)
					(0.90889,-0.071000) (0.93778,-0.069340) (0.96667,-0.067557) (0.99556,-0.065649) (1.02444,-0.063613)
					(1.05333,-0.061446) (1.08222,-0.059143) (1.11111,-0.056698) (1.14000,-0.054105) (1.16889,-0.051356)
					(1.19778,-0.048440) (1.22667,-0.045345) (1.25556,-0.042055) (1.28444,-0.038551) (1.31333,-0.034809)
					(1.34222,-0.030796) (1.37111,-0.026470) (1.40000,-0.021768) (1.40768,-0.020444) (1.41537,-0.019084)
					(1.42305,-0.017686) (1.43074,-0.016248) (1.43842,-0.014764) (1.44611,-0.013232) (1.45379,-0.011644)
					(1.46148,-0.009994) (1.46916,-0.008272) (1.47685,-0.006463) (1.48453,-0.004545) (1.49222,-0.002468)
					(1.49990,-0.000038)};
				\fill (0.56222,-0.081224) circle (0.9pt);
				\node[font=\footnotesize] at (1.5,0.0012) [above] {$T$};
			\end{tikzpicture}
			\caption{Loose cap: $M=5$.}
			\label{fig:review-gain-loose}
		\end{subfigure}
		\caption{The review gain $G(\tau)$ under a tight and a loose payment cap}
		\label{fig:review-gain}
	\end{figure}
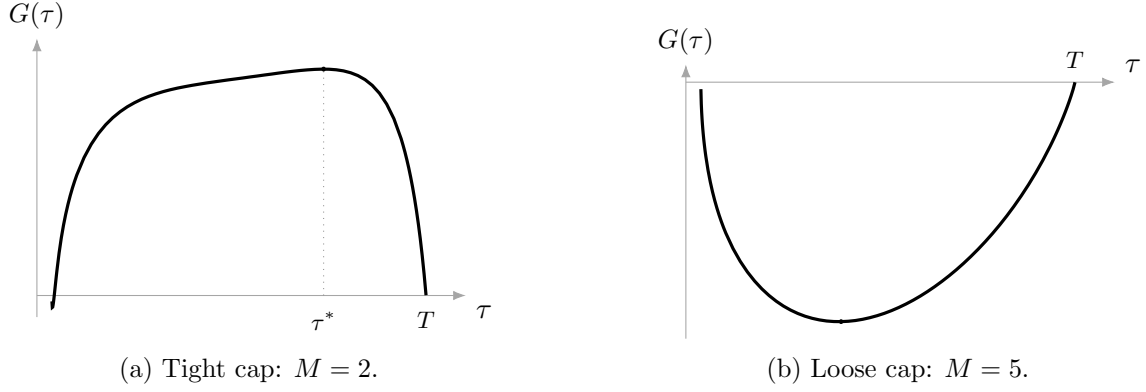

	\paragraph{Review restores implementability under tight caps.}
	
	Under tight caps, the value of a midterm review also lies in restoring implementability.  At the no-review implementability
	boundary, splitting the horizon creates strict incentive slack in both phases, and a
	review permits full effort for caps slightly below the level required without one.
	
	\begin{proposition}
		\label{prop:review-restores-feasibility-boundary}
		Let $M_0\equiv M_I(T,c)$, at which full effort without review is just implementable.
		For every split $\tau\in(0,T)$, with $h=T-\tau$,
		\begin{equation*}
			M_0\overline\Delta(\tau)>c\tau,
			\qquad
			M_0\overline\Delta(h)>ch.
		\end{equation*}
		Moreover there exist $\tau_0\in(0,T)$ and $\varepsilon>0$ such that for every
		$M\in(M_0-\varepsilon,M_0)$ full effort is infeasible without a review, but a review at
		$\tau_0$ admits a two-region menu inducing first-phase work and, whenever the project
		is continued, full continuation effort.  Continuation occurs with positive probability.
	\end{proposition}
	
	A review can therefore create incentives unavailable when performance is evaluated only at
	the terminal date.  Splitting the project gives two shorter incentive problems, each
	strictly easier to implement from $ M_0\overline\Delta(\tau)>c\tau $ and $M_0\overline\Delta(h)>ch$.  More
	importantly, it creates a new source of motivation: the continuation phase has a
	nondegenerate range of work-inducing rents $[r_D,r_E]$, and the principal can condition
	those rents on the review score.  Rents that must in any event be conceded to motivate
	future effort thus double as rewards for past effort, supplemented when necessary by the
	two stopping actions.
		
	\section{Extensions}
	\label{sec:extensions}
	
	The baseline model isolates the role of an endogenous interim review under phase-level effort,
	Gaussian progress, and deterministic review.  This section summarizes three extensions that relax those assumptions in turn.  The main text records the economically
	important objects and the structural conclusions; derivations, complete statements, proofs, and
	numerical illustrations are collected in the Online Appendix.  A useful theme across the
	extensions is that the four-rung ladder is more robust than the exact Gaussian formulas: what
	changes first is typically the location or availability of a rung, while the ordering logic survives
	whenever the review score continues to rank continuation rents monotonically.
	
	\subsection{Pointwise effort choice}
	\label{sec:pointwise-effort}
	
	The baseline model restricts the agent to a phase-by-phase effort choice: in a phase of
	length $h$, he either works throughout or does not work at all.  This section allows him to
	choose effort pointwise.  Because no new information arrives within a phase, only total effort
	matters.  Writing $e_t\in\{0,1\}$ for instantaneous effort, let
	\begin{equation*}
		\eta=\int_0^h e_t\,dt\in[0,h].
	\end{equation*}
	Total effort $\eta$ produces $Y\sim N(\eta,h)$ and costs $c\eta$.  Thus the binary choice
	$\{0,h\}$ in the baseline becomes the interval $[0,h]$: the agent can now ``coast'' by
	working for only part of the phase.
	
	This extension has two sharp implications.  First, the four-rung structure of the optimal
	review menu survives: no interior effort level becomes an exposed continuation action.  Second,
	and more importantly, the easy continuation standard must become harder.  Pointwise effort
	leaves the difficult continuation unchanged, but reduces the largest rent that the principal
	can deliver through the easy continuation.  The extension therefore preserves the form of the
	incentive ladder while weakening one of its rungs.
	
	\paragraph{From a global to a marginal incentive test.}
	For a terminal wage $w$, let $\pi(\eta)=\E_\eta[w(Y)]$ denote expected compensation when total
	effort is $\eta$.  A contract implements the target $\eta$ exactly when no other effort level
	is more attractive:
	\begin{equation*}
		\pi(\eta)-c\eta\;\geq\;\pi(\eta')-c\eta'
		\qquad\text{for every }\eta'\in[0,h].
		\tag{IC$_h^{\mathrm{flex}}$}
		\label{eq:pointwise-full-IC}
	\end{equation*}
	The phase-by-phase model checks only the extreme deviation $\eta'=0$.  Pointwise choice also
	introduces nearby ones, including the temptation to shave the last increment of effort.  Two
	consequences of \eqref{eq:pointwise-full-IC} carry the analysis: the comparison with $\eta'=0$,
	and local optimality,
	\begin{align*}
		\pi(\eta)-c\eta&\;\geq\;\pi(0),
		\tag{G$_\eta$}\label{eq:pointwise-global}\\
		\pi'(\eta)&\;=\;c\ \ \text{if }\eta\in(0,h),
		\qquad
		\pi'(h)\;\geq\;c\ \ \text{if }\eta=h.
		\tag{L$_\eta$}\label{eq:pointwise-local}
	\end{align*}
	The first is the baseline's \emph{global} work-versus-shirk condition; the second is the new
	\emph{marginal} one, an equality at an interior target and an inequality at full effort, where
	only downward deviations are available.  For the Gaussian upper-tail bonus
	$w(y)=M\ind_{\{y\geq h-z\sqrt h\}}$ these read $\Delta(h,z)\geq ch/M$ and
	$\sqrt h\,\phi(z)\geq ch/M$, and Online Appendix~\ref{oa:flexible-effort} shows that the two
	are jointly sufficient for this threshold contract to implement full effort.
	
	The comparison with phase-by-phase effort is now immediate.  At the difficult bar $z_D(h)$,
	the global condition binds and the marginal condition is slack.  The difficult continuation is
	therefore unchanged.  At the easy bar $z_E(h)$, by contrast, the baseline global condition
	binds but the marginal return to the last increment of effort is too small.  The reason is
	simple: an easy bar is cleared with high probability even if the agent works slightly less, so
	coasting barely changes the probability of payment while saving effort cost.  A contract can
	therefore induce ``work rather than shirk'' and still fail to induce ``work all the way.''
	
	The principal removes this temptation by raising the easy standard.  In the $z$ notation a
	larger $z$ means an easier output bar, so the new bar is the positive root $\zeta(h)$ of
	\begin{equation}
		\sqrt h\,\phi\bigl(\zeta(h)\bigr)=\frac{ch}{M},
		\qquad\text{equivalently}\qquad
		\phi\bigl(\zeta(h)\bigr)=\frac{c\sqrt h}{M},
		\label{eq:harder-easy-bar}
	\end{equation}
	which satisfies $\zeta(h)<z_E(h)$; the associated easy action is
	\begin{equation}
		\widetilde A_E=\bigl(\widetilde r_E(h),(1-c)h\bigr),
		\qquad
		\widetilde r_E(h)=M\Phi\bigl(\zeta(h)\bigr)-ch\;<\;r_E(h).
		\label{eq:pointwise-easy-action}
	\end{equation}
	Thus pointwise effort does not merely
	tighten full-effort feasibility; it specifically lowers the maximum continuation rent available
	at the easy end.
	
	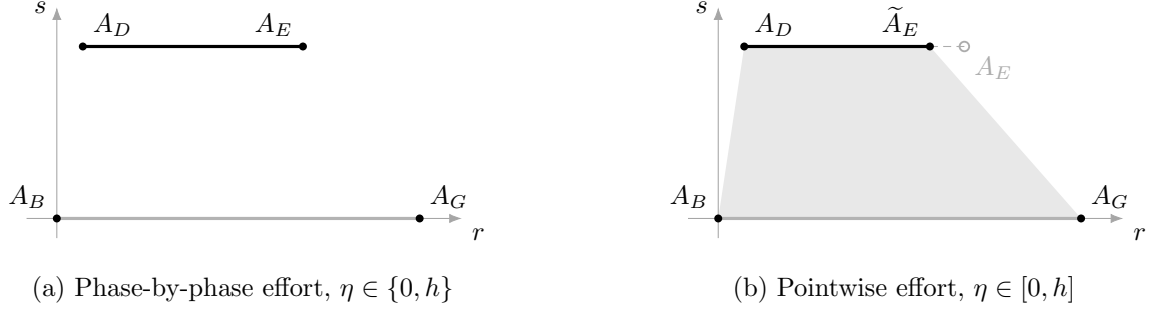
\begin{figure}[t]
		\centering
		\begin{subfigure}[t]{0.47\textwidth}
			\centering
			\begin{tikzpicture}[x=4.0cm,y=3.25cm,>=Latex]
				\clip (-0.18,-0.16) rectangle (1.42,0.93);
				\draw[->,gray!70] (-0.10,0)--(1.34,0) node[below right,black,font=\small] {$r$};
				\draw[->,gray!70] (0,-0.08)--(0,0.86) node[left,black,font=\small] {$s$};
				\draw[very thick,gray!60] (0,0)--(1.20,0);
				\draw[very thick] (0.086,0.700)--(0.814,0.700);
				\fill (0,0) circle (1.5pt) node[above left,font=\small] {$A_B$};
				\fill (0.086,0.700) circle (1.5pt) node[above right,font=\small] {$A_D$};
				\fill (0.814,0.700) circle (1.5pt) node[above left,font=\small] {$A_E$};
				\fill (1.20,0) circle (1.5pt) node[above right,font=\small] {$A_G$};
			\end{tikzpicture}
			\caption{Phase-by-phase effort, $\eta\in\{0,h\}$}
			\label{fig:pointwise-binary}
		\end{subfigure}
		\hfill
		\begin{subfigure}[t]{0.47\textwidth}
			\centering
			\begin{tikzpicture}[x=4.0cm,y=3.25cm,>=Latex]
				\clip (-0.18,-0.16) rectangle (1.42,0.93);
				\draw[->,gray!70] (-0.10,0)--(1.34,0) node[below right,black,font=\small] {$r$};
				\draw[->,gray!70] (0,-0.08)--(0,0.86) node[left,black,font=\small] {$s$};
				\fill[gray!18] (0,0)--(0.086,0.700)--(0.700,0.700)--(1.20,0)--cycle;
				\draw[very thick,gray!60] (0,0)--(1.20,0);
				\draw[very thick] (0.086,0.700)--(0.700,0.700);
				\draw[gray!65,densely dashed] (0.700,0.700)--(0.814,0.700);
				\draw[gray!65,line width=.7pt] (0.814,0.700) circle (1.8pt);
				\fill (0,0) circle (1.5pt) node[above left,font=\small] {$A_B$};
				\fill (0.086,0.700) circle (1.5pt) node[above right,font=\small] {$A_D$};
				\fill (0.700,0.700) circle (1.5pt) node[above left,font=\small] {$\widetilde A_E$};
				\node[below right,font=\small,gray!65] at (0.814,0.700) {$A_E$};
				\fill (1.20,0) circle (1.5pt) node[above right,font=\small] {$A_G$};
			\end{tikzpicture}
			\caption{Pointwise effort, $\eta\in[0,h]$}
			\label{fig:pointwise-flexible}
		\end{subfigure}
		\caption{Rent--surplus sets under phase-by-phase and pointwise effort.}
		\label{fig:pointwise-feasible-set}
	\end{figure}
	
	Figure~\ref{fig:pointwise-feasible-set} displays the central comparison.  Moving from
	phase-by-phase to pointwise effort fills out the feasible rent--surplus set, but its convex hull
	still has four relevant vertices.  Three are unchanged: bad stop $A_B$, difficult continuation
	$A_D$, and good stop $A_G$.  Only the maximum-rent full-effort endpoint moves, from $A_E$ to
	$\widetilde A_E$.  Geometrically, the full-effort face is shortened; economically, the easy
	continuation loses incentive capital.
	
	\paragraph{The four-rung menu survives.}
	Online Appendix~\ref{oa:flexible-effort} proves that the convex hull of
	$\mathcal A^{\mathrm{flex}}(h)$ is $\operatorname{co}\{A_B,A_D,\widetilde A_E,A_G\}$:
	although pointwise effort permits a continuum of targets, no interior $\eta\in(0,h)$ maximizes
	the principal's statewise linear objective.  The optimal menu can therefore still be chosen
	from four actions, ordered by review performance exactly as in the baseline---bad stop,
	difficult continuation, easy continuation, good stop---and all baseline formulas carry over
	after the economically meaningful substitution $r_E(h)\mapsto\widetilde r_E(h)$.  In particular
	the upper breakpoint becomes
	$\widetilde q_2(h)=(1-c)h/\bigl(M-\widetilde r_E(h)\bigr)$, while the lower breakpoint and the
	ordering of the four regions are unchanged.  Since $\widetilde r_E(h)<r_E(h)$, the easy
	continuation is a weaker reward: the principal uses it less and relies relatively more on the
	indivisible good-stop prize.  Pointwise effort therefore changes the strength, but not the
	architecture, of the incentive ladder.  This substitution away from continuation rents toward
	good stopping is costly: it reduces productive continuation and lowers the principal's payoff.
	
	\subsection{Beyond Gaussian noise}
	\label{sec:general-noise}
	
	The normal distribution delivers closed-form thresholds, but it is not what creates the ordered
	incentive ladder.  To separate the distribution-free argument from the Gaussian formulas, replace
	the Brownian increment over a phase of length $h$ by a public signal with law $P_1^h$ under work
	and law $P_0^h$ under shirking.  We retain stationary independent increments, so the continuation
	experiment depends on the review date only through the remaining horizon.  Let $m(h)$ be the
	incremental expected output from work and define continuation surplus net of effort cost by
	$\Sigma(h)=m(h)-ch$.
	
	\paragraph{From threshold contracts to statistical tests.}
	In the Gaussian baseline, the monotone likelihood ratio makes the optimal wage a cutoff in
	output.  With a general signal there need not be a natural output cutoff.  Instead, normalize the
	wage as $\varphi=w/M\in[0,1]$.  A fractional wage $M\varphi(y)$ has the same expected payoffs as
	paying the cap with probability $\varphi(y)$ after signal $y$.  The contract can therefore be
	represented as a possibly randomized statistical test: its
	\emph{size} $\alpha=\E_0[\varphi]$ is the probability of payment under shirking, and its
	\emph{power} $\beta=\E_1[\varphi]$ is the probability of payment under work.
	
	Let $\mathcal R_h$ be the set of all feasible size--power pairs $(\alpha,\beta)$, and let
	$\beta_h(\alpha)$ be the greatest power attainable at size $\alpha$.  The incentive spread from a
	pair is $M(\beta-\alpha)$.  Its largest possible value is
	$M\operatorname{TV}_h$, where
	$\operatorname{TV}_h=\max_{(\alpha,\beta)\in\mathcal R_h}(\beta-\alpha)$,
	is the total-variation distance between the work and shirk signal
	laws, so the one-phase feasibility condition is exactly $ch\leq M\operatorname{TV}_h$.
	Under strict feasibility, the work-inducing tests have sizes in the interval whose endpoints are
	$\alpha_L(h)=\min\{\alpha:\beta_h(\alpha)-\alpha\geq ch/M\}$ and
	$\alpha_H(h)=\max\{\alpha:\beta_h(\alpha)-\alpha\geq ch/M\}$.
	The corresponding work-inducing rents are exactly the interval
	$[r_L(h),r_H(h)]=M[\alpha_L(h),\alpha_H(h)]$, and the full-effort frontier is the horizontal segment $s=\Sigma(h)$ over that interval.  As in
	the baseline, the principal's payoff $u=\Sigma(h)-r$ falls one for one with the rent conceded.
	
	Two-sided evidence in both phases, slack detection, positive continuation surplus, monotone
	review evidence, and public randomization at score atoms recover the baseline ladder.  Under these
	conditions $0<r_L<r_H<M$, and the support function for the statewise objective $s+qr$ is
	\begin{equation*}
		\mathcal V(q,h)
		=\max\{0,\ \Sigma(h)+qr_L(h),\ \Sigma(h)+qr_H(h),\ qM\}.
	\end{equation*}
	The four exposed actions are bad stop $A_B=(0,0)$, low-rent continuation
	$A_L=(r_L,\Sigma)$, high-rent continuation $A_H=(r_H,\Sigma)$, and good stop
	$A_G=(M,0)$.  A monotone likelihood ratio for the review score assigns them in that order as
	performance improves.  Thus Gaussianity is useful because it converts likelihood-ratio tests into
	signal cutoffs and supplies closed-form formulas; it is not needed for the four-action geometry.
	
	\paragraph{Rent is cost divided by evidence.}
	The general experiment also makes the source of agency rent transparent.  If
	$\overline\Lambda_h$ is the average work-to-shirk likelihood ratio on the reward event of the
	cheapest work-inducing test, then $r_L(h)=ch/(\overline\Lambda_h-1)$: rent is effort cost
	divided by the excess evidentiary content of the reward event.  Consequently,
	a loose cap eliminates rent only if rare reward events can carry arbitrarily strong evidence of
	work.  Under the regularity conditions in the Online Appendix,
	$\lim_{M\to\infty}r_L(h)=ch/(\Lambda_h^{\max}-1)$, where $\Lambda_h^{\max}$ is the essential supremum of the work-to-shirk likelihood ratio, with the
	right-hand side interpreted as zero when $\Lambda_h^{\max}=\infty$.  Bounded evidence can therefore
	leave a positive rent floor even as the payment cap becomes arbitrarily loose.
	
	\paragraph{What can fail.}
	The five conditions isolate distinct departures from the baseline.  One-sided evidence can push
	an extreme rent to $0$ or $M$ and remove a rung.  Nonpositive continuation surplus can make
	stopping technologically attractive rather than purely incentive driven.  Failure of monotone
	likelihood ratios can destroy the ordering in the observed score, while atoms can require public
	mixing at a review threshold.  The two likelihood-ratio tails also govern whether the stopping
	rungs are reached: bounded evidence of work makes the good stop harder to use, and bounded evidence
	of shirking similarly limits the bad stop.
	
	\paragraph{Timing depends on short-horizon evidence.}
	The review-date results require more than the feasible size--power set.  
	Brownian evidence becomes locally weak as the phase shrinks, which produces the infinitely negative right derivative at $\tau=0$.\footnote{This distinction is related to \citet{Li2017}, who shows that review contracts can attain near-efficiency under Poisson and Gamma monitoring but not under Brownian monitoring. Under Brownian monitoring, short performance records remain too noisy to support sufficiently demanding performance standards.}  A rare	event whose likelihood ratio remains bounded away from one over a short interval can instead provide incentives at first order.  Such hard evidence can remove the Brownian infinite-slope conclusion, although it does not by itself make an early review optimal.  It can also give a vanishing continuation phase first-order value near $T$, breaking the Brownian equivalence between the review derivative and the no-review derivative.  The distribution-free conclusion is thus a separation: the feasible size--power set delivers the ladder, whereas likelihood-ratio tails and short-horizon asymptotics determine which rungs are reached and how review timing behaves.  Online Appendix~\ref{oa:general-noise} states the exact assumptions and derives the generalized endpoint formulas.
	
	\subsection{Random review dates}
	\label{sec:random-review}
	
	The review date has so far been announced at date~$0$.  This extension lets the principal commit instead to a distribution $\Gamma$ over review dates on $[0,T]$, draw the date independently of the Brownian motion, and withhold its realization until the review occurs. She commits publicly to $\Gamma$ and to a date-contingent menu
	$a_\tau(\cdot)$ for every realization $\tau$; only the draw is hidden.
	
	Timing is valuable for two conflicting reasons.  A later review aggregates more performance
	and so yields a sharper score; an earlier one leaves a longer continuation phase and so more
	incentive capital.  A deterministic date must settle on one mixture of the two.  A lottery
	need not, because the agent chooses pre-review effort without knowing which date will
	arrive.
	
	\paragraph{Phase-level effort.}
	Under Assumption~\ref{as:phase-level} the agent chooses $a_1$ once, at date~$0$, and cannot
	revise it before the review.  Non-arrival therefore changes his beliefs but gives him no new
	action, and obedience is a single scalar constraint.  Writing
	\begin{equation}
		P(\tau,a_\tau)=\tau+\int_{\mathbb R}\bigl[s_\tau(x)-r_\tau(x)\bigr]g_1^\tau(x)\,dx,
		\qquad
		J(\tau,a_\tau)=\int_{\mathbb R}r_\tau(x)\bigl[g_1^\tau(x)-g_0^\tau(x)\bigr]dx-c\tau
		\label{eq:random-P-J}
	\end{equation}
	for the principal's payoff and the incentive surplus of a date--menu pair, the problem is
	\begin{equation*}
		V^{S}=\sup_{\Gamma,\{a_\tau\}}\int P(\tau,a_\tau)\,\Gamma(d\tau)
		\qquad\text{subject to}\qquad
		\int J(\tau,a_\tau)\,\Gamma(d\tau)\geq0 .
		\tag{SR}
		\label{eq:random-review-primal}
	\end{equation*}
	Thus $J\geq0$ is precisely the obedience constraint of the deterministic problem, and
	$V^{D}\equiv\sup_{\tau}V^R(\tau)$ is the value when $\Gamma$ is degenerate.  At $\tau=T$
	there is no continuation phase and the date-$T$ problem is the no-review problem of
	Section~\ref{sec:no-review}.
	
	Two attainable date--menu pairs can be blended into a feasible random contract whenever one
	of them runs an incentive surplus and the other an incentive deficit.  Write $b=(J,P)$ for
	the pair \eqref{eq:random-P-J} generated by a date and its menu, and let $b_+$ and $b_-$ be
	attainable pairs with $J_+>0>J_-$.  Mixing them with weight $p$ on $b_+$ makes aggregate
	obedience bind at exactly one weight,
	\begin{equation}
		p(b_+,b_-)=\frac{-J_-}{J_+-J_-}\in(0,1),
		\label{eq:two-date-mixing-probability}
	\end{equation}
	at which the blend delivers the principal
	\begin{equation}
		\Pi(b_+,b_-)
		=p(b_+,b_-)\,P_++\bigl[1-p(b_+,b_-)\bigr]P_-
		=\frac{-J_-P_++J_+P_-}{J_+-J_-}.
		\label{eq:binding-mixture-payoff}
	\end{equation}
	Note that $p$ weights the \emph{surplus} branch, whichever of the two dates that happens to
	be.  The deterministic benchmark $V^D$ is the yardstick for both statements below.
	
	\begin{proposition}
		\label{prop:random-two-dates}
		Suppose Assumptions~\ref{as:phase-level} and~\ref{as:maintained} hold.  Then
		$V^{D}\leq V^{S}$, and some optimal random contract uses at most two review dates.
		Moreover:
		\begin{enumerate}[(i)]
			\item randomization is strictly valuable, $V^{S}>V^{D}$, if and only if
			\begin{equation}
				\Pi(b_+,b_-)>V^{D}
				\label{eq:strict-gain-primal-test}
			\end{equation}
			for some attainable pairs $b_+,b_-$ with $J_+>0>J_-$;
			\item whenever $V^{S}>V^{D}$, every minimal-support optimum is such a blend: it is
			supported on two distinct dates carrying attainable pairs $b_+$ and $b_-$ with
			$J_+>0>J_-$, mixed with the unique probability $p(b_+,b_-)$ of
			\eqref{eq:two-date-mixing-probability}, and it attains $\Pi(b_+,b_-)=V^{S}$.
		\end{enumerate}
	\end{proposition}
	
	Two dates suffice because there is one aggregate constraint.  One date deliberately supplies
	more incentives than obedience requires and subsidizes another whose menu would not induce
	first-phase work on its own; only the average must be nonnegative.  Parts~(i) and~(ii) apply
	the same blend to different objects: (i) ranges over arbitrary attainable pairs and merely
	certifies that randomization pays, whereas (ii) says the optimum is itself one such blend and
	collects the gain.  In particular an optimal pair is always admissible in the test of~(i),
	but a pair that passes the test need not be optimal.
	
	\paragraph{Pointwise effort.}
	Under the pointwise effort of Section~\ref{sec:pointwise-effort} the scalar-constraint
	argument fails.  If the earliest possible date passes without a review, the agent learns the
	remaining date and can revise his effort from then on, so a deviation must be evaluated
	allowing him to reoptimize.  With two dates there is exactly one such event, and the problem
	stays tractable.
	
	Fix $0<\tau_1<\tau_2<T$ with probability $p$ on $\tau_1$, put $d=\tau_2-\tau_1$, and let
	$r_i$ be the rent schedule used if review occurs at $\tau_i$.  Throughout this extension the
	two branches are indexed by \emph{date}, not by the sign of the incentive surplus.  If
	cumulative pre-review effort is $e$, write
	\begin{equation}
		U_i(e)=\E\bigl[r_i\bigl(e+\sqrt{\tau_i}Z\bigr)\bigr]-ce,
		\qquad
		\widehat U_2(e)=\max_{m\in[e,e+d]}U_2(m),
		\qquad Z\sim N(0,1),
		\label{eq:random-U}
	\end{equation}
	so that $\widehat U_2$ is the late-branch value \emph{after} the agent has learned the date
	and reoptimized; the constraint $m\leq e+d$ is the flow bound on effort.  Let
	\begin{equation*}
		D_1(e)=U_1(\tau_1)-U_1(e),
		\qquad
		D_2(e)=U_2(\tau_2)-\widehat U_2(e)
	\end{equation*}
	be the two branch losses from deviating to $e$ before the first date.
	
	\begin{proposition}
		\label{prop:random-pointwise}
		Suppose Assumption~\ref{as:maintained} holds and effort is chosen pointwise, as in
		Section~\ref{sec:pointwise-effort}.  Fix $0<\tau_1<\tau_2<T$ and a weight $p$ on
		$\tau_1$.  Then:
		\begin{enumerate}[(i)]
			\item \emph{Exact obedience.}  Full pre-review effort is sequentially incentive
			compatible if and only if $U_2(\tau_2)=\widehat U_2(\tau_1)$ and
			\begin{equation}
				pD_1(e)+(1-p)D_2(e)\geq0
				\qquad\text{for every }e\in[0,\tau_1].
				\label{eq:random-exact-IC}
			\end{equation}
			
			\item \emph{A ceiling on the early weight.}  If in addition the late branch would
			induce full effort on its own, $U_2(\tau_2)\geq U_2(m)$ for every $m\in[0,\tau_2]$,
			then \eqref{eq:random-exact-IC} holds if and only if $p\leq\overline p$, where
			\begin{equation}
				\overline p
				=\inf_{\{e\,:\,D_1(e)<0\}}\frac{D_2(e)}{D_2(e)-D_1(e)}
				\label{eq:random-pbar}
			\end{equation}
			and $\overline p=1$ when the set is empty.
			
			\item \emph{Dominance.}  Let $P_i$ be the branch payoffs \eqref{eq:random-P-J}
			under full effort and let $V^{D,\pw}$ be the best deterministic value among contracts
			that induce full pointwise effort.  If $P_1>V^{D,\pw}>P_2$, a two-date lottery
			strictly dominates every such deterministic contract if and only if
			\begin{equation}
				\underline p\equiv\frac{V^{D,\pw}-P_2}{P_1-P_2}\;<\;\overline p ,
				\label{eq:random-dominance}
			\end{equation}
			and any $p\in(\underline p,\overline p]$ achieves it.
		\end{enumerate}
	\end{proposition}
	
	The two bounds separate the two requirements.  Weight $\underline p$ on the high-payoff
	branch is what beating one date costs; weight $\overline p$ is the most that branch can
	carry before some amount of early coasting becomes profitable.  Randomization is valuable
	exactly when the payoff requirement is the weaker of the two.  The phase-level problem
	checks fewer deviations, so $V^{D,\pw}\leq V^{D}$ and $V^D$ may be substituted for
	$V^{D,\pw}$ in \eqref{eq:random-dominance} to obtain a sufficient condition.
	
	\paragraph{A numerical illustration.}
	Both propositions can be illustrated with the same primitives, $T=1.5$, $c=0.72$ and
	$M=2.60$, and with the same pair of review dates.
	
	Under phase-level effort the best deterministic review occurs at $\tau^D\simeq0.862$ and
	yields $V^D\simeq0.07850$.  The optimal lottery reviews at $\tau_1\simeq0.095$ and
	$\tau_2\simeq0.817$.  The early date runs the incentive deficit and the later, more
	informative one supplies the offsetting surplus, so $b_-$ sits at $\tau_1$ and $b_+$ at
	$\tau_2$.  By \eqref{eq:two-date-mixing-probability} the later date carries probability
	$p(b_+,b_-)\simeq0.280$ and the early one the remaining $0.720$, raising the payoff to
	$V^S\simeq0.08279$ --- about $5.5\%$ above the deterministic benchmark.
	
	Under pointwise effort, keep those two dates and let each menu assign $A_B$, $A_D$,
	$\widetilde A_E$, $A_G$ in order, with $\widetilde A_E$ the pointwise easy action
	\eqref{eq:pointwise-easy-action}.  Because the menus now use $\widetilde A_E$ in place of
	$A_E$, the branch payoffs differ from those just reported: $P_1=0.13301$ at the early date
	and $P_2=0.03199$ at the late one.  The late branch is itself pointwise incentive
	compatible, and $\overline p=0.47293$.  Since $V^{D,\pw}\leq V^D=0.07850$, we get
	$\underline p\leq0.46047<\overline p$, so \eqref{eq:random-dominance} holds; at $p=0.465$
	the lottery is worth $0.07896$.  A feasible two-date pointwise contract therefore beats
	every deterministic full-effort pointwise review, without having to solve the deterministic
	pointwise problem itself.
	
	The improvement is small, and the second example is a witness rather than an optimized
	contract.  Its point is qualitative: the value of random timing is not an artifact of
	phase-level effort.  Even when the agent can coast before the first possible date, infers
	the late date from non-arrival, and then optimally revises his remaining effort, randomizing
	over two dates can strictly beat every deterministic review.  Online
	Appendix~\ref{oa:random-review} gives the proofs and the numerical construction.
	
	\section{Conclusion}\label{sec:conclusion}
	
	This paper studies how a principal uses a public midterm review when effort is sequential and compensation is bounded. The review does more than evaluate past performance: it allows the principal to use future production, continuation rents, and termination as incentive instruments. The optimal menu assigns these continuation opportunities through an ordered performance ladder, while the review date balances the informativeness of accumulated performance against the incentive value of the remaining continuation phase.
	
	Although three extensions demonstrate the robustness of our main findings,
two further natural departures would enlarge the institution itself. First, in many organizations, the timing of the final evaluation is flexible: it may be brought forward before the deadline or postponed beyond the nominal deadline. Allowing the principal to choose a deadline at the midterm review would make the length of the continuation phase itself an incentive instrument. Performance realizations near the boundaries of the good- or bad-stop regions could then be assigned a short probationary period rather than being terminated immediately or retained until the deadline under a final performance standard. Second, our benchmark implicitly treats midterm reviews as sufficiently costly that at most one is conducted. Allowing multiple reviews would replace the four-rung menu with a backward recursion over convex continuation sets. Some faces of these sets would be inherited from the next review, while other regions would be strictly convexified, potentially yielding a rating rule that combines jumps, plateaus, and continuously graded bands. Both extensions enlarge the principal's set of institutional instruments, and we leave their analysis for future work.

	\bibliographystyle{chicago}
	
	\bibliography{Midterm_Review_references}

@article{ChenChiu2013,
  author  = {Chen, Bin R. and Chiu, Y. Stephen},
  title   = {Interim Performance Evaluation in Contract Design},
  journal = {The Economic Journal},
  year    = {2013},
  volume  = {123},
  number  = {569},
  pages   = {665--698},
  doi     = {10.1111/ecoj.12018}
}

@article{Ray2007,
  author  = {Ray, Korok},
  title   = {Performance Evaluations and Efficient Sorting},
  journal = {Journal of Accounting Research},
  year    = {2007},
  volume  = {45},
  number  = {4},
  pages   = {839--882},
  doi     = {10.1111/j.1475-679X.2007.00253.x}
}

@article{Holmstrom1979,
  author  = {Holmstr{\"o}m, Bengt},
  title   = {Moral Hazard and Observability},
  journal = {The Bell Journal of Economics},
  year    = {1979},
  volume  = {10},
  number  = {1},
  pages   = {74--91},
  doi     = {10.2307/3003320}
}

@unpublished{Li2017,
  author = {Li, Anqi},
  title  = {Efficiency in Dynamic Agency Models},
  note   = {Working paper, Washington University in St.\ Louis},
  year   = {2017},
}

@article{ely2025,
  author  = {Ely, Jeffrey C. and Georgiadis, George and Rayo, Luis},
  title   = {Feedback Design in Dynamic Moral Hazard},
  journal = {Econometrica},
  year    = {2025},
  volume  = {93},
  number  = {2},
  pages   = {597--621},
  doi     = {10.3982/ECTA21871}
}

@article{elyszydlowski2020,
  author  = {Ely, Jeffrey C. and Szydlowski, Martin},
  title   = {Moving the Goalposts},
  journal = {Journal of Political Economy},
  year    = {2020},
  volume  = {128},
  number  = {2},
  pages   = {468--506},
  doi     = {10.1086/704387}
}

@article{georgiadis2020,
  author  = {Georgiadis, George and Szentes, Bal{\'a}zs},
  title   = {Optimal Monitoring Design},
  journal = {Econometrica},
  year    = {2020},
  volume  = {88},
  number  = {5},
  pages   = {2075--2107},
  doi     = {10.3982/ECTA16475}
}

@article{hoffmann2021,
  author  = {Hoffmann, Florian and Inderst, Roman and Opp, Marcus M.},
  title   = {Only Time Will Tell: A Theory of Deferred Compensation},
  journal = {The Review of Economic Studies},
  year    = {2021},
  volume  = {88},
  number  = {3},
  pages   = {1253--1278},
  doi     = {10.1093/restud/rdaa043}
}

@article{jewitt2008,
  author  = {Jewitt, Ian and Kadan, Ohad and Swinkels, Jeroen M.},
  title   = {Moral Hazard with Bounded Payments},
  journal = {Journal of Economic Theory},
  year    = {2008},
  volume  = {143},
  number  = {1},
  pages   = {59--82},
  doi     = {10.1016/j.jet.2007.12.004}
}

@unpublished{lizzeri2002,
  author = {Lizzeri, Alessandro and Meyer, Margaret A. and Persico, Nicola},
  title  = {The Incentive Effects of Interim Performance Evaluations},
  year   = {2002},
  note   = {CARESS Working Paper 02-09, University of Pennsylvania}
}

@article{lukas2010,
  author  = {Lukas, Christian},
  title   = {Optimality of Intertemporal Aggregation in Dynamic Agency},
  journal = {Journal of Management Accounting Research},
  year    = {2010},
  volume  = {22},
  number  = {1},
  pages   = {157--174},
  doi     = {10.2308/jmar.2010.22.1.157}
}

@article{sannikov2008,
  author  = {Sannikov, Yuliy},
  title   = {A Continuous-Time Version of the Principal--Agent Problem},
  journal = {The Review of Economic Studies},
  year    = {2008},
  volume  = {75},
  number  = {3},
  pages   = {957--984},
  doi     = {10.1111/j.1467-937X.2008.00486.x}
}

@article{spearwang2005,
  author  = {Spear, Stephen E. and Wang, Cheng},
  title   = {When to Fire a {CEO}: Optimal Termination in Dynamic Contracts},
  journal = {Journal of Economic Theory},
  year    = {2005},
  volume  = {120},
  number  = {2},
  pages   = {239--256},
  doi     = {10.1016/j.jet.2004.02.008}
}

@article{varas2020,
  author  = {Varas, Felipe and Marinovic, Iv{\'a}n and Skrzypacz, Andrzej},
  title   = {Random Inspections and Periodic Reviews: Optimal Dynamic Monitoring},
  journal = {The Review of Economic Studies},
  year    = {2020},
  volume  = {87},
  number  = {6},
  pages   = {2893--2937},
  doi     = {10.1093/restud/rdaa012}
}

@article{rodivilov2022,
  author  = {Rodivilov, Alexander},
  title   = {Monitoring Innovation},
  journal = {Games and Economic Behavior},
  year    = {2022},
  volume  = {135},
  pages   = {297--326},
  doi     = {10.1016/j.geb.2022.06.011}
}

@unpublished{ballknoepfle2026,
  author = {Ball, Ian and Knoepfle, Jan},
  title  = {Should the Timing of Inspections be Predictable?},
  year   = {2026},
  note   = {Revise and resubmit, Review of Economic Studies},
  url    = {https://arxiv.org/abs/2304.01385}
}

@unpublished{wong2026,
  author = {Wong, Yu Fu},
  title  = {Dynamic Monitoring Design},
  year   = {2026},
  note   = {Revise and resubmit, American Economic Review},
  url    = {https://sites.google.com/view/yufuwong}
}

@article{cappelliconyon2018,
  author  = {Cappelli, Peter and Conyon, Martin J.},
  title   = {What Do Performance Appraisals Do?},
  journal = {ILR Review},
  year    = {2018},
  volume  = {71},
  number  = {1},
  pages   = {88--116},
  doi     = {10.1177/0019793917698649}
}

@article{deserranno2025,
  author  = {Deserranno, Erika and Kastrau, Philipp and Le{\'o}n-Ciliotta, Gianmarco},
  title   = {Promotions and Productivity: The Role of Meritocracy and Pay Progression in the Public
Sector},
  journal = {American Economic Review: Insights},
  year    = {2025},
  volume  = {7},
  number  = {1},
  pages   = {71--89},
  doi     = {10.1257/aeri.20230594}
}

@article{devaro2007,
  author  = {DeVaro, Jed and Brookshire, Dana},
  title   = {Promotions and Incentives in Nonprofit and For-Profit Organizations},
  journal = {ILR Review},
  year    = {2007},
  volume  = {60},
  number  = {3},
  pages   = {311--339},
  doi     = {10.1177/001979390706000301}
}

@article{barlevyneal2019,
  author  = {Barlevy, Gadi and Neal, Derek},
  title   = {Allocating Effort and Talent in Professional Labor Markets},
  journal = {Journal of Labor Economics},
  year    = {2019},
  volume  = {37},
  number  = {1},
  pages   = {187--246},
  doi     = {10.1086/698899}
}

@article{mas2017,
  author  = {Mas, Alexandre},
  title   = {Does Transparency Lead to Pay Compression?},
  journal = {Journal of Political Economy},
  year    = {2017},
  volume  = {125},
  number  = {5},
  pages   = {1683--1721},
  doi     = {10.1086/693137}
}

@article{bakerjensenmurphy1988,
  author  = {Baker, George P. and Jensen, Michael C. and Murphy, Kevin J.},
  title   = {Compensation and Incentives: Practice vs. Theory},
  journal = {The Journal of Finance},
  year    = {1988},
  volume  = {43},
  number  = {3},
  pages   = {593--616}
}

@article{levintadelis2005,
  author  = {Levin, Jonathan and Tadelis, Steven},
  title   = {Profit Sharing and the Role of Professional Partnerships},
  journal = {The Quarterly Journal of Economics},
  year    = {2005},
  volume  = {120},
  number  = {1},
  pages   = {131--171}
}

@unpublished{DaiWangYang2026,
  author = {Dai, Liang and Wang, Yenan and Yang, Ming},
  title = {Dynamic Contracting with Flexible Monitoring},
  year = {2026},
  note = {Available at SSRN 3496785},
  doi = {10.2139/ssrn.3496785},
  url = {https://ssrn.com/abstract=3496785}
}

@unpublished{liu2026,
  author = {Liu, Yijun},
  title  = {Resource Allocation with Midterm Review},
  year   = {2026},
  note   = {Working paper}
}
	
	\pagebreak

	\appendix
	
	\setcounter{lemma}{0}
	\setcounter{proposition}{0}
	\setcounter{equation}{0}
	
	\counterwithin{equation}{section}
	\counterwithin{lemma}{section}
	\counterwithin{proposition}{section}
	\counterwithin{corollary}{section}
	\counterwithin{remark}{section}

	\part*{Appendix}

	\section{Proofs for Section~\ref{sec:single-phase}}
	\label{app:single-phase}
	
	Recall $\overline\Delta(h)=\max_z\Delta(h,z)=2\Phi(\sqrt h/2)-1$ and
	$M_I(h,c)=ch/\overline\Delta(h)$ from Section~\ref{sec:threshold-contracts}, and put
	$\bar c(h)=\bigl[2\Phi(\sqrt h/2)-1\bigr]/\Phi(\sqrt h/2)$.
	Write
	\begin{equation}
		\ell_T
		=
		\frac{\phi\bigl(z_D(T)-\sqrt T\bigr)}{\phi\bigl(z_D(T)\bigr)}
		=
		\exp\left(z_D(T)\sqrt T-\frac T2\right),
		\qquad
		\lambda_T^{\NR}=\frac{1}{1-\ell_T},
		\label{eq:terminal-no-review-multiplier}
	\end{equation}
	for the likelihood ratio at the no-review difficult bar and the shadow price of the
	no-review incentive constraint.  Assumption~\ref{as:maintained} gives $z_D(T)<\sqrt T/2$,
	hence $\ell_T\in(0,1)$ and $\lambda_T^{\NR}>1$.
	
	\begin{proof}[Proof of Lemma~\ref{lem:threshold-reduction}]
		Let $k\in[-\infty,+\infty]$ solve $\Prb\{Y\geq k\mid e=1\}=\E_1[w]/M$, which exists and is unique because that
		probability falls continuously and strictly from one to zero (with the endpoints understood as limits), and write
		$\varphi=w/M$, $\varphi^{*}=\ind_{\{y\geq k\}} $.  Since $\int(\varphi-\varphi^{*})f_1^h=0$,
		subtracting $\ell^h(k)$ times this vanishing integral gives
		\begin{equation*}
			\E_0[w]-\E_0[\tilde w]
			=M\int(\varphi-\varphi^{*})\bigl[\ell^h-\ell^h(k)\bigr]f_1^h .
		\end{equation*}
		For $y>k$, $\varphi^{*}=1\geq\varphi$ and $\ell^h(y)<\ell^h(k)$; for $y<k$ both inequalities
		reverse.  The integrand is therefore nonnegative almost everywhere, so
		$\E_0[\tilde w]\leq\E_0[w]$, with equality only if $\varphi=\varphi^{*}$ a.e.  Subtracting
		from $\E_1[\tilde w]=\E_1[w]$ gives the incentive comparison.
	\end{proof}
	
	\medskip
	
	\begin{proof}[Proof of Lemma~\ref{lem:Mbar}]
		
		\emph{Implementability.}  By Lemma~\ref{lem:threshold-reduction} the largest attainable
		incentive spread is $M\overline\Delta(h)$, so $\mathcal F_h\neq\emptyset$ exactly when
		$M\geq M_I(h,c)$, and $[z_D,z_E]$ is nondegenerate exactly when the inequality is strict.
		With $a=\sqrt h/2$ we have $\overline\Delta(h)/h=\mathcal G(a)\equiv[2\Phi(a)-1]/(4a^{2})$ and $\mathcal G'(a)=-F(a)/(2a^{3})$ with $F(a)=2\Phi(a)-1-a\phi(a)$; since $F(0)=0$ and $F'(a)=(1+a^{2})\phi(a)>0$, $\mathcal G$ is strictly decreasing, so $M_I=c/\mathcal G$ is strictly increasing in $h$.
		
		\noindent
		\emph{Profitability.}  By Lemma~\ref{lem:threshold-reduction} and
		\eqref{eq:threshold-moments} the cheapest work-inducing contract is the difficult bar, so the
		principal's payoff is $V(M)=h-M\Phi(z_D)$.  Dividing the binding condition
		$M[\Phi(z_D)-\Phi(z_D-\sqrt h)]=ch$ by $\Phi(z_D)$,
		\begin{equation}
			M\Phi(z_D)=\frac{ch}{1-q(z_D)},
			\qquad
			q(z)=\frac{\Phi(z-\sqrt h)}{\Phi(z)} .
			\label{eq:app-cost}
		\end{equation}
		Since $\phi/\Phi$ is strictly decreasing, $\frac{d}{dz}\log q>0$, so $q$ is a strictly
		increasing bijection onto $(0,1)$.  As $M$ rises, $ch/M$ falls and $z_D$ lies on the
		increasing branch of $\Delta(h,\cdot)$, so $z_D$ falls, hence by \eqref{eq:app-cost} so does
		$M\Phi(z_D)$: the payoff $V$ is strictly increasing in $M$, with $V\uparrow(1-c)h>0$ as
		$M\to\infty$ and, as $M\downarrow M_I(h,c)$ where the two bars merge at $\sqrt h/2$,
		\begin{equation}
			V\ \longrightarrow\ h\left[1-\frac{c}{\bar c(h)}\right].
			\label{eq:app-limit}
		\end{equation}
		
		\noindent
		\emph{The threshold.}  If $c\leq\bar c(h)$ then \eqref{eq:app-limit} is nonnegative and
		$V>0$ at every $M>M_I(h,c)$; set $\underline M(h,c)=M_I(h,c)$.  If $c>\bar c(h)$ it is
		negative, so by strict monotonicity there is a unique $M_P(h,c)>M_I(h,c)$ with $V=0$, and
		$\underline M(h,c)=M_P(h,c)$.  At that cap $M\Phi(z_D)=h$ while the binding constraint gives
		$M\Phi(z_D-\sqrt h)=(1-c)h$; dividing, $z_D$ solves $q(z)=1-c$, and $M_P$ is $h$ over
		$\Phi$ of that root.  In both cases $\underline M\geq M_I$, with equality exactly when $c\leq\bar c(h)$.
	\end{proof}

	\begin{proof}[Proof of Proposition~\ref{prop:no-review}]
		By Lemma~\ref{lem:threshold-reduction} and \eqref{eq:threshold-moments} the cheapest work-inducing contract is the difficult bar, whose expected cost under work is $M\Phi(z_D(T))$; thus, $V^{\NR}(T;c,M)=T-M\Phi\bigl(z_D(T,c,M)\bigr)$. Positivity follows from	Lemma~\ref{lem:Mbar} under Assumption~\ref{as:maintained}.  The limits of $V^{\NR}$ and $z_D(T)\to-\infty$ are established in the proof of	Lemma~\ref{lem:Mbar}.
	\end{proof}
	
	\begin{proof}[Proof of Lemma~\ref{lem:attainable-rents}]
		Take $w\in\mathcal F_h$ and let $\tilde w$ be the cost-matched threshold contract of	Lemma~\ref{lem:threshold-reduction}, which also lies in $\mathcal F_h$; its bar therefore	satisfies \eqref{eq:bar-condition}, so $z\in[z_D,z_E]$ and	$M\Phi(z_D)\leq\E_1[w]\leq M\Phi(z_E)$.  Both bounds are attained, and $r=\E_1[w]-ch$ is affine in $w$ on the convex set $\mathcal F_h$, so the attainable rents are exactly $[r_D,r_E]$, swept by mixtures of the two extreme contracts.  
		The second equality is \eqref{eq:threshold-moments} at a binding constraint, and	$r_E<M-ch$ because $\Phi(z_E)<1$.  
		Finally $\Delta(h,z)=\Phi(z)+\Phi(\sqrt h-z)-1$ is	symmetric about $\sqrt h/2$, so $z_E=\sqrt h-z_D$ and $r_E=M\Phi(-z_D)=M-M\Phi(z_D)=M-ch-r_D$.
	\end{proof}
	
	\medskip
	
	The next lemma records two properties of the extreme continuation rents that are used
	repeatedly once the horizon is allowed to vary: they move continuously with the length of the
	continuation phase, and they spread to the full interval $[0,M]$ as that phase vanishes.
	
	\begin{lemma}
		\label{lem:continuation-rent-limits}
		Suppose Assumption~\ref{as:maintained} holds.  Then $M\overline\Delta(h)>ch$ for every
		$h\in(0,T]$; the bars $z_D(h)$ and $z_E(h)$ of \eqref{eq:zD-zE} are continuously
		differentiable in $h$ there; and $r_D$ and $r_E$ are continuous on $(0,T]$ with
		\begin{equation*}
			\lim_{h\downarrow0}r_D(h)=0,
			\qquad
			\lim_{h\downarrow0}r_E(h)=M .
		\end{equation*}
		Consequently $h\mapsto\mathcal A(h)$ is continuous on $(0,T]$ in the Hausdorff metric, and
		$\mathcal A(h)\to\mathcal A(0)=\{(r,0):r\in[0,M]\}$ as $h\downarrow0$.
	\end{lemma}
	
	\begin{proof}
		\emph{Strict feasibility.}  Assumption~\ref{as:maintained} and Lemma~\ref{lem:Mbar} give
		$M>\underline M(T,c)\geq M_I(T,c)$, hence $M\overline\Delta(T)>cT$.  Since
		$\overline\Delta(h)/h$ is strictly decreasing, as shown in the proof of
		Lemma~\ref{lem:Mbar}, every $h\in(0,T)$ satisfies
		$\overline\Delta(h)/h>\overline\Delta(T)/T>c/M$.  In particular the two roots are
		separated: $z_D(h)<\sqrt h/2<z_E(h)$.
		
		\emph{Continuity.}  Fix $h\in(0,T]$.  Because $\Delta(h,\cdot)$ is strictly increasing to
		the left of $\sqrt h/2$ and strictly decreasing to its right, the partial derivative
		$\partial_z\Delta(h,z)=\phi(z)-\phi(z-\sqrt h)$ is strictly positive at $z_D(h)$ and
		strictly negative at $z_E(h)$.  Since $\Delta$ is $C^1$ in $(h,z)$, the implicit function
		theorem applied to $\Delta(h,z)=ch/M$ makes $z_D$ and $z_E$ continuously differentiable at
		$h$.  By Lemma~\ref{lem:attainable-rents}, $r_j(h)=M\Phi\bigl(z_j(h)-\sqrt h\bigr)$ is then
		continuous.
		
		\emph{Short-horizon limits.}  Let $h\downarrow0$.  Applying the mean value theorem to
		$\Delta(h,z)=\Phi(z)-\Phi(z-\sqrt h)$ at the binding difficult bar yields
		$\xi_h\in\bigl(z_D(h)-\sqrt h,\,z_D(h)\bigr)$ with
		$M\sqrt h\,\phi(\xi_h)=M\Delta\bigl(h,z_D(h)\bigr)=ch$, that is,
		$\phi(\xi_h)=c\sqrt h/M\to0$, so $\lvert\xi_h\rvert\to\infty$.  As
		$\xi_h<z_D(h)<\sqrt h/2\to0$, the family is eventually bounded above, so
		$\xi_h\to-\infty$; and $z_D(h)<\xi_h+\sqrt h$ then gives $z_D(h)\to-\infty$.  Hence
		$r_D(h)=M\Phi\bigl(z_D(h)-\sqrt h\bigr)\leq M\Phi\bigl(z_D(h)\bigr)\to0$, and by
		Lemma~\ref{lem:attainable-rents}, $r_E(h)=M-ch-r_D(h)\to M$.
		
		\emph{The action set.}  $\mathcal A_0=[0,M]\times\{0\}$ does not depend on $h$, while
		$\mathcal A_1(h)=[r_D(h),r_E(h)]\times\{(1-c)h\}$ is a segment whose endpoints and height
		all vary continuously in $h$; Hausdorff continuity of the union follows.  As $h\downarrow0$
		that segment converges to $[0,M]\times\{0\}$, which is $\mathcal A_0$.
	\end{proof}
	
	
	\section{Proofs for Section~\ref{sec:exogenous-review}}
	\label{app:review-menu}
	
	\subsection{Review-score cutoffs and incentive supply}
	\label{appB:cutoffs}
	
	Section~\ref{sec:deriving-optimal-menu} establishes that, for a given
	$q=Q_\lambda(x)$, the maximizing action moves through
	$A_B,A_D,A_E,A_G$ at the breakpoints $q_0<q_1=0<q_2$ in
	\eqref{eq:three-breakpoints}.  For every $\lambda>0$, $Q_\lambda$ is strictly
	increasing in $x$ and maps $\mathbb R$ onto $(-\infty,\lambda-1)$.  Hence the
	cutoff associated with $q_j$ is finite exactly when $\lambda>1+q_j$, and
	inverting $Q_\lambda(k)=q_j$ gives
	\begin{equation}
		k_j(\lambda)
		=
		\frac{\tau}{2}-\log\left(1-\frac{1+q_j}{\lambda}\right)
		\quad\text{if }\lambda>1+q_j,
		\qquad
		k_j(\lambda)=+\infty
		\quad\text{otherwise.}
		\label{eq:k-formula}
	\end{equation}
	
	For later use, define
	$\underline\lambda(h)\equiv\max\{0,1+q_0\}=\max\bigl\{0,-V^{\NR}(h)/r_D\bigr\}$,
	where $V^{\NR}(h)=(1-c)h-r_D$. 
	Moreover, $1+q_2= (h+r_D)/(ch+r_D) <1/c$, so any supporting multiplier $\lambda^*\geq1/c$ makes all four regions active.
	
	Define the standardized cutoff under first-phase work by
	$z_j(\lambda,\tau)\equiv\bigl(\tau-k_j(\lambda)\bigr)/\sqrt\tau$.
	By \eqref{eq:k-formula},
	\begin{equation}
		z_j(\lambda,\tau)
		=
		\frac{\sqrt\tau}{2}
		+
		\frac{1}{\sqrt\tau}
		\log\left(1-\frac{1+q_j}{\lambda}\right),
		\qquad j\in\{0,1,2\},
		\label{eq:standardized-review-cutoffs}
	\end{equation}
	with $z_j=-\infty$ when $k_j=+\infty$.  The induced rent schedule has jumps
	$r_D$, $r_E-r_D$, and $M-r_E$ at $k_0,k_1,k_2$, so its first-phase incentive
	supply is
	\begin{equation}
		I(\lambda;\tau)
		=
		r_D\,\Delta\bigl(\tau,z_0\bigr)
		+\bigl[r_E-r_D\bigr]\,\Delta\bigl(\tau,z_1\bigr)
		+\bigl[M-r_E\bigr]\,\Delta\bigl(\tau,z_2\bigr),
		\label{eq:incentive-supply-function}
	\end{equation}
	where $\Delta(\tau,-\infty)=0$.  Since each tie set
	$\{x:Q_\lambda(x)=q_j\}$ is null, \eqref{eq:incentive-supply-function} is
	unaffected by how ties are resolved.
	
	\subsection{The supporting multiplier}
	
	\begin{lemma}
		\label{lem:supporting-multiplier}
		Fix $\tau\in(0,T)$ and suppose Assumption~\ref{as:maintained} holds.
		\begin{enumerate}[(i)]
			\item \emph{Monotonicity and range.}
			$I(\lambda;\tau)$ is continuous and weakly increasing on $\lambda>0$,
			and strictly increasing whenever the selected rent changes on a set of
			positive probability.  Moreover, $\lim_{\lambda\to\infty}I(\lambda;\tau)
			=M\overline\Delta(\tau)$. 
			If $(1-c)h\neq r_D$, then also $\lim_{\lambda\downarrow\underline\lambda(h)}I(\lambda;\tau)=0$.
			
			\item \emph{Existence.}
			$M\overline\Delta(\tau)>c\tau$, and
			Problem~\eqref{eq:fixed-review-primal} admits a supporting multiplier
			$\lambda^*\geq0$ satisfying complementary slackness.
			
			\item \emph{Positivity and binding obedience.}
			If $(1-c)h\neq r_D$, every supporting multiplier satisfies
			$\lambda^*>0$, and pre-review obedience binds:
			\begin{equation}
				\int_{\mathbb R}
				r^*(x)\left[g_1^\tau(x)-g_0^\tau(x)\right]dx
				=
				c\tau,
				\label{eq:binding-pre-review-obedience-lambda}
			\end{equation}
			where $r^*$ is the rent delivered by the optimal menu.
			At the knife edge $(1-c)h=r_D$, a zero multiplier may also support
			an optimum.
			
			\item \emph{Characterization.}
			If $(1-c)h\neq r_D$, every supporting multiplier solves $I(\lambda;\tau)=c\tau$.
			The solution set is a nonempty compact interval contained in
			$(\underline\lambda(h),\infty)$, and is a singleton whenever
			$I(\cdot;\tau)$ is strictly increasing at that level.
		\end{enumerate}
	\end{lemma}
	
	\begin{proof}
		(i) Let $\lambda''>\lambda'$ and let $(r',s')$ and $(r'',s'')$ be
		Lagrangian maximizers at a given review score.  Write
		$d_\tau(x)=1-g_0^\tau(x)/g_1^\tau(x)$.  Adding the two optimality
		inequalities for
		$s-r+\lambda d_\tau r$ gives
		$(\lambda''-\lambda')d_\tau(x)\bigl[r''(x)-r'(x)\bigr]\geq0$.
		Multiplying by $g_1^\tau$ and integrating yields monotonicity of $I$,
		with strictness under the stated condition.  Continuity follows from
		\eqref{eq:standardized-review-cutoffs}: when a cutoff first becomes
		finite, its standardized value tends to $-\infty$ and its contribution
		to \eqref{eq:incentive-supply-function} vanishes.
		
		If $(1-c)h\neq r_D$, the limiting assignment becomes constant across review scores. When $1+q_0>0$, it converges to $A_B$; when
		$1+q_0<0$, it converges to $A_D$.  In either case $I\to0$.  Finally,
		as $\lambda\to\infty$, every $k_j\to\tau/2$, so every
		$z_j\to\sqrt\tau/2$; since the three rent jumps sum to $M$,
		$I\to M\overline\Delta(\tau)$.
		
		(ii) By Lemma~\ref{lem:Mbar} and Assumption~\ref{as:maintained},
		$M>\underline M(T,c)\geq M_I(T,c)>M_I(\tau,c)=c\tau/\overline\Delta(\tau)$,
		where the strict inequality uses $\tau<T$ and monotonicity of
		$M_I(\cdot,c)$.  Thus pre-review obedience is strictly feasible.
		Slater's condition and Kuhn--Tucker yield a supporting
		$\lambda^*\geq0$ with complementary slackness.
		
		(iii) Suppose $\lambda^*=0$. Then $Q_0\equiv-1$, so the review-score
		objective is $s-r$ and does not depend on $x$. If $(1-c)h\neq r_D$,
		its maximizer is unique: $A_B$ when $(1-c)h<r_D$ and $A_D$ when
		$(1-c)h>r_D$. The assigned rent is therefore constant across scores and
		supplies no first-phase incentives, contradicting $c\tau>0$. Hence
		$\lambda^*>0$, and complementary slackness gives
		\eqref{eq:binding-pre-review-obedience-lambda}.
		
		At $(1-c)h=r_D$, instead, $q_0=-1$, so under $\lambda=0$ the bad stop
		$A_B$ and difficult standard $A_D$ are tied at every review score, while
		$A_E$ and $A_G$ are strictly dominated. A score-dependent selection
		between the tied actions can therefore provide the required first-phase
		incentives, so a zero supporting multiplier cannot be ruled out.
		
		(iv) By~(iii), $\lambda^*>0$, so binding obedience is exactly
		$I(\lambda;\tau)=c\tau$.  Parts~(i)--(ii) imply that its solution
		set is nonempty, compact, and bounded away from
		$\underline\lambda(h)$.  Weak monotonicity makes the level set an
		interval, and strict monotonicity at the level makes it a singleton.
	\end{proof}
	
	\subsection{The full ladder near the terminal date}
	\label{appB:terminal-full-ladder}
	
	The terminal boundary also yields a useful characterization of the shape of
	the fixed-review menu.  The argument uses two limits.  First, the threshold
	for activating the good stop converges to one.  Second, the fixed-review
	supporting multiplier converges to the shadow price of the terminal
	no-review problem, which is strictly larger than one.
	
	\begin{lemma}
		\label{lem:terminal-ladder-limits}
		Suppose Assumption~\ref{as:maintained} holds and write $h=T-\tau$.  Then,
		as $h\downarrow0$,
		\begin{equation*}
			\frac{r_D(h)}{h}\longrightarrow+\infty,
			\qquad
			q_2(h)\longrightarrow0,
			\qquad
			\lambda^*(T-h)\longrightarrow\lambda_T^{\NR}>1,
		\end{equation*}
		with $\ell_T$ and $\lambda_T^{\NR}$ as in
		\eqref{eq:terminal-no-review-multiplier}.
	\end{lemma}
	
	\begin{proof}
		For the first limit, the difficult continuation bar satisfies
		$M\bigl[\Phi(z_D(h))-\Phi(z_D(h)-\sqrt h)\bigr]=ch$.
		By the mean-value theorem there is
		$\xi_h\in(z_D(h)-\sqrt h,z_D(h))$ such that
		$M\sqrt h\,\phi(\xi_h)=ch$, that is, $\phi(\xi_h)=c\sqrt h/M$.
		On the difficult branch $\xi_h\to-\infty$, so this implies
		$|\xi_h|=O(\sqrt{\log(1/h)})$.  Since
		$|z_D(h)-\sqrt h-\xi_h|\leq\sqrt h$, the corresponding normal densities
		are asymptotically equivalent.  The lower Mills bound therefore gives,
		for some constant $C>0$ and all sufficiently small $h$,
		$r_D(h)=M\Phi\bigl(z_D(h)-\sqrt h\bigr)\geq C\sqrt h/|\xi_h|$.
		Hence
		$r_D(h)/h\geq C/(\sqrt h\,|\xi_h|)\longrightarrow+\infty$.
		Using $M-r_E(h)=ch+r_D(h)$,
		$q_2(h)=(1-c)h/\bigl(ch+r_D(h)\bigr)\longrightarrow0$.
		
		It remains to establish convergence of the supporting multiplier.  Fix
		$1<\lambda_-<\lambda_T^{\NR}<\lambda_+$.  By
		Lemma~\ref{lem:continuation-rent-limits}, $r_D(h)\to0$ and $r_E(h)\to M$ as
		$h\downarrow0$.  In the incentive-supply function
		\eqref{eq:incentive-supply-function}, the first and third terms therefore
		vanish, while the middle term converges to
		$M\Delta\bigl(T,z_1(\lambda,T)\bigr)$.
		At $\lambda=\lambda_T^{\NR}$,
		\eqref{eq:standardized-review-cutoffs} and
		\eqref{eq:terminal-no-review-multiplier} give
		$z_1(\lambda_T^{\NR},T)=z_D(T)$.  The no-review incentive constraint
		therefore implies
		$M\Delta\bigl(T,z_1(\lambda_T^{\NR},T)\bigr)=cT$.
		By \eqref{eq:terminal-no-review-multiplier}, $\ell_T\in(0,1)$ and $\lambda_T^{\NR}>1$.  Since $z_1(\cdot,T)$ is strictly increasing and
		$\Delta(T,\cdot)$ is strictly increasing to the left of $\sqrt T/2$, the limiting
		incentive supply is strictly increasing in $\lambda$ at $\lambda_T^{\NR}$.  Hence, for
		all sufficiently small $h$,
		$I(\lambda_-;T-h)<c(T-h)<I(\lambda_+;T-h)$.
		Lemma~\ref{lem:supporting-multiplier} then puts every supporting
		multiplier between $\lambda_-$ and $\lambda_+$.  Since the bracket can be
		chosen arbitrarily tightly around $\lambda_T^{\NR}$,
		$\lambda^*(T-h)\to\lambda_T^{\NR}$.
	\end{proof}
	
	\subsection{Proofs}
	
	\begin{proof}[Proof of Proposition~\ref{prop:four-region-review-menu}]
		By Lemma~\ref{lem:supporting-multiplier}, a supporting multiplier exists.
		Outside the knife edge it is strictly positive, so the review-score ranking	derived in Section~\ref{sec:deriving-optimal-menu}, together with the strict	monotonicity of $Q_{\lambda^*}(x)$, gives the ladder structure of $a^*$: $A_B$ if $x < k_0^*$, $A_D$ if $k_0^* \le x < k_1^*$, $A_E$ if $k_1^* \le x < k_2^* $, and $A_G $ if $x \ge k_2^*$. The cutoff formula and the conditions
		for absent regions follow from \eqref{eq:k-formula}. Since
		$q_0<q_1<q_2$, all finite cutoffs are strictly ordered.
		
		At the knife edge $(1-c)h=r_D$, the preceding argument applies whenever the optimum is
		supported by a positive multiplier.  If instead $\lambda^*=0$ supports an optimum, $A_B$
		and $A_D$ are tied at every review score, while $A_E$ and $A_G$ are strictly dominated.
		Since first-phase work is implemented, an ordered threshold selection between $A_B$ and
		$A_D$ can be chosen to satisfy obedience.  The proposition therefore continues to hold,
		in this case with $k_1^*=k_2^*=+\infty$.
	\end{proof}
	
	\begin{proof}[Proof of Corollary~\ref{cor:menu-form}]
		If $k_0^*=+\infty$, the assigned rent is identically zero and cannot
		induce first-phase work.  Hence $k_0^*<+\infty$.  Since
		$Q_{\lambda^*}$ is continuous and strictly increasing, the
		difficult-standard region is also nonempty.  Finally,
		$Q_{\lambda^*}(\mathbb R)=(-\infty,\lambda^*-1)$ reaches $q_1=0$ iff
		$\lambda^*>1$ and reaches $q_2$ iff $\lambda^*>1+q_2$, which gives the
		three cases in the corollary.
	\end{proof}
	
	\begin{proof}[Proof of Proposition~\ref{prop:terminal-full-ladder}]
		Write $h=T-\tau$.  By Lemma~\ref{lem:terminal-ladder-limits},
		$\lambda^*(T-h)\to\lambda_T^{\NR}>1$ and $1+q_2(h)\to1$.
		Hence, for all sufficiently small $h>0$,
		$\lambda^*(T-h)>1+q_2(h)$.
		Moreover, $r_D(h)/h\to\infty$ implies
		$r_D(h)\neq(1-c)h$ for all sufficiently small $h$, so the generic
		cutoff characterization in Corollary~\ref{cor:menu-form} applies.
		Therefore
		$k_2^*(T-h)<+\infty$
		for every sufficiently small $h>0$, as claimed.
	\end{proof}
	
	
	\section{Proofs for Section~\ref{sec:endogenous-review}}
	\label{app:review-date}
	
	\subsection{The fixed-review value and the one-sided derivatives}
	
	Let $u_D(h)=(1-c)h-r_D(h)$ and $u_E(h)=(1-c)h-r_E(h)$ be the principal's payoffs
	under the two continuation standards, with the argument suppressed where the horizon is
	fixed.  Conditional on first-phase work, the probabilities of the four
	regions are determined by the standardized cutoffs, and the principal's value from the
	optimal fixed-$\tau$ menu is
	\begin{equation}
		\begin{aligned}
			V^R(\tau)
			={}&
			\tau
			+
			u_D
			\left[
			\Phi\bigl(z_0(\lambda,\tau)\bigr)
			-
			\Phi\bigl(z_1(\lambda,\tau)\bigr)
			\right]
			\\
			&+
			u_E
			\left[
			\Phi\bigl(z_1(\lambda,\tau)\bigr)
			-
			\Phi\bigl(z_2(\lambda,\tau)\bigr)
			\right]
			-
			M\Phi\bigl(z_2(\lambda,\tau)\bigr),
		\end{aligned}
		\label{eq:fixed-review-value-cutoffs}
	\end{equation}
	where $\lambda=\lambda^*(\tau)$ is a supporting multiplier and
	$\Phi(-\infty)=0$.  Equation~\eqref{eq:fixed-review-value-cutoffs} applies to the two-,
	three-, and four-region cases under the conventions of \eqref{eq:standardized-review-cutoffs}.
	
	For interior review dates, continuity follows from the implicit-function
	characterization of $z_D(h)$ and $z_E(h)$, continuity of the review-action
	set, and the maximum theorem applied to the fixed-review linear program.
	Changes in the number of review regions may create changes in the formula
	used to compute the value, but not jumps in the value itself.
	
	By Lemma~\ref{lem:review-boundary-values} the review gain $G(\tau)=V^R(\tau)-V^{\NR}(T;c,M)$ vanishes at both
	boundaries, and the natural objects to study there are the one-sided derivatives
	\begin{equation}
		V^{R\prime}(0_+)
		:=\lim_{\tau\downarrow0}
		\frac{V^R(\tau)-V^{\NR}(T;c,M)}{\tau},
		\qquad
		V_-^{R\prime}(T)
		:=\lim_{\tau\uparrow T}
		\frac{V^R(\tau)-V^{\NR}(T;c,M)}{\tau-T},
		\label{eq:one-sided-derivative-definitions}
	\end{equation}
	whenever the corresponding limits exist.  These definitions do not extend $V^R$ to
	either boundary.

	\subsection{The terminal derivative}
	
	Propositions~\ref{prop:extreme-review-dates}(ii)
	and~\ref{prop:tight-cap-review-gain} turn on the sign of
	$V_-^{R\prime}(T)$, computed here.  Throughout, $z_D(T)$ is the difficult
	bar of \eqref{eq:zD-zE} at horizon $T$, and $\ell_T$ and
	$\lambda_T^{\NR}$ are defined in
	\eqref{eq:terminal-no-review-multiplier}.
	
	\begin{lemma}
		\label{lem:terminal-derivative}
		Suppose Assumption~\ref{as:maintained} holds.  Then $V_-^{R\prime}(T)$ exists,
		the review value and the no-review value agree to first order at $T$,
		\begin{equation*}
			V^R(T-h)=V^{\NR}(T-h;c,M)+o(h),
			\qquad
			V_-^{R\prime}(T)=\frac{\partial V^{\NR}(T;c,M)}{\partial T},
		\end{equation*}
		and
		\begin{equation}
			V_-^{R\prime}(T)
			=
			1-\lambda_T^{\NR}\left[c-
			\frac{M\phi(z_D(T)-\sqrt T)}{2\sqrt T}\right].
			\label{eq:left-derivative-T-explicit}
		\end{equation}
		Consequently:
		\begin{enumerate}[(i)]
			\item $V_-^{R\prime}(T)<0$ if and only if
			\begin{equation}
				\ell_T-(1-c)
				\;>\;
				\frac{M\phi(z_D(T)-\sqrt T)}{2\sqrt T};
				\label{eq:terminal-sign-condition}
			\end{equation}
			\item $V_-^{R\prime}(T)<0$ implies $\lambda_T^{\NR}>1/c$, equivalently
			$\ell_T>1-c$.
		\end{enumerate}
	\end{lemma}
	
	\begin{proof}
		As $h\downarrow0$, the work-inducing continuation frontier differs from the
		no-effort frontier by a vertical surplus of only $(1-c)h$.  Moreover,
		$r_D$ and $M-r_E$ are both $O\bigl(\sqrt h/\sqrt{\log(1/h)}\bigr)$.
		The continuation actions can therefore improve on the no-effort support
		function only on likelihood-index intervals whose widths are
		$O\bigl(\sqrt{h\log(1/h)}\bigr)$.  Since the density of the review signal
		is bounded, their total incremental value is $o(h)$, which gives
		$V^R(T-h)=V^{\NR}(T-h;c,M)+o(h)$.  With
		$\lim_{\tau\uparrow T}V^R(\tau)=V^{\NR}(T;c,M)$ from
		Lemma~\ref{lem:review-boundary-values} and
		definition~\eqref{eq:one-sided-derivative-definitions}, this yields
		$V_-^{R\prime}(T)=\partial V^{\NR}(T;c,M)/\partial T$.
		
		It remains to differentiate the no-review value.  Implicit differentiation of
		$M\bigl[\Phi(z_D(T))-\Phi(z_D(T)-\sqrt T)\bigr]=cT$
		gives
		\begin{equation*}
			\frac{d z_D(T)}{dT}
			=
			\frac{
				c/M-
				\phi(z_D(T)-\sqrt T)/(2\sqrt T)
			}{
				\phi(z_D(T))-\phi(z_D(T)-\sqrt T)
			}.
		\end{equation*}
		Differentiating $V^{\NR}(T;c,M)=T-M\Phi(z_D(T))$ and using
		\eqref{eq:terminal-no-review-multiplier} gives
		\eqref{eq:left-derivative-T-explicit}.
		
		For (i), substitute $\lambda_T^{\NR}=1/(1-\ell_T)$ into
		\eqref{eq:left-derivative-T-explicit} and multiply by $1-\ell_T>0$:
		$V_-^{R\prime}(T)<0$ is equivalent to
		$c-M\phi(z_D(T)-\sqrt T)/(2\sqrt T)>1-\ell_T$, which is
		\eqref{eq:terminal-sign-condition}.  For (ii), the term
		$\lambda_T^{\NR}M\phi(z_D(T)-\sqrt T)/(2\sqrt T)$ in
		\eqref{eq:left-derivative-T-explicit} is strictly positive, so
		$V_-^{R\prime}(T)<0$ forces $1-\lambda_T^{\NR}c<0$; and
		$\lambda_T^{\NR}>1/c$ is equivalent to $1-\ell_T<c$.
	\end{proof}
	
	The left side of \eqref{eq:terminal-sign-condition} is large when the reward event is
	weak evidence of work, that is when the shadow price of terminal incentives is high;
	the right side is the first-order informational gain from lengthening the horizon.
	Part~(ii) is \eqref{eq:terminal-sign-condition} with that right side replaced by zero,
	and hence weaker.  Recall from Lemma~\ref{lem:Mbar} that
	$\underline M(T,c)=M_I(T,c)$ when $c\leq\bar c(T)$.

	\begin{lemma}
		\label{lem:terminal-derivative-limits}
		Fix $T>0$ and $c\in(0,1)$.  Then
		\begin{equation}
			\lim_{M\downarrow M_I(T,c)}V_-^{R\prime}(T)=-\infty,
			\qquad
			\lim_{M\to\infty}V_-^{R\prime}(T)=1-c>0.
			\label{eq:tight-cap-derivative-limit}
		\end{equation}
	\end{lemma}
	
	\begin{proof}
		Let $a=\sqrt T/2$.  As $M\downarrow M_I(T,c)$ the two no-review roots
		merge, so $z_D(T)\uparrow a$, $\ell_T\uparrow1$, and
		$\lambda_T^{\NR}\to\infty$.  The bracket in
		\eqref{eq:left-derivative-T-explicit} converges to
		$c\bigl[1-a\phi(a)/(2\Phi(a)-1)\bigr]$,
		which is strictly positive because
		\begin{equation}
			2\Phi(a)-1-a\phi(a)>0
			\qquad\text{for every }a>0.
			\label{eq:normal-cap-inequality}
		\end{equation}
		Indeed, the expression on the left of
		\eqref{eq:normal-cap-inequality} is zero at $a=0$ and has derivative
		$(1+a^2)\phi(a)>0$.  The first limit in
		\eqref{eq:tight-cap-derivative-limit} follows.
		
		As $M\to\infty$, $z_D(T)\to-\infty$ by the proof of
		Lemma~\ref{lem:Mbar}; hence $\ell_T\to0$,
		$\lambda_T^{\NR}\to1$, and
		$M\phi(z_D(T)-\sqrt T)\to0$.  Equation
		\eqref{eq:left-derivative-T-explicit} then gives the second limit.
	\end{proof}
	
	\begin{corollary}
		\label{cor:tight-cap-neighborhood}
		Suppose $c\leq\bar c(T)$, so that
		$\underline M(T,c)=M_I(T,c)$.  Then there exists $\delta>0$ such that
		$V_-^{R\prime}(T)<0$
		for every
		$M\in\bigl(\underline M(T,c),\underline M(T,c)+\delta\bigr)$.
	\end{corollary}
	
	\begin{proof}
		Immediate from the first limit in
		\eqref{eq:tight-cap-derivative-limit}.
	\end{proof}
	
	\subsection{Proofs}
	
	\begin{proof}[Proof of Lemma~\ref{lem:review-boundary-values}]
		Consider first $\tau\downarrow0$.  The largest principal payoff available
		from a second-phase action is $u_D=(1-c)h-r_D$ for all sufficiently small
		$\tau$, because profitability of the no-review contract implies
		$u_D(T)=V^{\NR}(T;c,M)>0$.  Hence
		$V^R(\tau)\leq\tau+u_D(T-\tau)$,
		whose limit is $u_D(T)=V^{\NR}(T;c,M)$.
		
		For a matching lower bound, assign the low-rent continuation action $A_D$
		at every review score except on a sufficiently high upper tail, on which the
		principal assigns the good stop $A_G$.  The utility increment on that tail is
		$M-r_D$.  Choose the upper-tail cutoff so that
		$\bigl[M-r_D\bigr]\Delta(\tau,z)=c\tau$.
		For sufficiently small $\tau$, such a cutoff exists, and the work probability
		of the reward region converges to zero.  The resulting payoff converges to
		$u_D(T)$.  Therefore
		$\lim_{\tau\downarrow0}V^R(\tau)=V^{\NR}(T;c,M)$.
		
		Now let $\tau\uparrow T$, so $h\downarrow0$.  A menu using only $A_B$ and
		$A_G$ is exactly a bounded upper-tail bonus based on $X_\tau$.  Choosing its
		threshold to satisfy
		$M\,\Delta(\tau,z)=c\tau$
		gives payoff $V^{\NR}(\tau;c,M)$, which converges to
		$V^{\NR}(T;c,M)$.
		
		For the reverse inequality, the full-effort continuation frontier converges
		to the no-effort frontier.  In particular $(1-c)h\to0$ and, by
		Lemma~\ref{lem:continuation-rent-limits}, $r_D\to0$ and $r_E\to M$.
		Thus the complete review-action set converges to
		$\left\{(r,0):r\in[0,M]\right\}$,
		so the associated principal payoff converges to $-r$.
		The review signal converges to the terminal signal $N(T,T)$, and the
		fixed-review problem converges to the no-review bounded-payment problem.
		This proves
		$\lim_{\tau\uparrow T}V^R(\tau)=V^{\NR}(T;c,M)$.
	\end{proof}

	\begin{proof}[Proof of Proposition~\ref{prop:extreme-review-dates}]
		The two claims concern the one-sided derivatives
		\eqref{eq:one-sided-derivative-definitions}.  Part~(i) establishes that
		$V^{R\prime}(0_+)=-\infty$ directly; part~(ii) reads the sign of $V_-^{R\prime}(T)$
		off Lemma~\ref{lem:terminal-derivative}.
		
		For the first claim, write $h=T-\tau$ and use the low-rent continuation action
		$A_D(h)$ as a baseline.  For all sufficiently small $\tau$, this action
		uniquely maximizes the principal's post-review payoff.  Compactness of the
		two continuation frontiers and strict profitability at $h=T$ imply that
		there is a constant $\kappa>0$ such that every feasible review action
		$(r,s)$ satisfies
		\begin{equation}
			u_D-(s-r)
			\geq
			\kappa\lvert r-r_D\rvert.
			\label{eq:review-loss-controls-utility-change}
		\end{equation}
		Let
		$L_\tau(x)=1-g_0^\tau(x)/g_1^\tau(x)$.
		Because a constant review utility provides no incentives, pre-review
		obedience can be written as
		\begin{equation}
			\E_1\left[
			\bigl(r(X_\tau)-r_D\bigr)L_\tau(X_\tau)
			\right]
			\geq
			c\tau.
			\label{eq:small-tau-centered-obedience}
		\end{equation}
		Under work, $X_\tau=\tau+\sqrt\tau Z$ and
		$L_\tau(X_\tau)=1-\exp\bigl(-\tau/2-\sqrt\tau Z\bigr)$.
		A standard Gaussian tail decomposition gives
		\begin{equation}
			\E_1\left[
			\lvert r(X_\tau)-r_D\rvert
			\right]
			\geq
			K
			\frac{\sqrt\tau}{\sqrt{\log(1/\tau)}}
			\label{eq:small-tau-utility-dispersion-lower-bound}
		\end{equation}
		for some $K>0$ and all sufficiently small $\tau$.  To see the order, split
		at
		$\lvert Z\rvert=\sqrt{8\log(1/\tau)}$.
		On the central event, $\lvert L_\tau\rvert$ is of order
		$\sqrt{\tau\log(1/\tau)}$; under both work and shirking, the complementary
		Gaussian tails contribute $o(\tau)$.  Since review utilities are bounded by
		$M$, \eqref{eq:small-tau-centered-obedience} then implies
		\eqref{eq:small-tau-utility-dispersion-lower-bound}.
		
		Combining \eqref{eq:review-loss-controls-utility-change} and
		\eqref{eq:small-tau-utility-dispersion-lower-bound} yields
		\begin{equation*}
			V^R(\tau)
			\leq
			\tau+u_D(T-\tau)
			-
			\kappa K
			\frac{\sqrt\tau}{\sqrt{\log(1/\tau)}}.
		\end{equation*}
		Since
		$\tau+u_D(T-\tau)=V^{\NR}(T;c,M)+O(\tau)$,
		it follows that
		\begin{equation}
			\lim_{\tau\downarrow0}
			\frac{V^R(\tau)-V^{\NR}(T;c,M)}{\tau}=-\infty.
			\label{eq:right-derivative-zero}
		\end{equation}
		Thus $V^R(\tau)<V^{\NR}(T;c,M)$ for all sufficiently small $\tau>0$.
		Since the boundary limit established above gives
		$V^R(s)\to V^{\NR}(T;c,M)$ as $s\downarrow0$, each such $\tau$ is strictly
		dominated by some genuine review date $s\in(0,\tau)$.  This proves part~(i).

		For the second claim, let $M$ lie in the right neighborhood supplied by
		Corollary~\ref{cor:tight-cap-neighborhood}, so that $V_-^{R\prime}(T)<0$.
		By Lemma~\ref{lem:terminal-derivative},
		$V^R(T-h)=V^{\NR}(T;c,M)-V_-^{R\prime}(T)h+o(h)$.  Comparing the dates $\tau=T-h$
		and $\tau'=T-2h$,
		$V^R(\tau')-V^R(\tau)=-V_-^{R\prime}(T)h+o(h)>0$
		for all sufficiently small $h>0$.  Hence every sufficiently late review is strictly
		dominated by an earlier genuine review, which proves the second claim.
	\end{proof}
	
	\begin{proof}[Proof of Proposition~\ref{prop:loose-cap-no-review}]
		Write $R_M(t)\equiv R^{\NR}(t;c,M)=r_D(t)$.  By
		Lemma~\ref{lem:Mbar}, $R_M(T)\to0$ as $M\to\infty$, so the
		no-review full-effort contract is feasible and profitable for all sufficiently
		large $M$.
		
		Suppose toward a contradiction that there are sequences $M_n\to\infty$ and
		$\tau_n\in(0,T)$ such that
		\begin{equation}
			V^R(\tau_n)\geq V^{\NR}(T;c,M_n).
			\label{eq:app-loose-cap-contradiction}
		\end{equation}
		Dependence on $n$ is suppressed below.  Any continuation-utility schedule
		$r(X_\tau)\in[0,M]$ that induces first-phase work is itself a bounded
		incentive contract for a phase of length $\tau$.  By the definition of
		$R_M(\tau)$, the agent's ex ante rent under the review contract is therefore
		at least $R_M(\tau)$.  Total surplus is at most $(1-c)T$, and hence
		$V^R(\tau)\leq(1-c)T-R_M(\tau)$.
		Comparison with \eqref{eq:app-loose-cap-contradiction} yields the necessary
		condition
		\begin{equation}
			R_M(\tau)\leq R_M(T).
			\label{eq:app-loose-cap-rent-necessary}
		\end{equation}
		
		We first show that any such sequence must satisfy
		\begin{equation}
			\tau_n\log M_n\longrightarrow0.
			\label{eq:app-loose-cap-localization}
		\end{equation}
		Implicit differentiation of the binding difficult-root condition gives
		\begin{equation}
			R_M'(t)
			=
			\frac{\phi(z_D(t)-\sqrt t)}
			{\phi(z_D(t))-\phi(z_D(t)-\sqrt t)}
			\left[c-\frac{M\phi(z_D(t))}{2\sqrt t}\right].
			\label{eq:app-loose-cap-rent-derivative}
		\end{equation}
		For every $\varepsilon>0$, $z_D(t)\to-\infty$ uniformly on
		$t\in[\varepsilon,T]$ as $M\to\infty$.  The Mills bound
		$\Phi(z)\leq\phi(z)/|z|$ for $z<0$, together with
		$M\Delta(t,z_D(t))=ct$, makes the bracket in
		\eqref{eq:app-loose-cap-rent-derivative} negative uniformly on that interval.
		The prefactor is positive on the difficult branch.  Hence
		$R_M'(t)<0$ on $[\varepsilon,T]$ for all sufficiently large $M$.
		Condition \eqref{eq:app-loose-cap-rent-necessary} then forces
		$\tau_n\to0$.
		
		We sharpen this localization.  Let $a_M<0$ solve
		$\Phi(a_M)=2cT/M$, and put
		$q_M=\Phi(a_M-\sqrt T)/\Phi(a_M)$.
		The threshold bonus $M/[2(1-q_M)]$ is feasible for large $M$, implements
		effort, and leaves rent $cTq_M/(1-q_M)$.  Gaussian tail bounds therefore give
		\begin{equation}
			R_M(T)
			\leq
			2cT\exp\left(-\sqrt{T\log M}-\frac{T}{2}\right)
			\label{eq:app-terminal-rent-upper-bound}
		\end{equation}
		for all sufficiently large $M$.
		
		For a short phase of length $\tau$, write
		$z_D(\tau)=-x$ and $u=\sqrt\tau$.  The binding condition implies
		$R_M(\tau)=c\tau q/(1-q)\geq c\tau q$, where $q=\Phi(-(x+u))/\Phi(-x)$.
		The two-sided Mills inequalities and
		$x\leq\sqrt{2\log(M/(c\tau))}$ give
		\begin{equation}
			R_M(\tau)
			\geq
			\frac{c\tau}{2(1+\sqrt T)}
			\exp\left[
			-\sqrt{2\tau\log\left(\frac{M}{c\tau}\right)}-\frac{\tau}{2}
			\right].
			\label{eq:app-short-rent-lower-bound}
		\end{equation}
		If $\tau_n\log M_n$ were bounded away from zero along a subsequence, then
		the logarithm of the right side of
		\eqref{eq:app-short-rent-lower-bound} would be $-o(\sqrt{\log M_n})$,
		whereas \eqref{eq:app-terminal-rent-upper-bound} is of order
		$-\sqrt{T\log M_n}$.  It would follow that
		$R_{M_n}(\tau_n)>R_{M_n}(T)$, contradicting
		\eqref{eq:app-loose-cap-rent-necessary}.  This proves
		\eqref{eq:app-loose-cap-localization}.
		
		It remains to show that a review in this vanishingly early region is also
		strictly dominated.  Put $h=T-\tau$, $r_0=R_M(h)$, and let the difficult
		continuation action be $(r_0,s_0)$, where $s_0=(1-c)h$ is total surplus and
		its principal payoff is $s_0-r_0$.  Uniformly for $h\in[T/2,T]$,
		$r_0\to0$; hence $2r_0<(1-c)h$ for all sufficiently large $M$.  Every
		continuation action $(r,s)\in\mathcal A(h)$ then satisfies
		\begin{equation}
			s-r\leq s_0-r_0-|r-r_0|.
			\label{eq:app-frontier-distance-total-surplus}
		\end{equation}
		For a work-inducing action $s=(1-c)h=s_0$ and $r\geq r_0$, so the inequality
		holds with equality.  For a no-effort action $s=0$; it follows immediately
		when $r\geq r_0$, and when $r<r_0$ it follows from
		$2r_0<(1-c)h$.
		
		Let $d(x)=r(x)-r_0$.  Integrating
		\eqref{eq:app-frontier-distance-total-surplus} under first-phase work gives
		\begin{equation}
			V^R(\tau)
			\leq \tau+s_0-r_0-\E_1|d(X_\tau)|.
			\label{eq:app-review-distance-bound}
		\end{equation}
		Since a constant rent contributes nothing to incentives, first-phase obedience is
		\begin{equation}
			\E_1[d(X_\tau)L_\tau(X_\tau)]\geq c\tau,
			\label{eq:app-centered-obedience}
		\end{equation}
		Under work, $X_\tau=\tau+\sqrt\tau Z$ and
		$L_\tau(X_\tau)=1-\exp(-\tau/2-\sqrt\tau Z)$.  On the event
		$|L_\tau|\leq1/2$, the left side of
		\eqref{eq:app-centered-obedience} is bounded by $\E_1|d|/2$.
		On its complement, $|d|\leq M$, and Gaussian tail bounds give, with
		$a=\log(3/2)$,
		$\E_1\bigl[|L_\tau|\ind_{\{|L_\tau|>1/2\}}\bigr]\leq4\exp\bigl(-a^2/(8\tau)\bigr)$.
		Because \eqref{eq:app-loose-cap-localization} is equivalent to
		$\log M=o(1/\tau)$, this tail contribution is $o(\tau)$.  Thus
		\begin{equation}
			\E_1|d(X_\tau)|\geq2c\tau-o(\tau)>c\tau
			\label{eq:app-rent-dispersion}
		\end{equation}
		for all sufficiently large $n$.
		
		Finally, subtracting
		$V^{\NR}(T;c,M)=(1-c)T-R_M(T)$ from
		\eqref{eq:app-review-distance-bound} yields
		\begin{equation*}
			V^R(\tau)-V^{\NR}(T;c,M)
			\leq c\tau+R_M(T)-R_M(T-\tau)-\E_1|d(X_\tau)|.
		\end{equation*}
		For large $n$, $T-\tau_n\in[T/2,T)$ and the monotonicity established after
		\eqref{eq:app-loose-cap-rent-derivative} gives
		$R_M(T-\tau)>R_M(T)$.  Combining this strict inequality with
		\eqref{eq:app-rent-dispersion} makes the displayed difference negative,
		contradicting \eqref{eq:app-loose-cap-contradiction}.  No such sequence
		exists, so a finite uniform $\overline M(T,c)$ exists.
	\end{proof}

	\begin{proof}[Proof of Proposition~\ref{prop:tight-cap-review-gain}]
		By Lemma~\ref{lem:review-boundary-values}, $G(\tau)\to0$ as $\tau\uparrow T$, and
		$V_-^{R\prime}(T)$ in \eqref{eq:one-sided-derivative-definitions} is exactly the left
		derivative of $G$ at $T$:
		\begin{equation*}
			G(\tau)=V_-^{R\prime}(T)\,(\tau-T)+o(T-\tau)
			\qquad\text{as }\tau\uparrow T .
		\end{equation*}
		Under the stated hypotheses, Corollary~\ref{cor:tight-cap-neighborhood} gives
		$V_-^{R\prime}(T)<0$.  Since $\tau-T<0$, the leading term is strictly positive, so
		there is $\delta>0$ with $G(\tau)>0$ for every $\tau\in(T-\delta,T)$, which is the
		claim.
	\end{proof}
	
	The same display gives the converse: if $V_-^{R\prime}(T)>0$ then $G(\tau)<0$ throughout a
	left neighborhood of $T$.
	
	\begin{proof}[Proof of Proposition~\ref{prop:review-restores-feasibility-boundary}]
		At the boundary,
		$\overline\Delta(T)/T=c/M_0$.
		Since $\overline\Delta(t)/t$ is strictly decreasing by
		Lemma~\ref{lem:Mbar}, every $t<T$ satisfies
		$\overline\Delta(t)/t>c/M_0$,
		which gives $M_0\overline\Delta(\tau)>c\tau$ and $M_0\overline\Delta(h)>ch$.
		
		To obtain the second claim, choose $\tau_0>0$ small and write
		$h_0=T-\tau_0$.  At $M=M_0$ the continuation problem at $h_0<T$ is
		strictly feasible, so its minimum work-inducing rent $r_D(h_0)$ is
		well defined and strictly below $M_0$.  At the boundary $h=T$, the two
		threshold roots coincide and the corresponding rent is
		$M_0\Phi(-\sqrt T/2)$; hence the available rent spread converges to
		$M_0\Phi(\sqrt T/2)>0$ as $h_0\uparrow T$.  Since
		$\overline\Delta(\tau)/\tau\to\infty$ as $\tau\downarrow0$, $\tau_0$ can
		therefore be chosen small enough that
		\begin{equation}
			\bigl[M_0-r_D(h_0)\bigr]\overline\Delta(\tau_0)>c\tau_0.
			\label{eq:review-restores-pre-review-slack}
		\end{equation}
		Both continuation feasibility at $h_0$ and
		\eqref{eq:review-restores-pre-review-slack} are strict, so by continuity
		they continue to hold for all caps in a sufficiently small interval
		$(M_0-\varepsilon,M_0]$.
		
		For such a cap, use the minimum-rent full-effort continuation contract
		after review scores $X_{\tau_0}<\tau_0/2$, and use the good-stop action
		with rent $M$ after $X_{\tau_0}\geq\tau_0/2$.  The spread in continuation
		utility between the two regions is $M-r_D(h_0)$, and the midpoint cutoff
		maximizes the difference between the work and shirk probabilities.  The
		date-$0$ incentive spread is therefore
		$\bigl[M-r_D(h_0)\bigr]\overline\Delta(\tau_0)>c\tau_0$,
		so first-phase work is induced.  In the lower region the continuation
		contract induces full effort, while in the upper region the project is
		stopped.  The lower region has strictly positive probability under work.
		Finally, $M<M_0=M_I(T,c)$ makes full effort infeasible without
		review by the implementability condition in Lemma~\ref{lem:Mbar}.
	\end{proof}

	\clearpage
	\setcounter{section}{0}
	\setcounter{subsection}{0}
	\setcounter{lemma}{0}
	\setcounter{proposition}{0}
	\setcounter{corollary}{0}
	\setcounter{remark}{0}
	\setcounter{equation}{0}
	\setcounter{page}{1}
	\renewcommand{\thepage}{OA-\arabic{page}}
	
	\counterwithout{equation}{section}
	\counterwithout{lemma}{section}
	\counterwithout{proposition}{section}
	\counterwithout{corollary}{section}
	\counterwithout{remark}{section}
	
	\renewcommand{\thesection}{OA.\arabic{section}}
	
	\renewcommand{\thesubsection}{OA.\arabic{section}.\arabic{subsection}}
	\renewcommand{\thelemma}{OA.\arabic{lemma}}
	\renewcommand{\theproposition}{OA.\arabic{proposition}}
	\renewcommand{\thecorollary}{OA.\arabic{corollary}}
	\renewcommand{\theremark}{OA.\arabic{remark}}
	\renewcommand{\theequation}{OA.\arabic{equation}}
	\providecommand{\theHsection}{}
	\providecommand{\theHsubsection}{}
	\providecommand{\theHlemma}{}
	\providecommand{\theHproposition}{}
	\providecommand{\theHcorollary}{}
	\providecommand{\theHremark}{}
	\providecommand{\theHequation}{}
	\renewcommand{\theHsection}{OA.\arabic{section}}
	\renewcommand{\theHsubsection}{OA.\arabic{section}.\arabic{subsection}}
	\renewcommand{\theHlemma}{OA.\arabic{lemma}}
	\renewcommand{\theHproposition}{OA.\arabic{proposition}}
	\renewcommand{\theHcorollary}{OA.\arabic{corollary}}
	\renewcommand{\theHremark}{OA.\arabic{remark}}
	\renewcommand{\theHequation}{OA.\arabic{equation}}
	
	\begin{center}
		{ \huge \textbf{Online Appendix for ``Midterm Review''}}
		
		\vspace*{20pt}
		{\large Doruk Cetemen \qquad
			Yonggyun Kim  \qquad
			Fei Li		   \qquad
			Curtis R. Taylor}
	\end{center}
	
	\addtocontents{toc}{\protect\begingroup\protect\setcounter{tocdepth}{-1}}
	
	\section{Pointwise effort choice}
	\label{oa:flexible-effort}
	
	Section~\ref{sec:pointwise-effort} states the two conclusions of this extension: no interior
	effort level becomes an exposed continuation action, and the easy continuation standard must be
	raised from $z_E(h)$ to $\zeta(h)$, lowering the maximum continuation rent from $r_E(h)$ to
	$\widetilde r_E(h)$.  This appendix supplies the argument.
	
	\subsection{Flexible effort and a relaxed feasible set}
	\label{oa:sec:flexible-setup}
	
	Fix a continuation phase of length $h$ and let total effort be $\eta\in[0,h]$, so that a plan with total effort $\eta$ produces $Y\sim N(\eta,h)$ and costs the agent $c\eta$.  For a terminal wage $w$ write $\pi(\eta)=\E_\eta[w]=\int_{\mathbb R}w(y)f_\eta(y)\,dy$ for the agent's expected payment when he supplies total effort $\eta$, and let $\mathcal F_h^{\mathrm{flex}}(\eta)$ denote the contracts satisfying \eqref{eq:pointwise-full-IC}.  Proofs are collected in Section~\ref{oa:sec:flexible-proofs}.
	
	The economic difference from phase-level effort is that the old model asks only whether the
	agent wants to work rather than not work, whereas the flexible model also asks whether he wants
	to work \emph{all the way to the target}.  Those nearby deviations are what tighten the easy
	standard below.
	
	\paragraph{A relaxed feasible set.}
	
	Condition~\eqref{eq:pointwise-full-IC} contains a continuum of deviations, and we will not need
	to characterize all of them.  Three necessary conditions suffice.  Two are already in the main
	text: \eqref{eq:pointwise-global}, the comparison with abandoning effort altogether, which is
	exactly the comparison used in the phase-level model; and \eqref{eq:pointwise-local}, which,
	because $\pi'(\eta)=\Cov_\eta(w(Y),Y)/h$, says that the wage must load enough on output to pay
	for the \emph{marginal} unit of effort.  The second is the new restriction that matters for the
	easy standard.  Third, at an interior target the agent's payoff must bend the right way:
	\begin{equation*}
		\pi''(\eta)\leq0
		\qquad\text{if }\eta\in(0,h).
		\tag{S$_\eta$}
		\label{oa:eq:flexible-soc}
	\end{equation*}
	
	Let $\mathcal F_h(\eta)$ impose only \eqref{eq:pointwise-global}, and let
	$\widehat{\mathcal F}_h(\eta)$ impose the global, local, and second-order conditions.  Then
	\begin{equation}
		\mathcal F_h^{\mathrm{flex}}(\eta)
		\ \subseteq\
		\widehat{\mathcal F}_h(\eta)
		\ \subseteq\
		\mathcal F_h(\eta).
		\label{oa:eq:relaxation-inclusion}
	\end{equation}
	The rightmost set is the old binary-effort benchmark.  The middle set is an analytically useful
	outer bound on the true flexible-effort set.  The strategy below is to show that even this larger
	set has no economically relevant extreme points beyond the four actions we identify.  Since the
	four candidate vertices themselves satisfy the full incentive constraint, this sandwich argument
	is enough; we never have to solve the entire continuum-deviation problem at interior effort.
	
	\paragraph{Rent and surplus.}
	
	As in the baseline model, an action is the pair $(r,s)$ of continuation rent and \emph{total}
	continuation surplus, the principal's own payoff being $u=s-r$.  If the contract implements
	$\eta$,
	\begin{equation}
		r=\pi(\eta)-c\eta,
		\qquad
		s=(1-c)\eta,
		\qquad
		u=s-r=\eta-\pi(\eta).
		\label{oa:eq:flexible-rent-surplus}
	\end{equation}
	For a fixed effort target, feasible contracts therefore lie on a horizontal segment at height
	$(1-c)\eta$, along which the principal's payoff falls one for one with the rent conceded.
	Flexible effort fills in all the segments corresponding to $\eta\in[0,h]$:
	\begin{equation}
		\mathcal A^{\mathrm{flex}}(h)
		=\bigcup_{\eta\in[0,h]}
		\bigl\{\bigl(P-c\eta,\ (1-c)\eta\bigr):\ P\in[P_{\min}(\eta),P_{\max}(\eta)]\bigr\},
		\label{oa:eq:flexible-action-set}
	\end{equation}
	where $P_{\min}(\eta)$ and $P_{\max}(\eta)$ are the least and greatest expected payments that
	can implement $\eta$.  Define $\widehat P_{\min}$ and $\widehat P_{\max}$ analogously using the
	relaxed set $\widehat{\mathcal F}_h(\eta)$.
	
	The picture is therefore richer than in the baseline model: instead of two line segments, the
	principal faces a two-dimensional body.  The key question, however, is not how large that body
	is.  It is whether the body creates new \emph{extreme} rent--surplus pairs.  Since the review
	problem maximizes a linear function of $(r,s)$ state by state, only those extreme points can ever
	matter for the optimal menu.
	
	\subsection{The three regimes}
	\label{oa:sec:regimes}

	There are three economically distinct targets: partial effort, full effort, and no effort.
	Partial effort is the new case.  Full effort turns out to look almost exactly like the baseline
	model, while no effort is unchanged.
	
	\paragraph{Interior effort: why a band replaces a threshold.}
	Suppose the principal wants to implement $\eta\in(0,h)$.  In the baseline model an upper-tail
	bonus is attractive because high output is evidence of effort.  With an interior target that
	logic is incomplete: a very high outcome is also evidence that the agent supplied \emph{more}
	than the target.  A contract that keeps paying on arbitrarily high outcomes therefore encourages
	the wrong upward deviation.
	
	The least-cost way to induce an interior target consequently rewards a \emph{range} of good
	outcomes rather than the entire upper tail.  The lower edge of the range makes effort attractive;
	the upper edge prevents the agent from wanting to overshoot the target.  Formally, matching both
	the expected payment and the marginal incentive produces the following band-reduction result.
	
	\begin{lemma}
		\label{oa:lem:band-reduction}
		Fix $\eta\in(0,h)$ and let $w:\mathbb R\to[0,M]$ satisfy $\pi'_w(\eta)=c$.  Then there is a
		band bonus
		\begin{equation}
			\tilde w(y)=M\ind\{k_1\leq y\leq k_2\},
			\qquad -\infty\leq k_1<k_2\leq+\infty,
			\label{oa:eq:band-contract}
		\end{equation}
		with the same expected payment and the same marginal incentive at the target,
		\begin{equation*}
			\pi_{\tilde w}(\eta)=\pi_w(\eta),
			\qquad
			\pi'_{\tilde w}(\eta)=c,
			\qquad\text{and}\qquad
			\pi_{\tilde w}(0)\leq\pi_w(0).
		\end{equation*}
		Consequently $w\in\widehat{\mathcal F}_h(\eta)$ implies $\tilde w\in\widehat{\mathcal F}_h(\eta)$,
		so the range of $\pi(\eta)$ over $\widehat{\mathcal F}_h(\eta)$ is unchanged if attention is
		restricted to band bonuses.
	\end{lemma}
	
	The economic content is simple.  The principal still wants to spend the payment cap where output
	is informative about effort.  But she now has to control incentives on both sides of the target.
	The optimal payment region is therefore bounded above as well as below: outcomes that are
	sufficiently good earn the bonus, while extremely good outcomes do not.  The upper cutoff is not
	a punishment for success; it is what prevents the contract from rewarding effort beyond the
	amount the principal wants to buy.
	
	Lemma~\ref{oa:lem:band-reduction} reduces the problem to two numbers.  Standardize the band by
	\begin{equation*}
		a=\frac{k_1-\eta}{\sqrt h},
		\qquad
		b=\frac{k_2-\eta}{\sqrt h},
		\qquad
		\kappa=\frac{c\sqrt h}{M},
	\end{equation*}
	so that $\pi(\eta)=M[\Phi(b)-\Phi(a)]$ and $\pi'(\eta)=M[\phi(a)-\phi(b)]/\sqrt h$.  The interior
	condition in \eqref{eq:pointwise-local} then reads
	\begin{equation}
		\phi(a)-\phi(b)=\kappa,
		\label{oa:eq:band-FOC}
	\end{equation}
	a single equation in $(a,b)$ that does not involve $\eta$.  Let $\zeta(h)>0$ be the bar of
	Section~\ref{sec:pointwise-effort}, the positive root of $\phi(\zeta)=\kappa$, equivalently of
	$M\phi(\zeta)/\sqrt h=c$.
	Along the curve \eqref{oa:eq:band-FOC} the payment $M[\Phi(b)-\Phi(a)]$ is largest at $a=-\zeta(h)$,
	$b=+\infty$: the band degenerates to an upper-tail bonus with bar $\eta-\zeta(h)\sqrt h$, and the
	payment it delivers,
	\begin{equation}
		p^{*}=M\,\Phi\bigl(\zeta(h)\bigr),
		\label{oa:eq:Pmax-flexible}
	\end{equation}
	uses \eqref{eq:pointwise-local} alone and does not vary with $\eta$.  The bound is attained, and the
	slice is an interval.
	
	\begin{lemma}
		\label{oa:lem:interior-slice}
		Suppose $\Delta\bigl(h,\zeta(h)\bigr)\geq ch/M$.  Then for every $\eta\in(0,h)$ the expected
		payments achievable at $\eta$ form the interval $\bigl[P_{\min}(\eta),\,p^{*}\bigr]$.  The upper
		end is attained by the upper-tail bonus with bar $\eta-\zeta(h)\sqrt h$, which lies in
		$\mathcal F_h^{\mathrm{flex}}(\eta)$ itself and not merely in the relaxation.  The lower end is
		$P_{\min}(\eta)=M[\Phi(b)-\Phi(a)]$ at the solution of \eqref{oa:eq:band-FOC} together with
		\eqref{eq:pointwise-global} at equality: writing $\theta=\eta/\sqrt h$,
		\begin{equation}
			\bigl[\Phi(b)-\Phi(a)\bigr]-\bigl[\Phi(b+\theta)-\Phi(a+\theta)\bigr]=\kappa\theta .
			\label{oa:eq:Pmin-flexible}
		\end{equation}
	\end{lemma}
	
	Because the upper end does not vary with $\eta$, the upper boundary of
	$\mathcal A^{\mathrm{flex}}(h)$ is a straight line.
	Because the maximum expected payment $p^{*}$ is the same for every interior target, the
	high-rent side of the feasible set has a particularly simple interpretation.  Along that side,
	the principal is already paying as much as the local incentive condition permits.  Increasing
	the effort target therefore means asking the agent to work more without increasing expected
	compensation.  Substituting $P=p^{*}$ into \eqref{oa:eq:flexible-rent-surplus} gives
	$(r,s)=\bigl(p^{*}-c\eta,\ (1-c)\eta\bigr)$ for $\eta\in[0,h]$, which runs from the no-effort face to
	\begin{equation}
		\widetilde A_E=\bigl(p^{*}-ch,\ (1-c)h\bigr)
		\label{oa:eq:AEtilde}
	\end{equation}
	at full effort.
	
	This already suggests why partial effort will not survive in the final menu.  Every interior
	point on this high-rent locus can be replicated by mixing the full-effort endpoint
	$\widetilde A_E$ with no-effort actions.  The only candidate new extreme action is therefore the
	full-effort endpoint itself.  Economically, flexible effort does not create a new kind of reward;
	it merely limits how generous the easy full-effort reward can be.
	
	\paragraph{Full effort.}
	At full effort the agent cannot overshoot the target.  The relevant contract is therefore again
	an upper-tail bonus, exactly as in the baseline model.  What changes is that full effort must now
	beat \emph{every} lower effort level, not just zero effort.  For an upper-tail bonus this creates
	two transparent tests: a global test comparing full effort with zero effort and a marginal test
	asking whether the last unit of effort is worth its cost.  Remarkably, for the Gaussian
	threshold contract these two tests are jointly sufficient.
	
	\begin{lemma}
		\label{oa:lem:face-exact}
		Let $w=M\ind\{y\geq h-z\sqrt h\}$.  Then $w$ implements $\eta=h$ in the sense of
		\eqref{eq:pointwise-full-IC} if and only if
		\begin{equation*}
			\Delta(h,z)\geq\frac{ch}{M}
			\qquad\text{and}\qquad
			\phi(z)\geq\kappa .
		\end{equation*}
	\end{lemma}
	
	The two inequalities measure different ways in which the same standard can fail.  The first asks
	whether the standard generates enough incentive \emph{on average} to make full effort preferable
	to giving up altogether.  The second asks whether it generates enough incentive \emph{at the
		margin} to stop the agent from shaving the last bit of effort.  Phase-level effort checks only
	the first.  Flexible effort checks both.
	
	\begin{lemma}
		\label{oa:lem:flexible-full-effort-face}
		Write $\Lambda(h,z)=\sqrt h\,\phi(z)$, let $z^{\dagger}(h)$ be the unique bar at which
		$\Delta(h,z)=\Lambda(h,z)$ --- it lies in $\bigl(0,\sqrt h/2\bigr)$ --- and set
		\begin{equation}
			\overline{\mathcal M}(h)
			=\max_{z}\ \min\bigl\{\Delta(h,z),\ \Lambda(h,z)\bigr\}
			=\sqrt h\,\phi\bigl(z^{\dagger}(h)\bigr).
			\label{oa:eq:Mbar-definition}
		\end{equation}
		\begin{enumerate}
			\item[(i)] Full effort is implementable under flexible effort if and only if
			$\overline{\mathcal M}(h)\geq ch/M$; and $\overline{\mathcal M}(h)<\overline\Delta(h)$ for
			every $h>0$, so flexible effort always tightens feasibility.
			\item[(ii)] When that inequality is strict, the full-effort slice of
			$\mathcal A^{\mathrm{flex}}(h)$ is the segment of the line $s=(1-c)h$ joining $A_D$ to
			$\widetilde A_E$.  \emph{The difficult end is unchanged}: it is borne by the bar $z_D(h)$
			of Section~\ref{sec:threshold-contracts}, at which the global condition binds and the local one
			is slack, so
			\begin{equation*}
				r_D(h)=M\,\Phi\bigl(z_D(h)-\sqrt h\bigr),
				\qquad
				\pi'(h)=\frac{M\,\phi\bigl(z_D(h)\bigr)}{\sqrt h}>c .
			\end{equation*}
			\emph{The easy end changes}: it is borne by the strictly harder bar $\zeta(h)<z_E(h)$, at
			which the local condition binds and the global one is slack, so
			\begin{equation}
				\widetilde r_E(h)=M\,\Phi\bigl(\zeta(h)\bigr)-ch\;<\;r_E(h),
				\qquad
				\pi'(h)=c .
				\label{oa:eq:new-max-rent}
			\end{equation}
		\end{enumerate}
	\end{lemma}
	
	Which of the two conditions of Lemma~\ref{oa:lem:face-exact} binds is what makes the correction
	asymmetric, for the reason given in Section~\ref{sec:pointwise-effort}: the difficult standard is
	already sensitive to effort at the margin, whereas the easy one is not.  The magnitudes are worth
	recording.  Feasibility itself changes only modestly:
	$\overline\Delta(h)/\overline{\mathcal M}(h)=1+h/72+O(h^{2})$.  The larger effect is on the rent
	that the easy action can promise.  For example, at $h=1/4$ and $c=0.3$,
	$\widetilde r_E/r_E$ ranges from about $0.84$ near the feasibility floor to $0.98$ when the cap
	is four times that floor; at $h=1$ the ratio can fall to about $0.68$ near the tight end.
	Flexible effort therefore matters most exactly where continuation rents are scarce and valuable
	as incentive capital.
	
	\paragraph{No effort.}
	Nothing changes at the no-effort end.  A constant payment $w\equiv p$ gives the same expected
	payment at every effort level, so exerting positive effort only adds cost.  Hence every
	$p\in[0,M]$ implements zero effort and delivers
	$(r,s)=(p,0)$.
	The entire old no-effort segment remains feasible, including
	$A_B=(0,0)$ and $A_G=(M,0)$.
	
	\subsection{Extreme points}
	\label{oa:sec:flexible-extreme}

	The review problem cares not about all feasible contracts but about the convex hull of feasible
	rent--surplus pairs, which is what Figure~\ref{fig:pointwise-feasible-set} draws.  This section
	supplies the characterization behind that figure: the relaxed pointwise-effort set lies inside
	$\operatorname{co}\{A_B,A_D,\widetilde A_E,A_G\}$, with $\widetilde A_E$ as in
	\eqref{oa:eq:AEtilde} and $A_B$, $A_D$, $A_G$ as in Section~\ref{sec:extreme-points}, and each of
	those four vertices is implementable.  The one non-obvious step is the low-rent boundary: partial
	effort never delivers effort more cheaply in rent-per-unit terms than full effort.
	
	\paragraph{The residual question, in rent terms.}
	Could the principal ever prefer to buy \emph{some} continuation effort rather than the full
	amount, because partial effort requires less rent per unit?  Let $\rho(\eta)$ denote the rent under the cheapest relaxed contract implementing $\eta$:
	\begin{equation}
		\rho(\eta)=\widehat P_{\min}(\eta)-c\eta,
		\qquad\text{with}\qquad
		\rho(h)=M\Phi\bigl(z_D(h)\bigr)-ch=r_D(h).
		\label{oa:eq:flexible-rent-function}
	\end{equation}
	The relevant chord inequality is equivalent to
	\begin{equation}
		\frac{\rho(\eta)}{\eta}\ \geq\
		\frac{\rho(h)}{h}=\frac{r_D(h)}{h}
		\qquad\text{for every }\eta\in(0,h).
		\label{oa:eq:rent-per-unit}
	\end{equation}
	So the issue has a direct economic interpretation: is full effort the cheapest way to buy effort
	in rent terms?  The answer is yes.  Rent per unit falls as the principal buys more effort, so a
	partial-effort continuation is never an exposed action.
	
	The normalization used in the proof makes this a one-dimensional statement.  Writing
	$\theta=\eta/\sqrt h$ and $\kappa=c\sqrt h/M$, the minimum rent takes the form
	\begin{equation}
		\rho(\eta)=M\,\mathcal R_\kappa\bigl(\eta/\sqrt h\bigr),
		\qquad
		\mathcal R_\kappa(\theta)=\Phi(b+\theta)-\Phi(a+\theta),
		\label{oa:eq:rent-normalised}
	\end{equation}
	where $(a,b)$ solve the band conditions.  The next lemma establishes the desired economies of
	scale in incentive rent.
	
	\begin{lemma}
		\label{oa:lem:rent-per-unit}
		Suppose $\overline{\mathcal M}(h)>ch/M$.  Then
		\begin{equation*}
			\widehat P_{\min}(\eta)\ \geq\ \frac{\eta}{h}\,\widehat P_{\min}(h)
			\qquad\text{for every }\eta\in(0,h);
		\end{equation*}
		equivalently, by \eqref{oa:eq:flexible-rent-function}, the agent's rent per unit of effort
		implemented, $\rho(\eta)/\eta$, is smallest at full effort.
	\end{lemma}
	
	\begin{proposition}
		\label{oa:prop:flexible-extreme-points}
		Suppose $\overline{\mathcal M}(h)>ch/M$.  Then
		\begin{equation*}
			\operatorname{co}\mathcal A^{\mathrm{flex}}(h)
			=
			\operatorname{co}\bigl\{A_B,\ A_D,\ \widetilde A_E,\ A_G\bigr\},
		\end{equation*}
		and all four vertices belong to $\mathcal A^{\mathrm{flex}}(h)$.  In particular, for every
		$q\in\mathbb R$ a maximizer of $s+qr$ over $\mathcal A^{\mathrm{flex}}(h)$ can be chosen from
		these four actions, and no interior effort level is ever used.
	\end{proposition}
	
	This is the payoff from the preceding machinery.  Although flexible effort creates a continuum
	of possible effort levels and a two-dimensional feasible set, none of the interior effort levels
	is ever chosen by the principal.  The economically relevant set still has four vertices.  Three
	are exactly the baseline actions, and the fourth is the same easy continuation with a slightly
	harder standard.
	
	The proof deliberately uses the relaxed set only as an outer bound.  We therefore do not need a
	complete characterization of incentive compatibility at every interior target.
	
	\subsection{Consequences for the analysis}
	
	Proposition~\ref{oa:prop:flexible-extreme-points} is what licenses the substitution
	$r_E(h)\mapsto\widetilde r_E(h)$ announced in Section~\ref{sec:pointwise-effort}.  Only the
	upper breakpoint of \eqref{eq:three-breakpoints} moves, to
	$\widetilde q_2(h)=(1-c)h/\bigl(M-\widetilde r_E(h)\bigr)<q_2(h)$, so the good-stop region
	opens at a lower shadow value than in the baseline, while the cutoff formula
	\eqref{eq:k-formula} is unchanged and the ordered four-region menu of
	Proposition~\ref{prop:four-region-review-menu} survives verbatim.  Feasibility of full
	continuation effort is now the strictly stronger requirement
	$\overline{\mathcal M}(h)\geq ch/M$ of Lemma~\ref{oa:lem:flexible-full-effort-face}(i), in place
	of $\overline\Delta(h)\geq ch/M$.
	
	\subsection{Flexible effort before the review}
	\label{oa:sec:flexible-pre-review}

	So far we have allowed the agent to coast only after the review.  The same concern can arise
	before the review: an agent who is supposed to work until $\tau$ may instead work for only part
	of that interval.  The effect is again local and economically transparent.  Because the desired
	first-phase effort is the upper boundary $\eta_1=\tau$, the principal does not need to make the
	agent indifferent at the margin; she only needs the marginal return to the last bit of effort to
	be at least its cost.
	
	\paragraph{Only one new constraint.}
	Let $\eta_1\in[0,\tau]$ denote total first-phase effort.  The review state is
	$X_\tau\sim N(\eta_1,\tau)$, and the menu promises continuation utility $r(x)$ after review state
	$x$.  Implementing full first-phase effort requires
	$\E_\tau[r(X_\tau)]-c\tau\geq\E_{\eta_1'}[r(X_\tau)]-c\eta_1'$ for every
	$\eta_1'\in[0,\tau]$.
	The baseline model checks only the most distant deviation, $\eta_1'=0$.  Flexible effort adds
	the possibility of shaving a small amount of first-phase effort.  At the full-effort corner this
	requires
	\begin{equation*}
		\frac{d}{d\eta_1}\E_{\eta_1}\bigl[r(X_\tau)\bigr]\bigg|_{\eta_1=\tau}
		=\frac{1}{\tau}\Cov_\tau\bigl(r(X_\tau),X_\tau\bigr)\ \geq\ c .
	\end{equation*}
	Thus the review menu must provide enough continuation-utility sensitivity not only in total, but
	also at the margin.
	
	\paragraph{The two conditions in the menu's cutoffs.}
	Because the continuation problem still has four extreme actions, the promised-rent schedule is
	an increasing step function,
	\begin{equation*}
		r(x)=\sum_{i=1}^{3}\delta_i\ind\{x\geq k_i\},
		\qquad k_1<k_2<k_3,
	\end{equation*}
	with jumps
	$\delta_1=r_D(h)$,
	$\delta_2=\widetilde r_E(h)-r_D(h)$, and
	$\delta_3=M-\widetilde r_E(h)$.
	Let $z_i=(\tau-k_i)/\sqrt\tau$.  Then the global and marginal first-phase incentive requirements
	are
	\begin{equation*}
		\underbrace{\sum_i\delta_i\,\Delta(\tau,z_i)\ \geq\ c\tau}_{\text{global}},
		\qquad
		\underbrace{\sum_i\delta_i\,\Lambda(\tau,z_i)\ \geq\ c\tau}_{\text{local}} .
	\end{equation*}
	The first asks whether the menu makes full effort preferable to zero effort.  The second asks
	whether the last increment of first-phase effort is worth taking.  This is exactly the same
	average-versus-marginal distinction that hardened the easy continuation standard.
	
	\paragraph{The ordered menu survives.}
	Attach multipliers $\lambda\geq0$ and $\mu\geq0$ to the global and local first-phase
	constraints.  The review-state index becomes
	\begin{equation}
		Q(x)=\lambda\left[1-\exp\left(\frac{\tau}{2}-x\right)\right]
		+\mu\,\frac{x-\tau}{\tau}-1,
		\qquad
		Q'(x)=\lambda\exp\left(\frac{\tau}{2}-x\right)+\frac{\mu}{\tau}>0 .
		\label{oa:eq:Q-two-multiplier}
	\end{equation}
	The important fact is not the extra term itself but its monotonicity.  Better review performance
	still raises the attractiveness of higher-rent actions.  Hence the statewise problem still
	selects among the same four actions in the same order.  Flexible first-phase effort changes the
	cutoffs and introduces a second multiplier, but it does not destroy the ordered menu.
	
	\paragraph{One prediction does change.}
	There is one notable implication.  In the baseline model the review-state index is bounded above,
	so the good-stop region appears only for some parameters.  If the marginal first-phase
	constraint binds, then $\mu>0$ and the linear term in \eqref{oa:eq:Q-two-multiplier} makes $Q(x)$
	unbounded above.  The good-stop region must then be nonempty.
	
	Economically, marginal first-phase incentives are especially well supplied by very high review
	outcomes.  A low termination threshold can create a large difference between working and not
	working while doing little to reward the \emph{last} increment of effort.  A top prize does the
	opposite: its covariance with performance is large in the upper tail.  Freeing first-phase
	effort therefore shifts incentive provision toward the good-stop prize and makes the full
	four-region ladder generic whenever the local constraint actually bites.
	
	\paragraph{When nothing changes.}
	The extra local constraint need not bind.  Since
	$\Lambda(\tau,z)\geq\Delta(\tau,z)$ exactly when
	$z\leq z^{\dagger}(\tau)$, the global condition automatically implies the local one whenever
	all cutoffs satisfy
	$k_i\geq k^{\dagger}(\tau)=\tau-z^{\dagger}(\tau)\sqrt\tau\in(\tau/2,\tau)$.
	Because the cutoffs are ordered, it is enough to check the lowest one.  Thus if the termination
	threshold is already sufficiently demanding---roughly above half of expected first-phase
	output---allowing the agent to coast before the review changes nothing.
	
	\paragraph{A step that does not carry over.}
	There is one limitation.  For a single continuation threshold, the global and local conditions
	are jointly sufficient because the agent's loss from deviating is unimodal.  Before the review
	the rent schedule can have three jumps, so the marginal-return function is a Gaussian mixture:
	\begin{equation*}
		\frac{d}{d\eta_1}\E_{\eta_1}\bigl[r(X_\tau)\bigr]
		=\frac{1}{\sqrt\tau}\sum_i\delta_i\,
		\phi\!\left(\frac{k_i-\eta_1}{\sqrt\tau}\right).
	\end{equation*}
	With several peaks, an intermediate effort level can be attractive even when both the global
	and endpoint-marginal tests pass.
	
	This qualification appears to matter only for long first phases.  Numerically the mixture is
	single-peaked on $[0,\tau]$ for every cutoff configuration we tried at $\tau\leq4$; a second
	peak appears by $\tau=9$ and a third by $\tau=64$.  Thus for short first phases the two
	conditions above are exact in the cases examined, whereas late reviews require direct checks of
	intermediate deviations.
	
	\paragraph{A margin left open.}
	Finally, this section asks whether the principal can still implement full first-phase effort.
	It does not ask whether she would optimally choose less.  Allowing the principal to target an
	interior first-phase effort would introduce a first-order equality before the review as well and
	would put the band-contract logic on both sides of the review.  We leave that additional margin
	for future work.

	\subsection{Proofs}
	\label{oa:sec:flexible-proofs}

	\subsubsection{Proof of Lemma~\ref{oa:lem:band-reduction}}
	
	\begin{proof}
		Consider minimizing $\pi(0)=\int wf_0$ over $0\leq w\leq M$ subject to the two equality
		constraints $\int wf_\eta=\pi_w(\eta)$ and $\int w\,\partial_\eta f_\eta=c$.  The feasible set
		contains $w$ and is the intersection of the weak-$*$ compact ball with two weak-$*$ closed
		hyperplanes, hence weak-$*$ compact; the objective is weak-$*$ continuous, so a minimizer
		$\tilde w$ exists.  Attaching multipliers $\lambda_1,\lambda_2$ to the two constraints gives a
		Lagrangian linear in $w$, which is therefore minimized pointwise over $[0,M]$, with
		$\tilde w=M$ exactly where the index is negative.  Dividing that index by $f_\eta(y)>0$ and
		using
		\begin{equation*}
			\frac{f_0(y)}{f_\eta(y)}=\exp\left(\frac{\eta^{2}}{2h}-\frac{\eta y}{h}\right),
			\qquad
			\frac{\partial_\eta f_\eta(y)}{f_\eta(y)}=\frac{y-\eta}{h},
		\end{equation*}
		it becomes
		\begin{equation*}
			\Psi(y)=\exp\left(\frac{\eta^{2}}{2h}-\frac{\eta y}{h}\right)
			-\lambda_1-\frac{\lambda_2}{h}(y-\eta),
		\end{equation*}
		a strictly convex function of $y$ less an affine one, hence strictly convex.  Its negativity
		set is an interval, which is \eqref{oa:eq:band-contract}, and since a strictly convex function
		vanishes at most twice the boundary is null, so no randomization is needed.  The first two
		displayed properties hold because $\tilde w$ is feasible for the auxiliary problem and the
		third because it is optimal.  Finally, if $w$ satisfies \eqref{eq:pointwise-global} then
		$\pi_{\tilde w}(\eta)-c\eta=\pi_w(\eta)-c\eta\geq\pi_w(0)\geq\pi_{\tilde w}(0)$, so
		$\tilde w$ satisfies it too.
	\end{proof}
	
	\subsubsection{Proof of Lemma~\ref{oa:lem:interior-slice}}
	
	\begin{proof}
		That the achievable payments form an interval with the stated endpoints is
		Lemma~\ref{oa:lem:band-reduction} together with linearity of $\pi(\eta)$ in $w$ on a convex set;
		the global condition binds at the lower end because lowering $\pi(\eta)$ tightens it.  For the
		upper end let $w^{*}=M\ind\{y\geq\eta-\zeta\sqrt h\}$ with $\zeta=\zeta(h)$ and let the agent
		choose $\eta'$, writing $v=(\eta'-\eta)/\sqrt h$.  Then $\pi(\eta')=M\Phi(\zeta+v)$ and, using
		$c\sqrt h=M\kappa$, his payoff is $M\Phi(\zeta+v)-c\eta-M\kappa v$, whose derivative in $v$ is
		$M[\phi(\zeta+v)-\kappa]$.  Since $\phi(\zeta)=\kappa$ and $\phi$ is symmetric and unimodal,
		this is positive exactly on $(-2\zeta,0)$ and negative outside $[-2\zeta,0]$.  The payoff is
		therefore maximized at $v=0$ or at the left endpoint $v=-\theta$, so \eqref{eq:pointwise-full-IC}
		reduces to the single comparison with $\eta'=0$, namely
		$D(\theta):=\Phi(\zeta)-\Phi(\zeta-\theta)-\kappa\theta\geq0$.  Now $D(0)=0$ and
		$D'(\theta)=\phi(\zeta-\theta)-\kappa$ is positive for $\theta\in(0,2\zeta)$ and negative for
		$\theta>2\zeta$, so $\{\theta\geq0:D(\theta)\geq0\}$ is an interval $[0,\bar\theta]$.  The
		hypothesis is exactly $D(\sqrt h)\geq0$, so $\bar\theta\geq\sqrt h$ and $w^{*}$ is incentive
		compatible at every $\eta\leq h$.
	\end{proof}
	
	\subsubsection{Proof of Lemma~\ref{oa:lem:face-exact}}
	
	\begin{proof}
		Write a deviation as $\eta'=h-u\sqrt h$ with $u\in[0,\sqrt h]$.  The agent's loss from it is
		$g(u)=M\bigl[\Phi(z)-\Phi(z-u)\bigr]-cu\sqrt h$, so \eqref{eq:pointwise-full-IC} says $g\geq0$ on
		$[0,\sqrt h]$.  Now $g(0)=0$ and $g'(u)=M\phi(z-u)-c\sqrt h$.  As $u$ rises, $z-u$ falls, so
		$\phi(z-u)$ rises while $z-u>0$ and falls thereafter: $g'$ is single-peaked.  Given
		$g'(0)\geq0$, which is $\phi(z)\geq\kappa$, it follows that $g'$ changes sign at most once and
		only from positive to negative, so $g$ is unimodal and attains its minimum over $[0,\sqrt h]$
		at an endpoint.  Since $g(0)=0$, $g\geq0$ if and only if $g(\sqrt h)\geq0$, which is
		$M\Delta(h,z)\geq ch$.  Conversely, if $\phi(z)<\kappa$ then $g'(0)<0$ and $g<0$ just to the
		right of zero.
	\end{proof}
	
	\subsubsection{Proof of Lemma~\ref{oa:lem:flexible-full-effort-face}}
	
	\begin{proof}
		By Lemma~\ref{oa:lem:face-exact} a bar is admissible if and only if
		$\min\{\Delta(h,z),\Lambda(h,z)\}\geq ch/M$, where $\Lambda(h,z)=\sqrt h\,\phi(z)$, so an
		admissible bar exists exactly when $\overline{\mathcal M}(h)\geq ch/M$.
		
		\emph{Where the maximum sits.}  Substituting $s\mapsto z-s$,
		$\Delta/\Lambda=h^{-1/2}\int_0^{\sqrt h}e^{zs-s^{2}/2}\,ds$, whose derivative in $z$ is
		$h^{-1/2}\int_0^{\sqrt h}s\,e^{zs-s^{2}/2}\,ds>0$; the limits are $0$ and $\infty$, so the two
		cross once, at some $z^{\dagger}(h)$.  At $z=0$ the ratio is $h^{-1/2}\int_0^{\sqrt h}e^{-s^2/2}ds<1$,
		so $z^{\dagger}>0$.  At $z=\sqrt h/2$, writing $a=\sqrt h/2$, the ratio exceeds one exactly when
		$2\Phi(a)-1>2a\phi(a)$; the difference vanishes at $a=0$ and has derivative $2a^{2}\phi(a)>0$,
		so $z^{\dagger}<\sqrt h/2$.  Hence $\min\{\Delta,\Lambda\}$ equals $\Delta$, increasing, to the
		left of $z^{\dagger}$ and $\Lambda$, decreasing on $z>0$, to its right, and the maximum in
		\eqref{oa:eq:Mbar-definition} is attained at $z^{\dagger}$.  Since $\Delta(h,\cdot)$ is strictly
		unimodal with peak $\overline\Delta(h)$ at $\sqrt h/2\neq z^{\dagger}$, the inequality
		$\overline{\mathcal M}(h)<\overline\Delta(h)$ is strict at every $h>0$.
		
		\emph{The two ends.}  The admissible bars are $z\in[z_D(h),\zeta(h)]$.  For the easy end,
		$z_E(h)>\sqrt h/2>z^{\dagger}(h)$ gives
		$\Lambda(h,z_E)<\Delta(h,z_E)=ch/M=\Lambda\bigl(h,\zeta(h)\bigr)$, and $\Lambda$ is decreasing
		on $z>0$, so $\zeta(h)<z_E(h)$.  For the difficult end,
		$\Delta(h,z_D)=ch/M\leq\Delta(h,z^{\dagger})$ with $\Delta(h,\cdot)$ increasing to the left of
		$\sqrt h/2$ gives $z_D(h)\leq z^{\dagger}(h)$, whence $\Lambda(h,z_D)\geq\Delta(h,z_D)=ch/M$:
		the local condition is slack there and $\pi'(h)=M\phi(z_D)/\sqrt h>c$.  The rent at an
		admissible bar is $M\Phi(z)-ch$, strictly increasing in $z$, so the two ends of the segment are
		the two ends of the interval of bars, which is \eqref{oa:eq:new-max-rent}.  At $z_D(h)$ the global
		constraint binds and the rent collapses to $M\Phi(z_D-\sqrt h)$ by
		\eqref{eq:threshold-moments}; at $\zeta(h)$ it does not, which is why the easy rent must be
		computed as $\E_h[w]-ch$ rather than as $\E_0[w]$.
	\end{proof}
	
	\subsubsection{Proof of Lemma~\ref{oa:lem:rent-per-unit}}
	
	\begin{proof}
		\emph{Step 1: the jump at the full-effort face.}  We first show
		$\widehat P_{\min}(h)\leq\liminf_{\eta\uparrow h}\widehat P_{\min}(\eta)$, equivalently
		$r_D(h)\leq M\,\mathcal R_\kappa(\sqrt h)$.  The local condition is an equality below the face
		and an inequality on it, so the constraint set at $\eta=h$ is the larger one.  Formally, take
		$\eta_n\uparrow h$ and let $w_n\in\widehat{\mathcal F}_h(\eta_n)$ attain
		$\widehat P_{\min}(\eta_n)$.  The ball $\{0\leq w\leq M\}$ is weak-$*$ compact, so a subnet
		converges to some $w^{*}$.  The maps $\eta\mapsto f_\eta$ and
		$\eta\mapsto\partial_\eta f_\eta$ are continuous in $L^1$, so the constraints pass to the
		limit: $\pi_{w^{*}}(h)-ch\geq\pi_{w^{*}}(0)$ and $\pi'_{w^{*}}(h)=c$, whence
		$\pi'_{w^{*}}(h)\geq c$ and $w^{*}\in\widehat{\mathcal F}_h(h)$.  Since
		$\pi_{w^{*}}(h)=\lim_n\widehat P_{\min}(\eta_n)$, the claim follows.  Equivalently and more
		directly: any solution of the pair \eqref{oa:eq:band-FOC} and \eqref{oa:eq:Pmin-flexible} at
		$\theta=\sqrt h$ satisfies \eqref{eq:pointwise-global} at $\eta=h$ by construction and
		\eqref{eq:pointwise-local} with equality, hence is admissible at full effort, so its payment
		bounds $\widehat P_{\min}(h)$ from above.
		
		\emph{Step 2: a closed form for $\mathcal R_\kappa'$.}  Parametrize the bands satisfying
		\eqref{oa:eq:band-FOC} by $a$, with $b=b(a)$ solving $\phi(b)=\phi(a)-\kappa$, and write
		$E_x=\phi(x+\theta)/\phi(x)=e^{-\theta x-\theta^{2}/2}$.
		Differentiating $\phi(b)=\phi(a)-\kappa$ and using $\phi'(x)=-x\phi(x)$ gives
		$b'(a)=a\phi(a)/\bigl(b\phi(b)\bigr)$.  Write $P(a)=\Phi(b)-\Phi(a)$,
		$\mathcal R(a,\theta)=\Phi(b+\theta)-\Phi(a+\theta)$ and
		$G(a,\theta)=P(a)-\mathcal R(a,\theta)-\kappa\theta$ for the slack in
		\eqref{eq:pointwise-global}.  Then
		\begin{equation*}
			P'(a)=\phi(b)b'(a)-\phi(a)=\phi(a)\,\frac{a-b}{b}\ <\ 0,
		\end{equation*}
		so minimizing the payment means taking the largest admissible $a$; call it $a^{*}(\theta)$,
		at which \eqref{eq:pointwise-global} binds.  Substituting $\phi(x+\theta)=\phi(x)E_x$,
		\begin{equation*}
			\mathcal R_a=\phi(a)\left(\frac ab E_b-E_a\right),
			\qquad
			\mathcal R_\theta=\phi(b)E_b-\phi(a)E_a,
			\qquad
			G_a=\frac{\phi(a)}{b}\Bigl[a(1-E_b)-b(1-E_a)\Bigr].
		\end{equation*}
		Differentiating $G(a^{*}(\theta),\theta)=0$ and
		$\mathcal R_\kappa(\theta)=\mathcal R(a^{*}(\theta),\theta)$, and using
		$G_\theta=-\mathcal R_\theta-\kappa$ and $G_a=P'-\mathcal R_a$,
		\begin{equation*}
			\mathcal R_\kappa'
			=\mathcal R_\theta-\mathcal R_a\frac{G_\theta}{G_a}
			=\frac{P'\mathcal R_\theta+\kappa\,\mathcal R_a}{P'-\mathcal R_a}.
		\end{equation*}
		Expanding the numerator and substituting $\kappa=\phi(a)-\phi(b)$, the terms in $E_a$ and
		$E_b$ collect as
		$P'\mathcal R_\theta+\kappa\mathcal R_a
		=\tfrac{\phi(a)}{b}\bigl(a\phi(a)-b\phi(b)\bigr)(E_b-E_a)$,
		so dividing by $G_a$,
		\begin{equation}
			\mathcal R_\kappa'(\theta)
			=
			\frac{\bigl(a\phi(a)-b\phi(b)\bigr)\bigl(E_b-E_a\bigr)}
			{a\,(1-E_b)-b\,(1-E_a)} .
			\label{oa:eq:rent-derivative}
		\end{equation}
		
		\emph{Step 3: signing it.}  Since $\phi(b)<\phi(a)$ we have $|b|>|a|$, which together with
		$b>a$ forces $b>0$; as $\theta>0$, $E_b<E_a$ and the second factor of the numerator is
		negative.  Differentiating $\pi(\eta')=M[\Phi(b+\theta')-\Phi(a+\theta')]$ twice at
		$\theta'=0$ gives $\pi''(\eta)=\tfrac{M}{h}\bigl(a\phi(a)-b\phi(b)\bigr)$, so the first factor
		is nonpositive precisely by the second-order condition \eqref{oa:eq:flexible-soc}.  Finally, the
		denominator equals $b\,G_a/\phi(a)$ with $\phi(a)>0$ and $b>0$; since $a^{*}$ is the largest
		$a$ at which $G\geq0$, and $G(a^{*},\theta)=0$, we have $G_a(a^{*},\theta)\leq0$.  The
		numerator of \eqref{oa:eq:rent-derivative} is therefore nonnegative and the denominator
		nonpositive, so $\mathcal R_\kappa'\leq0$.
		
		\emph{Conclusion.}  Since $\mathcal R_\kappa\geq0$ and $\mathcal R_\kappa'\leq0$,
		\begin{equation*}
			\frac{d}{d\theta}\left[\frac{\mathcal R_\kappa(\theta)}{\theta}\right]
			=\frac{\theta\,\mathcal R_\kappa'(\theta)-\mathcal R_\kappa(\theta)}{\theta^{2}}\ \leq\ 0 .
		\end{equation*}
		Hence for $\theta=\eta/\sqrt h<\sqrt h$, using \eqref{oa:eq:rent-normalised} and then Step~1,
		\begin{equation*}
			\frac{\rho(\eta)}{\eta}
			=\frac{M\,\mathcal R_\kappa(\theta)}{\theta\sqrt h}
			\ \geq\ \frac{M\,\mathcal R_\kappa(\sqrt h)}{h}
			\ \geq\ \frac{r_D(h)}{h}=\frac{\rho(h)}{h},
		\end{equation*}
		which is \eqref{oa:eq:rent-per-unit}.
	\end{proof}
	
	\begin{remark}
		\label{oa:rem:sign-conditions}
		The two signs invoked in Step~3 of the proof are economic rather than technical.  That
		$a\phi(a)\leq b\phi(b)$ is the agent's second-order condition: at the difficult end of the band
		family the contract would make the target a local \emph{minimum} of the agent's payoff, and
		such contracts implement nothing.  That $a(1-E_b)\leq b(1-E_a)$ says the payment-minimizing
		band sits on the boundary of the region where the agent still prefers the target to abandoning
		effort; pushing the bar any further would break that constraint.  The
		two never conflict: at the extreme $a=\zeta(h)$ the band degenerates to an upper-tail bonus,
		where $a\phi(a)-b\phi(b)=\zeta(h)\kappa>0$ and \eqref{oa:eq:flexible-soc} fails outright, so the
		binding restriction on $a$ is always \eqref{eq:pointwise-global}.  Numerically, over the entire
		admissible range $\kappa\in(0,\phi(0))$ and $\theta$ up to $12$, the global constraint binds at
		every optimum and \eqref{oa:eq:flexible-soc} is strictly slack.
	\end{remark}
	
	\subsubsection{Proof of Proposition~\ref{oa:prop:flexible-extreme-points}}
	
	\begin{proof}
		Let $\widehat{\mathcal A}(h)$ be defined as in \eqref{oa:eq:flexible-action-set} with
		$[\widehat P_{\min},\widehat P_{\max}]$ in place of $[P_{\min},P_{\max}]$.  By
		\eqref{oa:eq:relaxation-inclusion},
		$\mathcal A^{\mathrm{flex}}(h)\subseteq\widehat{\mathcal A}(h)$.
		
		We first show $\widehat{\mathcal A}(h)$ lies in the quadrilateral with vertices $A_B$, $A_D$,
		$\widetilde A_E$, $A_G$.  Because $(\eta,P)\mapsto(r,s)$ is linear and invertible, it is enough
		to show the feasible $(\eta,P)$ region lies in the quadrilateral with vertices $(0,0)$,
		$(h,\widehat P_{\min}(h))$, $(h,\widehat P_{\max}(h))$ and $(0,M)$.  The region lies between
		the graphs of $\widehat P_{\min}$ and $\widehat P_{\max}$, so it suffices to bound each.
		
		\emph{Upper edge.}  By \eqref{oa:eq:Pmax-flexible},
		$\widehat P_{\max}(\eta)\leq M\Phi(\zeta(h))=\widehat P_{\max}(h)$, while the upper edge of the
		quadrilateral is the chord from $(0,M)$ to $(h,\widehat P_{\max}(h))$, which lies weakly above
		that constant because $M\geq\widehat P_{\max}(h)$.  This step is unconditional.
		
		\emph{Lower edge.}  The lower edge is the chord from $(0,0)$ to $(h,\widehat P_{\min}(h))$, so
		the requirement is $\widehat P_{\min}(\eta)\geq(\eta/h)\widehat P_{\min}(h)$, which is
		Lemma~\ref{oa:lem:rent-per-unit}.  By Lemma~\ref{oa:lem:flexible-full-effort-face},
		$\widehat P_{\min}(h)=M\Phi(z_D(h))$ and $\widehat P_{\max}(h)=M\Phi(\zeta(h))$, so the
		quadrilateral is exactly $\operatorname{co}\{A_B,A_D,\widetilde A_E,A_G\}$.
		
		It remains to place the four vertices in $\mathcal A^{\mathrm{flex}}(h)$.  $A_B$ and $A_G$ are
		the constant contracts $w\equiv0$ and $w\equiv M$, covered by the no-effort regime of
		Section~\ref{oa:sec:regimes}; $A_D$ and $\widetilde A_E$ are covered by
		Lemma~\ref{oa:lem:flexible-full-effort-face}(ii).  Hence
		\begin{equation*}
			\operatorname{co}\bigl\{A_B,A_D,\widetilde A_E,A_G\bigr\}
			\subseteq\operatorname{co}\mathcal A^{\mathrm{flex}}(h)
			\subseteq\operatorname{co}\widehat{\mathcal A}(h)
			=\operatorname{co}\bigl\{A_B,A_D,\widetilde A_E,A_G\bigr\},
		\end{equation*}
		and the chain closes.
	\end{proof}
	
	\begin{remark}
		\label{oa:rem:condition-evidence}
		Because \eqref{oa:eq:rent-normalised} eliminates the horizon, the whole argument can be checked
		exhaustively rather than sampled.  Over the entire admissible range $\kappa\in(0,\phi(0))$ ---
		feasibility forces $\kappa<\phi(0)=1/\sqrt{2\pi}$, since
		$\overline{\mathcal M}(h)\leq\sqrt h\,\phi(0)$ by \eqref{oa:eq:Mbar-definition} --- and $\theta$ up
		to $12$, which covers every $(h,c,M)$ satisfying $\overline{\mathcal M}(h)>ch/M$ with
		$h\leq144$: the closed form \eqref{oa:eq:rent-derivative} agrees with numerical differentiation to
		six decimals, both sign conditions of Step~3 hold at every point, and the jump of Step~1 is
		strictly positive throughout.
	\end{remark}
	
	\section{Beyond Gaussian noise}
	\label{oa:general-noise}
	
	Sections~\ref{sec:single-phase}--\ref{sec:endogenous-review} use Brownian
	output.  This appendix separates the parts of the analysis that depend only on
	the geometry of a binary experiment from those that use Gaussian
	short-horizon asymptotics.  Effort remains binary in both phases; pointwise
	effort is studied in Section~\ref{sec:pointwise-effort}.
	
	Section~\ref{sec:general-noise} states the separation established here: the phase
	problem and the ordered fixed-review menu follow from binary-experiment geometry,
	whereas the review-date results require additional short-horizon properties of the
	likelihood ratio.  Proofs are collected in Section~\ref{oa:app:general-noise}.
	
	\subsection{A general phase experiment}
	\label{oa:sec:general-phase}
	
	Fix a phase of length $h>0$ and replace the Brownian increment by a primitive
	binary experiment.  The phase produces a public signal valued in a standard
	Borel space $\mathcal Y_h$, distributed $P_1^h$ if the agent works and $P_0^h$
	if he shirks.  Let $\nu_h=P_0^h+P_1^h$, write
	$p_j^h=dP_j^h/d\nu_h$, and define the extended likelihood ratios
	$\Lambda_h=p_1^h/p_0^h$ and $\ell_h=p_0^h/p_1^h$,
	with the usual conventions at zero densities.  Under two-sided evidence these
	ratios are finite and reciprocal almost surely.  Working costs the agent $ch$ and
	raises the principal's expected gross output by $m(h)$ relative to shirking;
	as in Section~\ref{sec:environment} the shirking baseline is normalized to zero,
	so $m(h)$ is throughout an \emph{incremental} expected output.\footnote{If
		shirking itself produces output $m_0(h)>0$, that baseline cancels from the
		within-phase work--shirk comparison but not automatically from the comparison
		with stopping.  The formulas below then require the stop payoff to be adjusted
		unless the same baseline accrues under the outside option.  In the Brownian
		model $m_0\equiv0$ and $m(h)=h$.}
	
	A contract is a measurable $w:\mathcal Y_h\to[0,M]$; let $\mathcal F_h$ be the
	set satisfying $\E_1[w]-\E_0[w]\geq ch$.  Writing $\varphi=w/M$ and
	$\alpha(\varphi)=\E_0[\varphi]$ and $\beta(\varphi)=\E_1[\varphi]$,
	the contract is a randomized test of shirking against work, of size $\alpha$
	and power $\beta$.  Expected compensation under work is $M\beta$, the
	incentive spread is $M(\beta-\alpha)$, and the agent's rent when
	\eqref{eq:IC-h} binds is $M\alpha$: \emph{the rent is the cap times
		the size of the test}.  Everything in the phase problem therefore depends on
	the experiment only through its \emph{receiver operating characteristic}, the
	compact convex set
	$\mathcal R_h=\bigl\{(\alpha(\varphi),\beta(\varphi)):0\leq\varphi\leq1\bigr\}\subset[0,1]^2$.
	Let $\beta_h(\cdot)$ be the upper boundary of $\mathcal R_h$ --- the power of
	the most powerful test of each size --- and set
	$g_h(\alpha)=\beta_h(\alpha)-\alpha$ and $\delta(h)=ch/M$.
	Under mutual absolute continuity, three standard facts are all we use.  The map $\beta_h$ is concave and
	nondecreasing with $\beta_h(0)=0$ and $\beta_h(1)=1$; its slope at $\alpha$ is
	the likelihood ratio at the Neyman--Pearson threshold of that size, so that
	$g_h$ is concave with $g_h(0)=g_h(1)=0$ and slope $\Lambda_h-1$; and
	$\max_\alpha g_h(\alpha)=\operatorname{TV}_h:=\sup_A\bigl[P_1^h(A)-P_0^h(A)\bigr]$,
	the total variation distance between the two phase laws.\footnote{Where
		$\Lambda_h$ has atoms, $\beta_h$ is piecewise linear and ``slope'' means the
		appropriate one-sided derivative; the maximizer of $g_h$ is then a vertex and
		need not be a size at which $\Lambda_h$ equals one.  We therefore define roots
		by infima and suprema below.  Smooth comparative statics require the additional
		differentiability conditions stated where they are used.}
	Thus the paper's implementability threshold becomes
	$M_I(h,c)=ch/\operatorname{TV}_h$.
	
	The translation into the notation of Section~\ref{sec:single-phase} is exact.
	With $P_j^h=N(jh,h)$ and $m(h)=h$ we have $\Lambda_h(y)=\exp(y-h/2)$ and,
	parameterizing sizes by $\alpha=\Phi(z-\sqrt h)$,
	\begin{equation*}
		\beta_h(\alpha)=\Phi(z),
		\qquad
		g_h(\alpha)=\Delta(h,z),
		\qquad
		\overline\Delta(h)=\operatorname{TV}_h=2\Phi\!\left(\frac{\sqrt h}{2}\right)-1 .
	\end{equation*}
	The incentive index $\Delta(h,\cdot)$ of Section~\ref{sec:single-phase} is therefore
	the ROC gap function, Figure~\ref{fig:incentive-index} is a picture of a
	concave function vanishing at both ends, and its two roots $z_D(h)$ and
	$z_E(h)$ are the two ends of a superlevel set of that function.  This is why
	the figure has the shape it does, and it is the reason the next lemma requires
	no parametric distributional form.
	
	\begin{lemma}
		\label{oa:lem:general-phase-problem}
		Fix $h>0$ and $M>0$, and suppose $P_1^h$ and $P_0^h$ are mutually
		absolutely continuous.
		\begin{enumerate}[(i)]
			\item $\mathcal F_h\neq\emptyset$ if and only if
			$ch\leq M\operatorname{TV}_h$.
			\item If $ch<M\operatorname{TV}_h$, every contract that minimizes expected
			compensation over $\mathcal F_h$ is a likelihood-ratio test
			$w=M\ind\{\Lambda_h>t\}$, with fractional payment permitted on
			$\{\Lambda_h=t\}$, whose size is
			$\alpha_L(h):=\inf\{\alpha:g_h(\alpha)\geq\delta(h)\}$.
			\item Writing
			$\alpha_H(h):=\sup\{\alpha:g_h(\alpha)\geq\delta(h)\}$, so that
			$[\alpha_L(h),\alpha_H(h)]$ is the compact superlevel interval, and
			\begin{equation*}
				r_L(h)=M\alpha_L(h),
				\qquad
				r_H(h)=M\alpha_H(h),
				\qquad
				\Sigma(h)=m(h)-ch,
			\end{equation*}
			the set of continuation rents $r=\E_1[w]-ch$ attainable while implementing work is
			exactly $[r_L(h),r_H(h)]$, and the full-effort frontier is the horizontal segment
			$s=\Sigma(h)$ over that interval, along which the principal's payoff
			$u=\Sigma(h)-r$ falls one for one with the rent conceded.
		\end{enumerate}
	\end{lemma}
	
	Whenever work is implementable, let
	$\beta_L(h)=\beta_h(\alpha_L(h))$ and use the paper's no-review notation for the
	cheapest work-inducing contract:
	$R^{\NR}(h;c,M)=r_L(h)$ and $V^{\NR}(h;c,M)=\Sigma(h)-r_L(h)=m(h)-M\beta_L(h)$.
	The contract is profitable relative to stopping exactly when
	$V^{\NR}(h;c,M)\geq0$.
	
	Lemma~\ref{oa:lem:general-phase-problem} is the implementability half of
	Lemma~\ref{lem:Mbar}, Lemma~\ref{lem:threshold-reduction}, and
	Lemma~\ref{lem:attainable-rents} at once, with $\overline\Delta(h)$
	read as $\operatorname{TV}_h$.  What the Gaussian added was the extra
	conclusion, in Lemma~\ref{lem:threshold-reduction}, that a likelihood-ratio test is
	a \emph{cutoff in the phase signal}.  That conclusion requires monotonicity of
	$\Lambda_h$ but is not needed for the continuation-rent geometry.  The distinct
	monotonicity condition imposed below concerns $\Lambda_\tau$ and orders review
	actions in the observed review score $X_\tau$.
	
	\subsection{Five assumptions behind the incentive ladder}
	\label{oa:sec:general-assumptions}
	
	Fix a review date $\tau\in(0,T)$ and write $h=T-\tau$.  The first phase
	generates the experiment $(P_1^\tau,P_0^\tau)$ and the continuation phase the
	experiment $(P_1^h,P_0^h)$.  We maintain throughout that the output process has
	stationary and independent increments, so that the continuation problem
	depends on the review state only through the remaining horizon $h$; this is what
	licenses the reduction at the head of
	Section~\ref{sec:extreme-points}, and it is a property of the Brownian
	model that has nothing to do with normality.\footnote{Without it the phase
		laws depend on $X_\tau$, so $r_L$, $r_H$ and $\Sigma$ become functions of the
		review state.  The statewise reduction of
		Proposition~\ref{prop:four-region-review-menu} survives, but the ordering does not:
		the ladder ranks four actions by the rent they deliver, and that ranking is
		informative only if the rents themselves do not move with the state.  Models
		with learning about project quality, or with effort-dependent volatility, must
		confront this separately.}
	
	Consider the following conditions.
	
	\begin{assumption}[Two-sided evidence]
		\label{oa:as:two-sided}
		The two laws are mutually absolutely continuous in both phases:
		$P_1^t\sim P_0^t$ for $t\in\{\tau,h\}$.
	\end{assumption}
	
	\begin{assumption}[Slack detection]
		\label{oa:as:slack-detection}
		$ch<M\operatorname{TV}_h$ and $c\tau<M\operatorname{TV}_\tau$.
	\end{assumption}
	
	\begin{assumption}[Profitable continuation]
		\label{oa:as:profitable-continuation}
		$\Sigma(h)=m(h)-ch>0$.
	\end{assumption}
	
	\begin{assumption}[Monotone review evidence]
		\label{oa:as:mlrp-review}
		$\Lambda_\tau$ is nondecreasing in the performance measure $X_\tau$.
	\end{assumption}
	
	\begin{assumption}[Randomization]
		\label{oa:as:randomisation}
		Either the law of $\Lambda_\tau$ under $P_1^\tau$ is atomless, or the
		principal may publicly randomize between adjacent review actions.
	\end{assumption}
	
	Assumption~\ref{oa:as:two-sided} says that no realization of the phase signal is
	conclusive: neither work nor shirking can be ruled out by any observation.
	Assumption~\ref{oa:as:slack-detection} is the strict phase-by-phase analogue of
	the implementability condition \eqref{eq:bar-condition}.  Unlike
	Assumption~\ref{as:maintained}, it concerns implementability only; profitability
	is imposed separately below.
	Assumption~\ref{oa:as:profitable-continuation} holds automatically in the Brownian
	model, where $\Sigma(h)=(1-c)h>0$ because $c<1$; with a general output
	technology it is a restriction.  Assumption~\ref{oa:as:mlrp-review} is the
	monotone likelihood-ratio property, which the Gaussian delivers through
	\eqref{eq:Q-lambda-definition}.  Assumption~\ref{oa:as:randomisation} is what
	makes the review cutoffs well defined and the primal problem attain the dual
	bound; it is automatic for the Brownian review signal.
	
	\begin{proposition}
		\label{oa:prop:general-ladder}
		Suppose Assumptions~\ref{oa:as:two-sided}--\ref{oa:as:randomisation} hold at
		$(\tau,h)$.  Then
		\begin{equation}
			0<r_L(h)<r_H(h)<M ,
			\label{oa:eq:general-rent-ordering}
		\end{equation}
		the convex hull of the review action set
		$\mathcal A(h)=\mathcal A_0(h)\cup\mathcal A_1(h)$ has exactly the four
		extreme points
		\begin{equation*}
			A_B=(0,0),
			\quad
			A_L=\bigl(r_L(h),\Sigma(h)\bigr),
			\quad
			A_H=\bigl(r_H(h),\Sigma(h)\bigr),
			\quad
			A_G=(M,0),
		\end{equation*}
		and the phase reduction, pre-review Lagrangian, and fixed-review action
		ordering hold after the substitutions
		$\overline\Delta\mapsto\operatorname{TV}$ and $(1-c)h\mapsto\Sigma(h)$.
		The review-date derivatives and endpoint results require the additional
		assumptions studied below.  In particular the support function is
		\begin{equation}
			\mathcal V(q,h)
			=
			\max\bigl\{0,\ \Sigma(h)+qr_L(h),\ \Sigma(h)+qr_H(h),\ qM\bigr\},
			\label{oa:eq:general-support}
		\end{equation}
		the index is $Q_\lambda=\lambda(1-\ell_\tau)-1$, the optimal menu selects
		$A_B,A_L,A_H,A_G$ in that order as $Q_\lambda$ crosses the three
		breakpoints
		\begin{equation}
			q_0=-\frac{\Sigma(h)}{r_L(h)}
			\ <\ 0\ <\
			q_2=\frac{\Sigma(h)}{M-r_H(h)},
			\label{oa:eq:general-breakpoints}
		\end{equation}
		and the fixed-review value is
		\begin{equation}
			V^R(\tau)
			=
			\inf_{\lambda\geq0}
			\Bigl\{m(\tau)-\lambda c\tau
			+\E_1\bigl[\mathcal V(Q_\lambda,h)\bigr]\Bigr\}.
			\label{oa:eq:general-review-dual}
		\end{equation}
		For a positive supporting multiplier away from score ties, the menu has at
		most four ordered regions in $X_\tau$; public randomization implements the
		corresponding mixture at an atom or a zero-multiplier tie.
	\end{proposition}
	
	The whole of the ladder is thus one geometric fact: bounded compensation applied to a signal
	that is informative but never conclusive generates a nondegenerate \emph{interval} of
	continuation rents compatible with continuation work, strictly inside the interval $[0,M]$ of
	rents compatible with stopping.  The two stops are the two extreme rents, the two work
	actions are the two ends of the work interval, and a monotone review score ranks all four by
	the rent they deliver.
	
	\subsection{Three identities that replace the Gaussian formulas}
	\label{oa:sec:general-identities}
	
	Proposition~\ref{oa:prop:general-ladder} preserves the structure of the analysis;
	it does not preserve the closed forms, which used the normal distribution
	function throughout.  Three exact identities replace them, and each specializes
	to a display already derived in
	Sections~\ref{sec:single-phase}--\ref{sec:endogenous-review}.
	
	\paragraph{Rent is cost divided by evidence.}  Recall that $\beta_L(h)$ is the
	normalized expected payment under work at the cheapest contract, and define the
	\emph{evidence ratio} of the reward
	event,
	\begin{equation}
		\overline\Lambda_h
		=\frac{\beta_L(h)}{\alpha_L(h)}
		=\frac{\E_1[\varphi_L]}{\E_0[\varphi_L]}\ >\ 1 .
		\label{oa:eq:evidence-ratio}
	\end{equation}
	If the generalized test is implemented with an independent uniform draw, the
	last ratio is the work-to-shirk likelihood ratio averaged over the event on
	which the cap is paid.  This interpretation also covers a fractional payment at
	a likelihood-ratio atom.
	Since $M[\beta_L-\alpha_L]=ch$ and $R^{\NR}(h)=r_L(h)=M\alpha_L(h)$,
	\begin{equation}
		r_L(h)=\frac{ch}{\overline\Lambda_h-1},
		\qquad\text{and dually}\qquad
		M-ch-r_H(h)=\frac{ch}{\overline\Lambda^-_h-1},
		\label{oa:eq:rent-identity}
	\end{equation}
	where $\beta_H(h)=\beta_h(\alpha_H(h))$ and
	$\overline\Lambda^-_h=(1-\alpha_H)/(1-\beta_H)$ is the shirk-to-work odds
	ratio of the event on which payment is withheld at the rent-maximizing
	contract.  Equation~\eqref{oa:eq:rent-identity} isolates the role of limited
	liability: the rent conceded per unit of effort cost is the reciprocal of the
	excess evidence carried by the rewarded event.  The two ratios in
	\eqref{oa:eq:rent-identity} are
	\emph{average} likelihood ratios over the reward and withholding events, not
	the marginal likelihood ratio at the cutoff, which is the distinct object
	$\ell_h$ appearing below.  Letting $M\to\infty$ drives $\alpha_L(h)$ to zero
	and $\overline\Lambda_h$ up to $\beta_h'(0^+)=\Lambda^{\max}_h$, the essential
	supremum of the likelihood ratio, so
	$\lim_{M\to\infty}R^{\NR}(h;c,M)=ch/(\Lambda^{\max}_h-1)$.
	The Gaussian case is $\Lambda^{\max}_h=\infty$, in which the limit is zero; this is the
	step in the proof of Lemma~\ref{lem:Mbar} that drives $r_D(h)\to0$ as $M\to\infty$.  When the
	available likelihood ratio is bounded --- for example under some coarse
	monitoring technologies --- agency rent does
	\emph{not} vanish as the cap loosens.  The closing discussion of
	Section~\ref{sec:single-phase} should accordingly be read as a statement
	about signals with unbounded likelihood ratios: it is the boundedness of the
	\emph{evidence}, and only derivatively the boundedness of the \emph{prize},
	that keeps rents alive.  A positive rent creates scope for a review to recycle
	continuation utility as incentives, but does not by itself prove that a review
	dominates the no-review contract.
	
	\paragraph{The symmetry identity.}  The relation
	$r_L(h)+r_H(h)=M-ch$, which in the Brownian model follows from
	$z_E(h)=\sqrt h-z_D(h)$, is a symmetry property of the experiment.
	
	\begin{proposition}
		\label{oa:prop:general-symmetry}
		Under Assumption~\ref{oa:as:two-sided} the following are equivalent.
		\begin{enumerate}[(i)]
			\item $r_L(h)+r_H(h)=M-ch$ for every $(c,M)$ with
			$0<ch<M\operatorname{TV}_h$.
			\item $\mathcal R_h$ is invariant under the anti-diagonal reflection
			$(\alpha,\beta)\mapsto(1-\beta,1-\alpha)$.
			\item The experiment obtained by interchanging the two hypotheses,
			$(P_0^h,P_1^h)$, is Blackwell equivalent to $(P_1^h,P_0^h)$.
		\end{enumerate}
		A sufficient condition is the existence of a measurable map $\sigma$ with
		$\sigma_\#P_1^h=P_0^h$ and $\sigma_\#P_0^h=P_1^h$; in particular the
		identity holds for every location family with a symmetric, everywhere
		positive noise density, with $\sigma(y)=m-y$ at mean shift $m$.  It fails,
		for instance, when effort shifts a Poisson arrival rate.
	\end{proposition}
	
	\paragraph{The no-review derivative and the tight-cap limit.}
	Let $\ell_h=1/\beta_h'(\alpha_L(h))$ be the likelihood ratio at the binding
	bar, keep $\lambda^{\NR}_h=1/(1-\ell_h)$ as in
	Section~\ref{subsec:endpoint-derivatives}, and write
	$\dot\beta_h=M\,\partial_h\beta_h(\alpha)|_{\alpha=\alpha_L(h)}$ for the value
	of lengthening the phase at a fixed size.
	
	\begin{proposition}
		\label{oa:prop:general-nr-derivative}
		Suppose $m$ is differentiable at $h$ and $\beta_h(\alpha)$ is differentiable in $(h,\alpha)$ at
		$\alpha_L(h)$, with $\beta_h'(\alpha_L(h))>1$.  Then
		\begin{equation}
			\frac{dR^{\NR}(h)}{dh}=\frac{c-\dot\beta_h}{\beta_h'(\alpha_L(h))-1},
			\qquad
			\frac{dV^{\NR}(h)}{dh}
			=m'(h)-\lambda^{\NR}_h\bigl[c-\ell_h\dot\beta_h\bigr].
			\label{oa:eq:general-nr-derivative}
		\end{equation}
		In the Brownian model $m'\equiv1$ and
		$\ell_h\dot\beta_h=M\phi(z_D(h)-\sqrt h)/(2\sqrt h)$, so
		\eqref{oa:eq:general-nr-derivative} is \eqref{eq:left-derivative-T-explicit}.
	\end{proposition}
	
	Now let the cap fall to the feasibility boundary.  The two roots merge at the
	maximizer $\alpha^\star_h$ of $g_h$, and the marginal likelihood ratio there
	approaches one precisely when $g_h$ has no kink at its peak.  Writing
	$\varepsilon(h)=d\log\operatorname{TV}_h/d\log h$
	for the elasticity of detection with respect to the phase length, the envelope
	theorem applied to $\operatorname{TV}_h=\max_\alpha g_h(\alpha)$ gives
	$\partial_h\beta_h(\alpha^\star_h)=\operatorname{TV}_h'$ and hence the following.
	
	\begin{corollary}
		\label{oa:cor:general-tight-cap}
		Suppose the differentiability conditions of
		Proposition~\ref{oa:prop:general-nr-derivative} hold in a neighborhood of the
		feasibility boundary, $\operatorname{TV}_h$ is differentiable at $h$, and
		$\operatorname{ess\,inf}\{\Lambda_h:\Lambda_h>1\}=1$, so that the
		likelihood ratio takes values arbitrarily close to one from above.  Then
		$\ell_h\uparrow1$ and $\lambda_h^{\NR}\to\infty$ as
		$M\downarrow M_I(h,c)$, and
		\begin{equation*}
			\varepsilon(h)<1
			\quad\Longrightarrow\quad
			\lim_{M\downarrow M_I(h,c)}\frac{dV^{\NR}(h)}{dh}=-\infty ,
			\qquad
			\varepsilon(h)>1
			\quad\Longrightarrow\quad
			\text{the limit is }+\infty .
		\end{equation*}
		Moreover $\varepsilon(h)<1$ if and only if $\operatorname{TV}_h/h$ is
		strictly decreasing at $h$.
	\end{corollary}
	
	The condition $\varepsilon(h)<1$ is the property established for the Gaussian in the proof
	of Lemma~\ref{lem:Mbar}, where $\overline\Delta(h)/h$ is shown to be strictly decreasing.  The inequality
	$2\Phi(a)-1>a\phi(a)$ in \eqref{eq:normal-cap-inequality}, proved in
	Appendix~\ref{app:review-date}, is its Gaussian instance: with
	$a=\sqrt h/2$,
	$\varepsilon(h)=a\phi(a)/(2\Phi(a)-1)$.
	Thus the same elasticity condition underlies both calculations.  For a general
	experiment family it must be imposed explicitly.  The regularity condition is
	also substantive.  A
	two-point signal with likelihood ratios $\{1.75,0.5\}$ has a gap around one
	and its ROC has a vertex at the total-variation maximizer;
	the marginal likelihood ratio on the lower branch is stuck at $1.75$, and
	$\lambda^{\NR}_h$ remains bounded however tight the cap.  For such signals the
	tight-cap mechanism discussed in
	Section~\ref{subsec:endpoint-derivatives} --- the two
	incentive-compatible thresholds nearly coincide and the shadow price of
	incentives explodes --- simply does not operate.  What drives that mechanism is
	the \emph{continuity} of the likelihood ratio around one, not its Gaussian
	form.
	
	\subsection{Timing: diffusive versus hard evidence}
	\label{oa:sec:general-timing}
	
	The results on the review \emph{date} are far less robust than the results on
	the review \emph{menu}.  Throughout this subsection the output process has
	stationary independent increments, $m(h)=m_1h$, and the experiment family
	is continuous in total variation at zero:
	$\operatorname{TV}_h\to0$ as $h\downarrow0$.  Stochastic continuity alone is
	not enough for the last property, so it is imposed explicitly.  As in the
	baseline analysis, we also maintain that the full-horizon no-review contract is
	strictly implementable and profitable, $V^{\NR}(T;c,M)>0$.
	
	\begin{lemma}
		\label{oa:lem:general-tv-subadditive}
		$\operatorname{TV}_{h_1+h_2}\leq\operatorname{TV}_{h_1}+\operatorname{TV}_{h_2}$
		for all $h_1,h_2>0$, and
		\begin{equation*}
			\Theta_0
			:=\lim_{h\downarrow0}\frac{\operatorname{TV}_h}{h}
			=\sup_{h>0}\frac{\operatorname{TV}_h}{h}
			\in[0,+\infty],
		\end{equation*}
		the limit existing.
	\end{lemma}
	
	The second short-horizon index is the limit of the evidence ratio
	\eqref{oa:eq:evidence-ratio}.  Fix a cap and cost for which all sufficiently
	short phases are implementable, and suppose
	$\Lambda_0:=\lim_{h\downarrow0}\overline\Lambda_h$ exists in $[1,\infty]$,
	so that by \eqref{oa:eq:rent-identity} the minimum rent per unit of time converges,
	$r_L(h)/h\to c/(\Lambda_0-1)$, with the convention $c/0=+\infty$.  Say that the
	family is \emph{diffusive at zero} if $\Lambda_0=1$ and that it \emph{carries
		hard evidence at zero} if $\Lambda_0>1$: in the first case the cheapest reward
	event capable of supplying the required incentive spread has an average
	likelihood ratio converging to one, while in the second its evidential content
	remains bounded away from one.  Combining, the no-review value per unit
	of time has the limit
	\begin{equation*}
		\nu:=\lim_{h\downarrow0}\frac{V^{\NR}(h)}{h}
		=m_1-c-\frac{c}{\Lambda_0-1}
		\ \in[-\infty,\infty) .
	\end{equation*}
	The number $\nu$ governs whether a vanishing continuation phase is valuable at
	the terminal endpoint.  The initial endpoint instead depends on the two tails
	of the short-horizon likelihood ratio and on the utility spread available at
	the review.
	
	\begin{proposition}
		\label{oa:prop:general-feasibility-dichotomy}
		If $M\Theta_0<c$, then work is implementable over no phase of positive
		length, so no contract with any number of reviews implements effort
		anywhere.  If $M\Theta_0>c$, work is implementable over every sufficiently
		short phase, and if in addition $\operatorname{TV}_h/h$ is strictly
		decreasing and continuous there is a unique maximal implementable phase
		length $\overline h(M)$, the root of $M\operatorname{TV}_h=ch$.
	\end{proposition}
	
	Proposition~\ref{oa:prop:general-feasibility-dichotomy} is where the Brownian
	model is special.  With Brownian noise $\operatorname{TV}_h\sim\sqrt{h/2\pi}$
	and $\Theta_0=\infty$, so short phases are always implementable, whatever the
	cap: effort cost accumulates linearly in the phase length while
	distinguishability grows like $\sqrt h$, and shortening the detection window
	therefore buys detection power.  For a pure finite-intensity arrival
	experiment, by contrast, $\Theta_0<\infty$ and feasibility becomes the
	\emph{scale-free} restriction $M\Theta_0>c$, which no subdivision of the
	horizon can relax.  The feasibility index $\Theta_0$ and the evidence ratio
	$\Lambda_0$ capture distinct properties; a technology combining diffusive and
	jump signals can have $\Theta_0=\infty$ while also carrying hard rare-event
	evidence.
	
	The endpoint results require separate conditions, stated next.
	
	\begin{proposition}
		\label{oa:prop:general-endpoints}
		Suppose work is strictly implementable in a neighborhood of $T$,
		$V^{\NR}(T;c,M)>0$, and $r_L(t)$ and $V^{\NR}(t;c,M)$ are locally
		Lipschitz at $t=T$.
		\begin{enumerate}[(i)]
			\item \textup{(Early review.)}  Suppose there are $\kappa>0$ and
			$C<\infty$ and events $E_\tau$ in the review $\sigma$-field with
			\begin{equation*}
				P_1^\tau(E_\tau)-P_0^\tau(E_\tau)\geq\kappa\tau,
				\qquad
				P_1^\tau(E_\tau)\leq C\tau
			\end{equation*}
			for all small $\tau$, and suppose $\kappa\bigl[M-r_L(T)\bigr]>c$.  Then
			\begin{equation*}
				\liminf_{\tau\downarrow0}
				\frac{V^R(\tau)-V^{\NR}(T;c,M)}{\tau}>-\infty ,
			\end{equation*}
			so \eqref{eq:right-derivative-zero} fails.
			\item \textup{(Late review.)}  Suppose $\nu>0$.  Let $B_\tau$ be the
			reward event in a cheapest work-inducing contract for a phase of length
			$\tau$, on the probability space augmented by any public randomization used
			at a likelihood-ratio atom.  Suppose there are events
			$E_\tau\subseteq B_\tau^c\cap\{\Lambda_\tau>1\}$ such that
			$P_1^\tau(E_\tau)>P_0^\tau(E_\tau)$ and
			$P_1^\tau(E_\tau)\to p>0$ as $\tau\uparrow T$.  Then
			\begin{equation*}
				V^R(T-h)\ \geq\ V^{\NR}(T-h;c,M)+\nu p\,h+o(h),
			\end{equation*}
			so the first-order equivalence
			$V^R(T-h)=V^{\NR}(T-h;c,M)+o(h)$ of Lemma~\ref{lem:terminal-derivative} fails
			and, using Lemma~\ref{lem:review-boundary-values},
			\begin{equation*}
				\limsup_{h\downarrow0}
				\frac{V^{\NR}(T;c,M)-V^R(T-h)}{h}
				\leq
				\limsup_{h\downarrow0}
				\frac{V^{\NR}(T;c,M)-V^{\NR}(T-h;c,M)}{h}-\nu p .
			\end{equation*}
			In particular, if the two left derivatives exist, then
			$V^{R\prime}_-(T)\leq \partial_TV^{\NR}(T;c,M)-\nu p
			<\partial_TV^{\NR}(T;c,M)$.
		\end{enumerate}
		These are separate sufficient conditions.  Hard upper-tail evidence helps
		satisfy the early-review event condition and can make $\nu$ finite, but neither
		the displayed incentive-capital inequality nor $\nu>0$ follows from hard
		evidence alone.
	\end{proposition}
	
	The mechanism in part~(ii) runs in the opposite direction from the Gaussian
	case.  The Gaussian proof of
	\eqref{eq:left-derivative-T-explicit} shows that a vanishing
	continuation phase contributes nothing at first order, and it does so by
	observing that $r_L(h)$ vastly exceeds the continuation surplus $\Sigma(h)$
	when $h$ is small: continuation work is ruinously rent-expensive over a short
	phase, so the two work actions are dominated on all but a vanishing range of
	the review index.  Under hard evidence the ranking reverses, $r_L(h)$ and
	$\Sigma(h)$ are both of order $h$, and if $\nu>0$ the low-rent continuation is
	strictly better than the bad stop over a fixed range of review states.  A short
	continuation phase then has \emph{first-order} value, which lowers the upper
	left slope at $T$ --- and the left derivative when it exists --- and therefore
	makes an interior review easier, not harder, to justify at that endpoint.
	
	The mechanism in part~(i) is different.  The Gaussian argument writes obedience as
	$\E_1\bigl[(r-r_L)(1-\ell_\tau)\bigr]\geq c\tau$ and uses that
	$|1-\ell_\tau|$ is of order $\sqrt{\tau\log(1/\tau)}$ outside an event of
	probability $o(\tau)$, so that buying $c\tau$ of incentive requires utility
	dispersion of order $\sqrt\tau/\sqrt{\log(1/\tau)}$, which is large relative to
	$\tau$.  The general form of that step is: for every $\theta\in(0,1)$,
	\begin{equation*}
		\E_1\bigl|r-r_L\bigr|
		\ \geq\
		\frac{c\tau-M\,G_\tau(\theta)}{\theta},
		\qquad
		G_\tau(\theta):=\E_1\bigl[(|1-\ell_\tau|-\theta)^+\bigr].
	\end{equation*}
	Under the maintained strict profitability of the low-rent continuation at
	$T$, the principal's loss controls this utility dispersion exactly as in the
	proof of Proposition~\ref{prop:extreme-review-dates}(i).  Hence
	\eqref{eq:right-derivative-zero} holds whenever one can choose
	$\theta_\tau\downarrow0$ with $MG_\tau(\theta_\tau)\leq c\tau/2$.  In the
	Brownian model $\theta_\tau=2\sqrt{\tau\log(1/\tau)}$ does this and returns the
	rate in the appendix.  The absolute value in $G_\tau$ is essential:
	incentives may be supplied by punishing highly informative evidence of shirking as well
	as by rewarding evidence of work, and both channels must be shut down for the
	bound to apply.  In a finite-intensity arrival model in which short phases are
	feasible,
	$G_\tau(\theta)$ is of order $\tau$ for every $\theta$ below the
	likelihood-ratio jump carried by a single arrival.  No vanishing
	$\theta_\tau$ then satisfies $MG_\tau(\theta_\tau)\leq c\tau/2$, so this
	argument cannot establish the Gaussian endpoint result.

	Proposition~\ref{oa:prop:general-endpoints}(i) is what qualifies the organizational reading
	of Proposition~\ref{prop:extreme-review-dates}(i) in
	Section~\ref{sec:general-noise}: it bounds the cost of a short probation by its length when
	a rare event --- a missed milestone, a lost account, a trial readout --- retains substantial
	evidential content as the window shrinks.  Note what it does not deliver.  The
	incentive-capital and profitability conditions remain necessary for probation to be
	\emph{optimal}; ruling out an infinitely steep initial loss is not the same as making an
	early review worthwhile.
	
	\subsection{Which rungs the ladder has: the two tails of the likelihood ratio}
	\label{oa:sec:general-rungs}
	
	Proposition~\ref{oa:prop:general-ladder} says the optimal menu climbs
	$A_B\to A_L\to A_H\to A_G$.  It does not say that all four rungs are reached.
	At a positive supporting multiplier the Brownian bad stop is always reached, because arbitrarily bad news is
	available; the good stop is reached only when incentives are expensive enough.
	That asymmetry is not general.  The following likelihood-ratio bounds
	characterize which outer actions are reached.
	
	The relevant primitives are the two tails of the review likelihood ratio,
	\begin{equation}
		\ell^{\max}_\tau=\operatorname*{ess\,sup}\frac{dP_0^\tau}{dP_1^\tau},
		\qquad
		\Lambda^{\max}_\tau=\operatorname*{ess\,sup}\frac{dP_1^\tau}{dP_0^\tau},
		\qquad\text{both in }(1,\infty] .
		\label{oa:eq:tail-evidence}
	\end{equation}
	The first measures the strongest available evidence of shirking, the second the
	strongest available evidence of work.  Since $Q_\lambda=\lambda(1-\ell_\tau)-1$
	is an increasing transformation of the score, the closed essential range of the index is
	$\bigl[\lambda(1-\ell^{\max}_\tau)-1,\ \lambda(1-1/\Lambda^{\max}_\tau)-1\bigr]$,
	and whether it reaches the outer breakpoints of
	\eqref{oa:eq:general-breakpoints} is a question about \eqref{oa:eq:tail-evidence}.
	
	\begin{proposition}
		\label{oa:prop:general-stops}
		Under the hypotheses of Proposition~\ref{oa:prop:general-ladder}, with a
		positive supporting multiplier $\lambda>0$:
		\begin{enumerate}[(i)]
			\item the bad stop is strictly optimal on a set of positive probability if
			and only if
			\begin{equation}
				\lambda\bigl(\ell^{\max}_\tau-1\bigr)\ >\ \frac{\Sigma(h)}{r_L(h)}-1 ,
				\label{oa:eq:bad-stop-condition}
			\end{equation}
			which holds for \emph{every} $\lambda>0$ when $\ell^{\max}_\tau=\infty$;
			\item the good stop is strictly optimal on a set of positive probability if
			and only if
			\begin{equation}
				\lambda\left(1-\frac{1}{\Lambda^{\max}_\tau}\right)
				\ >\ 1+\frac{\Sigma(h)}{M-r_H(h)} ,
				\label{oa:eq:good-stop-condition}
			\end{equation}
			that is, $\lambda$ must exceed
			$\bigl(1+\Sigma(h)/(M-r_H(h))\bigr)\Lambda^{\max}_\tau/
			(\Lambda^{\max}_\tau-1)$.  This reduces to
			$\lambda>1+\Sigma(h)/(M-r_H(h))$ when $\Lambda^{\max}_\tau=\infty$
			and diverges as $\Lambda^{\max}_\tau\downarrow1$.
		\end{enumerate}
		At equality in either condition, the corresponding stop can be used only on
		a likelihood-ratio atom at which it ties the adjacent continuation action; whether it is
		used then depends on tie-breaking and the aggregate obedience requirement.
	\end{proposition}
	
	The conditions reflect the opportunity cost of stopping.  The good stop pays
	the maximal prize and abandons production; the principal chooses it only after
	evidence of work strong
	enough to justify surrendering the continuation, and
	\eqref{oa:eq:good-stop-condition} says how strong.  The bad stop pays nothing and
	abandons production; that requires evidence of shirking strong enough to
	justify the same sacrifice, which is \eqref{oa:eq:bad-stop-condition}.  A
	monitoring technology that cannot generate one kind of evidence cannot make
	the corresponding rung strictly optimal; at a boundary atom the rung may still
	appear through tie-breaking.
	
	\subsection{Proofs for the beyond-Gaussian extension}\label{oa:app:general-noise}
	
	\begin{proof}[Proof of Lemma~\ref{oa:lem:general-phase-problem}]
		Part~(i) restates $\sup_\varphi(\beta-\alpha)=\operatorname{TV}_h$.  For
		(ii), minimizing expected compensation is minimizing $\beta$ subject to
		$\beta-\alpha\geq\delta(h)$; the constraint binds, since scaling $\varphi$
		down by $\delta(h)/(\beta-\alpha)$ preserves $0\leq\varphi\leq1$, restores
		equality and strictly lowers $\beta$, and among binding tests the
		Neyman--Pearson lemma selects a likelihood-ratio test.  Concavity and
		continuity of $g_h$ make its superlevel set a compact interval; its left
		endpoint $\alpha_L$ is the smallest feasible size and hence gives the cheapest
		binding test.  For (iii), maximizing $\beta$ subject to the same constraint
		also yields a binding test: if the spread is strictly larger than $\delta(h)$,
		mixing the test with the constant test $\varphi\equiv1$ until the spread binds
		strictly raises $\beta$.  After the substitution $\psi=1-\varphi$, this is the
		mirror-image problem $\E_0\psi-\E_1\psi\geq\delta(h)$; its solution withholds
		payment where $\Lambda_h$ is smallest and has size $\alpha_H(h)$.  Convex
		combinations of the two extreme binding tests remain binding and sweep every
		intermediate size.  Therefore the feasible values of $\beta$, and hence of
		$r=M\beta-ch$, form exactly $[r_L,r_H]$.  Finally
		$u+r=m(h)-M\beta+M\beta-ch=\Sigma(h)$ for every work contract, so total
		surplus is constant at $\Sigma(h)$ across the interval and the principal's
		payoff falls one for one with rent.
	\end{proof}
	
	\begin{proof}[Proof of Proposition~\ref{oa:prop:general-ladder}]
		Assumption~\ref{oa:as:slack-detection} makes $\{g_h\geq\delta(h)\}$ a
		nondegenerate interval, so $\alpha_L(h)<\alpha_H(h)$, and it makes
		$\alpha_H(h)<1$ because $g_h(1)=0<\delta(h)$; hence $r_L<r_H<M$.
		Assumption~\ref{oa:as:two-sided} gives $\alpha_L(h)>0$: a test of size zero has
		$\beta=0$ under mutual absolute continuity, hence $g_h(0)=0<\delta(h)$.  This
		is \eqref{oa:eq:general-rent-ordering}, and with
		Lemma~\ref{oa:lem:general-phase-problem}(iii) it identifies the four extreme
		points of the convex hull of the two segments.  The statewise Lagrangian is
		$\Sigma(h)i(x)+Q_\lambda(x)r(x)$, exactly
		\eqref{eq:frontier-substituted-lagrangian}, which is affine in $r$ with slope
		$Q_\lambda$; the four candidate rents $0<r_L<r_H<M$ are therefore selected in
		increasing order of $Q_\lambda$, with the breakpoints obtained by equating
		the corresponding affine functions in \eqref{oa:eq:general-support}, and
		Assumption~\ref{oa:as:profitable-continuation} orders them as in
		\eqref{oa:eq:general-breakpoints}.  Assumption~\ref{oa:as:mlrp-review} makes
		$Q_\lambda$ nondecreasing in $X_\tau$, so the regions are intervals of
		performance.  Assumption~\ref{oa:as:randomisation}, with
		Assumption~\ref{oa:as:slack-detection} supplying a strictly feasible menu,
		convexifies the aggregate constraint set and delivers attainment of
		\eqref{oa:eq:general-review-dual}.  No other property of the two experiments is
		used.
	\end{proof}
	
	\begin{proof}[Proof of Proposition~\ref{oa:prop:general-symmetry}]
		Complementing a test, $\varphi\mapsto1-\varphi$, shows that $\mathcal R_h$ is
		always symmetric about $(\tfrac12,\tfrac12)$; interchanging the hypotheses
		maps $(\alpha,\beta)$ to $(\beta,\alpha)$ and so reflects $\mathcal R_h$ in
		the diagonal.  Composing the two reflections gives the anti-diagonal
		reflection, whence (ii)$\Leftrightarrow$(iii), two binary experiments being
		Blackwell equivalent exactly when they share an ROC region.
		(ii)$\Rightarrow$(i): the reflection preserves $\beta-\alpha$ and maps the
		upper boundary to itself, so it exchanges the two endpoints of the binding
		chord, giving $\alpha_H=1-\beta_L=1-\alpha_L-\delta(h)$.
		(i)$\Rightarrow$(ii): mutual absolute continuity makes $\beta_h$ and $g_h$
		continuous.  For every $\delta\in(0,\operatorname{TV}_h)$, statement (i)
		gives
		$\alpha_H(\delta)=1-\alpha_L(\delta)-\delta
		=1-\beta_h(\alpha_L(\delta))$.  Thus the anti-diagonal reflection of the
		left endpoint of every superlevel interval is its right endpoint.  As
		$\delta$ varies, these endpoints trace the two sides of the upper ROC boundary;
		continuity covers any plateau at the maximum and the endpoints at
		$\delta=0$.  Hence the upper boundary, and therefore the ROC region beneath
		it, is invariant under the reflection.  For the sufficient condition,
		$\E_0[\varphi\circ\sigma]=\E_1[\varphi]$ and
		$\E_1[\varphi\circ\sigma]=\E_0[\varphi]$, so $1-\varphi\circ\sigma$ has size
		$1-\beta(\varphi)$ and power $1-\alpha(\varphi)$.
	\end{proof}
	
	\begin{proof}[Proof of Proposition~\ref{oa:prop:general-nr-derivative}]
		$V^{\NR}(h)=m(h)-M\beta_h(\alpha_L(h))$ and $R^{\NR}(h)=M\alpha_L(h)$, with
		$\alpha_L$ determined by $M[\beta_h(\alpha_L)-\alpha_L]=ch$.  Differentiating
		the constraint gives $\dot\beta_h+(\beta_h'-1)M\alpha_L'=c$, which is the
		first display; then
		$\frac{d}{dh}[M\beta_h(\alpha_L)]
		=\dot\beta_h+\beta_h'M\alpha_L'
		=(\beta_h'c-\dot\beta_h)/(\beta_h'-1)
		=(c-\ell_h\dot\beta_h)/(1-\ell_h)$.  For the Brownian specialization,
		holding $\alpha=\Phi(z-\sqrt h)$ fixed means holding $z-\sqrt h$ fixed, so
		$\partial_h\beta_h=\phi(z_D)/(2\sqrt h)$, while
		$\ell_h=\phi(z_D-\sqrt h)/\phi(z_D)$.
	\end{proof}
	
	\begin{proof}[Proof of Corollary~\ref{oa:cor:general-tight-cap}]
		As $M\downarrow M_I(h,c)$, the lower root converges to a maximizer
		$\alpha_h^\star$ of $g_h$.  The likelihood-ratio condition and the maintained
		differentiability imply
		$\beta_h'(\alpha_L(h))\downarrow1$, and therefore
		$\ell_h\uparrow1$ and $\lambda_h^{\NR}\to\infty$.  The envelope theorem gives
		$\partial_h\beta_h(\alpha_h^\star)=\operatorname{TV}_h'$, so
		$c-\ell_h\dot\beta_h\to c-M_I(h,c)\operatorname{TV}_h'=c\bigl[1-\varepsilon(h)\bigr]$.
		The two divergence claims now follow from
		\eqref{oa:eq:general-nr-derivative}; the finite term $m'(h)$ does not affect
		the limit.  Finally,
		$d\log(\operatorname{TV}_h/h)/d\log h=\varepsilon(h)-1$, which proves the
		last equivalence.
	\end{proof}
	
	\begin{proof}[Proof of Lemma~\ref{oa:lem:general-tv-subadditive}]
		The time-$(h_1+h_2)$ increment is a measurable function of the pair of
		independent phase increments, so the data-processing inequality and the
		tensorization bound for total variation give the first claim.  For the
		second, fix $h>0$ and $0<s<h$, write $h=ns+r$ with $n=\lfloor h/s\rfloor$ and
		$0\leq r<s$, and apply subadditivity:
		$\operatorname{TV}_h\leq n\operatorname{TV}_s+\operatorname{TV}_r$, so
		$\operatorname{TV}_h/h\leq\operatorname{TV}_s/s+\operatorname{TV}_r/h$.
		The maintained total-variation continuity makes
		$\operatorname{TV}_r\to0$ as $s\downarrow0$, so
		letting $s\downarrow0$ along an arbitrary sequence gives
		$\operatorname{TV}_h/h\leq\liminf_{s\downarrow0}\operatorname{TV}_s/s$.
		Taking the supremum over $h$ yields
		$\sup_h\operatorname{TV}_h/h\leq\liminf_{s\downarrow0}\operatorname{TV}_s/s
		\leq\limsup_{s\downarrow0}\operatorname{TV}_s/s\leq\sup_h\operatorname{TV}_h/h$.
	\end{proof}
	
	\begin{proof}[Proof of Proposition~\ref{oa:prop:general-feasibility-dichotomy}]
		By Lemma~\ref{oa:lem:general-phase-problem}(i), a phase of length $h$ is
		implementable exactly when $\operatorname{TV}_h/h\geq c/M$.  Lemma~\ref{oa:lem:general-tv-subadditive}
		shows that the left-hand side attains its supremum in the limit $h\downarrow0$,
		which proves the first two claims.  Under strict monotonicity and continuity it
		crosses $c/M$ at most once; it crosses exactly once because its limit at zero is
		larger than $c/M$ and $\operatorname{TV}_h/h\leq1/h\to0$ as $h\to\infty$.
	\end{proof}
	
	\begin{proof}[Proof of Proposition~\ref{oa:prop:general-endpoints}]
		(i)  Consider the two-region menu assigning $A_L(h)$, $h=T-\tau$, off
		$E_\tau$ and the good stop $A_G$ on $E_\tau$.  Its incentive supply is
		$\bigl[M-r_L(h)\bigr]\bigl[P_1^\tau(E_\tau)-P_0^\tau(E_\tau)\bigr]
		\geq\kappa\bigl[M-r_L(h)\bigr]\tau>c\tau$ for all small $\tau$, since
		$r_L(h)\to r_L(T)$; pre-review obedience is an inequality, so the menu is
		feasible as it stands.  Relative to assigning $A_L(h)$ everywhere, the loss
		is at most $\bigl[\Sigma(h)-r_L(h)+M\bigr]P_1^\tau(E_\tau)\leq
		\text{const}\cdot C\tau$.  Hence
		$V^R(\tau)\geq m(\tau)+V^{\NR}(T-\tau;c,M)-O(\tau)
		\geq V^{\NR}(T;c,M)-O(\tau)$.
		
		(ii)  Start from the menu that uses only the two stops, assigning $A_G$ on
		the cheapest-contract reward event $B_\tau$ and $A_B$ off it.  It delivers
		the agent the rent profile of the no-review contract at horizon $\tau$, so
		obedience binds exactly and its value is $V^{\NR}(\tau;c,M)$.  Reassign
		$A_L(h)$ in place of $A_B$ on the event $E_\tau$ in the statement.
		Obedience is relaxed, since the rent rises by $r_L(h)$ on a set with
		$P_1^\tau(E_\tau)>P_0^\tau(E_\tau)$, and the
		principal's surplus changes by
		$\bigl[\Sigma(h)-r_L(h)\bigr]P_1^\tau(E_\tau)
		=V^{\NR}(h)\,P_1^\tau(E_\tau)$.  With
		$V^{\NR}(h)=\nu h+o(h)$ and $\nu>0$ this is $\nu ph+o(h)>0$, giving the
		stated bound.  Subtracting it from Lemma~\ref{lem:review-boundary-values}, dividing by
		$h$, and taking upper limits gives the slope inequality in the statement.
	\end{proof}

	\begin{proof}[Proof of Proposition~\ref{oa:prop:general-stops}]
		By \eqref{oa:eq:general-breakpoints} the bad stop is strictly optimal at $x$ exactly
		when $Q_\lambda(x)<q_0$, where $q_0=-\Sigma(h)/r_L(h)$, that is when
		$\ell_\tau(x)>(\lambda-1-q_0)/\lambda$.  Such $x$ have positive
		probability if and only if
		$\ell^{\max}_\tau>(\lambda-1-q_0)/\lambda$, that is
		$\lambda\bigl(\ell^{\max}_\tau-1\bigr)>-1-q_0=\Sigma(h)/r_L(h)-1$.
		If $\ell^{\max}_\tau=\infty$ this holds for every $\lambda>0$.  Similarly
		the good stop is strictly optimal when
		$Q_\lambda(x)>q_2=\Sigma(h)/(M-r_H(h))$, that is when
		$\ell_\tau(x)<1-(1+q_2)/\lambda$.  Since the essential infimum of
		$\ell_\tau$ is $1/\Lambda^{\max}_\tau$, such $x$ have positive
		probability if and only if
		$1/\Lambda^{\max}_\tau<1-(1+q_2)/\lambda$,
		which is \eqref{oa:eq:good-stop-condition}.  Equality in either comparison is
		a pointwise tie and can matter only when the boundary likelihood ratio has an
		atom of positive probability.
	\end{proof}
	
	\newtheorem*{randomlemma}{Lemma}
	
	\section{Random review dates}
	\label{oa:random-review}
	
	\subsection{The attainable set and duality}
	\label{oa:sr:duality}
	
	For $\tau\in(0,T]$ let $\mathfrak A_\tau$ be the measurable menus $a_\tau(x)\in\mathcal A(T-\tau)$
	with $\mathcal A(0)=\{(r,0):r\in[0,M]\}$, so that the date-$T$ menu is a bounded payment
	schedule on the terminal signal, and let
	\begin{equation}
		\mathcal B_\tau=\bigl\{\bigl(J(\tau,a),P(\tau,a)\bigr):a\in\mathfrak A_\tau\bigr\}\subset\mathbb R^2 .
		\label{eq:random-Btau}
	\end{equation}
	At the zero-date endpoint set
	$\mathcal B_0=\{b:\ \exists\,\tau_n\downarrow0,\ b_n\in\mathcal B_{\tau_n},\ b_n\to b\}$,
	and put $\mathcal S=\bigcup_{\tau\in[0,T]}\mathcal B_\tau$ and $\mathcal C=\operatorname{co}(\mathcal S)$.
	The zero-date convention is innocuous: every point of $\mathcal B_0$ has zero incentive
	surplus and payoff at most $V^{\NR}(T;c,M)$, which the ordinary no-review contract already
	attains at $\tau=T$.
	
	\begin{lemma}
		\label{oa:lem:random-attainable}
		For every $\tau\in(0,T]$, $\mathcal B_\tau$ is compact and convex, and $\mathcal S$ is
		compact.  If a lottery over finitely many menus at a common positive date $\tau$ generates
		$(J,P)$, some single deterministic menu at that date generates the same pair.
	\end{lemma}
	
	\begin{proof}
		Fix $\tau$ and write $h=T-\tau$.  The action set $\mathcal A(h)$ is compact and uniformly
		bounded, $0\leq r\leq M$ and $0\leq s\leq T$.  For a measurable selection $a=(r,s)$ the
		two coordinates of $(J,P)$ are integrals of the bounded measurable vector
		\begin{equation}
			v_\tau(x,a)=\Bigl(r\bigl[g_1^\tau(x)-g_0^\tau(x)\bigr],\ (s-r)g_1^\tau(x)\Bigr)
			\label{eq:random-vector-integrand}
		\end{equation}
		plus the constant $(-c\tau,\tau)$.  The correspondence $F_\tau:x\mapsto v_\tau(x,\mathcal A(h))$
		is measurable, compact-valued and integrably bounded, and $\mathcal B_\tau$ is the translate
		by $(-c\tau,\tau)$ of its Aumann integral: every menu induces a selection, and conversely
		Filippov's implicit function theorem turns any measurable selection of $F_\tau$ into a
		measurable menu.  That integral is compact, and is convex by Lyapunov's theorem because the
		measure $[g_0^\tau+g_1^\tau]dx$ is atomless --- even though $\mathcal A(h)$ is a union of two
		segments.
		
		The same argument purifies.  For two menus $a^1,a^2$ and $p\in[0,1]$, the atomless vector
		measure $\mu(B)=\int_B[v_\tau(x,a^1)-v_\tau(x,a^2)]dx$ has compact convex range containing
		$0$ and $\mu(\mathbb R)$, hence $p\mu(\mathbb R)$; using $a^1$ on a set $B_p$ with
		$\mu(B_p)=p\mu(\mathbb R)$ and $a^2$ elsewhere reproduces both integrals.  Finite lotteries
		follow by induction.
		
		Compactness across dates: for $\tau_n\to\tau>0$, $g_j^{\tau_n}\to g_j^\tau$ in $L^1$, and
		$h\mapsto\mathcal A(h)$ has compact values and closed graph on $[0,T]$ --- under
		Assumption~\ref{as:maintained}, $r_D(h)$ and $r_E(h)$ are continuous on $(0,T]$ and
		$r_D(h)\to0$, $r_E(h)\to M$ as $h\downarrow0$
		(Lemma~\ref{lem:continuation-rent-limits}), so
		$\mathcal A(h)\to\mathcal A(0)$ in the Hausdorff metric.  Representing $a_n$ by the measure
		$m_n(x)dx\,\delta_{a_n(x)}(da)$ with $m_n=(g_0^{\tau_n}+g_1^{\tau_n})/2$, uniform boundedness
		and tightness give a weakly convergent subsequence whose limit disintegrates as
		$m(x)dx\,\kappa_x(da)$ with $\kappa_x$ supported on $\mathcal A(T-\tau)$.  The weights
		$g_j^{\tau_n}/m_n$ are bounded by two and converge locally uniformly, so both moments in
		\eqref{eq:random-vector-integrand} converge.  The barycenter is a selection of
		$\operatorname{co}F_\tau$, whose Aumann integral coincides with that of $F_\tau$ for atomless
		measures, and the Filippov step converts it into a deterministic menu.  Hence
		$b\in\mathcal B_\tau$.
		
		If instead $\tau_n\downarrow0$, the total variation between the two review-score laws is
		$2\Phi(\sqrt{\tau_n}/2)-1\to0$, so rents bounded by $M$ give $J\to0$; and every action
		satisfies $s-r\leq(1-c)h_n-r_D(h_n)$ for large $n$, so $P\leq\tau_n+(1-c)h_n-r_D(h_n)\to
		V^{\NR}(T;c,M)$.  Every point of $\mathcal B_0$ is thus $(0,p)$ with $p\leq V^{\NR}(T;c,M)$,
		and $(0,V^{\NR}(T;c,M))\in\mathcal B_0$.  Compactness of $\mathcal S$ follows.
	\end{proof}
	
	Since a lottery takes probability-weighted averages of deterministic pairs, and any
	finite convex combination of elements of $\mathcal S$ is implemented by a lottery over dates
	(merging duplicate dates by the lemma, and replacing a $\mathcal B_0$ point by
	$(0,V^{\NR}(T;c,M))\in\mathcal B_T$), Problem~\eqref{eq:random-review-primal} is exactly
	\begin{equation}
		V^{S}=\max\{p:(j,p)\in\mathcal C,\ j\geq0\},
		\qquad\text{and likewise}\qquad
		V^R(\tau)=\max\{p:(j,p)\in\mathcal B_\tau,\ j\geq0\}.
		\label{eq:random-convex-primal}
	\end{equation}
	Strict feasibility holds under Assumption~\ref{as:maintained}: since $M>\underline M(T,c)\geq
	M_I(T,c)$ and $M_I(\cdot,c)$ is strictly increasing (Lemma~\ref{lem:Mbar}), a bounded threshold
	reward supplies strictly more than $c\tau$ at every $\tau\in(0,T]$.  With
	$\mathcal L^*(\tau,\lambda)\equiv\max_{(j,p)\in\mathcal B_\tau}\{p+\lambda j\}$, strong duality
	for these one-constraint concave programs therefore gives
	\begin{equation}
		V^{D}=\sup_{\tau\in[0,T]}\inf_{\lambda\geq0}\mathcal L^*(\tau,\lambda),
		\qquad
		V^{S}=\inf_{\lambda\geq0}\sup_{\tau\in[0,T]}\mathcal L^*(\tau,\lambda),
		\label{eq:random-minimax}
	\end{equation}
	the second because a linear functional has the same maximum on $\mathcal S$ and on
	$\operatorname{co}(\mathcal S)$.  A minimizing multiplier exists: strict feasibility supplies
	$(j_+,p_+)\in\mathcal C$ with $j_+>0$, so the dual objective diverges as $\lambda\to\infty$.
	Substituting \eqref{eq:random-P-J} shows $\mathcal L^*(\tau,\cdot)$ is the pointwise problem
	of Section~\ref{sec:deriving-optimal-menu}, so at every date and multiplier a maximizing menu
	can be selected from the four ordered actions of
	Proposition~\ref{prop:four-region-review-menu}; the only difference is that a lottery uses one
	common multiplier across all dates in its support.
	
	\begin{proof}[Proof of Proposition~\ref{prop:random-two-dates}]
		The inequality $V^D\leq V^S$ is the minimax inequality applied to
		\eqref{eq:random-minimax}, and also follows because a degenerate $\Gamma$ reproduces every
		deterministic date.
		
		Let $\lambda^*$ minimize the dual and let
		$\mathcal S^*=\{(j,p)\in\mathcal S:p+\lambda^*j=V^S\}$, which is nonempty and compact.  The
		exposed face $\mathcal F^*=\{(j,p)\in\mathcal C:p+\lambda^*j=V^S\}$ equals
		$\operatorname{co}(\mathcal S^*)$: by Carath\'eodory every point of $\mathcal C$ is a convex
		combination of at most three points of $\mathcal S$, and such an average attains the maximum
		of a linear functional only if all weight sits on maximizers.
		
		If $\lambda^*=0$, strong duality gives $V^S=\max\{p:(j,p)\in\mathcal C\}$, so an optimal
		primal pair lies in $\mathcal F^*$ and is feasible; since every point of $\mathcal S^*$ has
		payoff $V^S$, some element of $\mathcal S^*$ has $j\geq0$, and one date suffices.  If
		$\lambda^*>0$, any optimal primal pair has $p=V^S$ and $p+\lambda^*j\leq V^S$, hence $j=0$:
		aggregate obedience binds and $(0,V^S)\in\mathcal F^*$.  That face lies on the line
		$p+\lambda^*j=V^S$, so $(0,V^S)$ is a convex combination of at most two points of
		$\mathcal S^*$ --- one if zero is itself an attainable incentive coordinate there, otherwise
		two whose incentive coordinates bracket zero.  Merging a repeated date by
		Lemma~\ref{oa:lem:random-attainable} gives at most two distinct dates.
		
		Suppose $V^S>V^D$.  No optimum can sit at a single date: support at the zero-date endpoint is
		impossible because every pair there has payoff at most $V^{\NR}(T;c,M)\leq V^D<V^S$, and at a
		single $\tau\in(0,T]$ the aggregate pair lies in the convex compact $\mathcal B_\tau$ and is
		therefore attained by one feasible deterministic contract, giving $V^D\geq V^S$.  Neither
		supported surplus can vanish, since a supported pair with $J_i=0$ lies in $\mathcal S^*$,
		hence has $P_i=V^S$, and is again a feasible deterministic contract.  As their convex
		combination is zero the two surpluses have opposite signs; label the two pairs $b_+$ and
		$b_-$ so that $J_+>0>J_-$, and solve $pJ_++(1-p)J_-=0$ for
		\eqref{eq:two-date-mixing-probability}.
		
		For the test, necessity follows by substituting \eqref{eq:two-date-mixing-probability} into
		$V^S=p(b_+,b_-)P_++\bigl[1-p(b_+,b_-)\bigr]P_-$, which is $\Pi(b_+,b_-)$ by
		\eqref{eq:binding-mixture-payoff} and equals $V^S>V^D$, so
		\eqref{eq:strict-gain-primal-test} holds.  Conversely, two pairs satisfying
		\eqref{eq:strict-gain-primal-test} can be mixed with probability $p(b_+,b_-)$ to give zero
		aggregate surplus and payoff above $V^D$, so $V^S>V^D$; they cannot share a date, since the
		lemma would then replace the mixture by a feasible deterministic menu at that date.
	\end{proof}
	
	Equivalently, each attainable pair generates an affine dual payoff $P+\lambda J$; the surplus
	branch slopes up and the deficit branch slopes down, and \eqref{eq:strict-gain-primal-test}
	says two such lines cross above $V^D$.  At an interior minimizing multiplier zero must lie in
	the convex hull of the slopes of the exposed lines: an exposed line of slope zero is a pair
	that is obedient on its own, and deterministic timing attains $V^S$; under a strict gain there
	is none, and the zero-slope combination of two opposite-sloped lines is the optimal lottery.
	
	\subsection{Two dates under pointwise effort}
	\label{oa:sr:pointwise}
	
	Throughout this subsection continuation actions after the review are the pointwise-feasible
	ones of Section~\ref{sec:pointwise-effort}, so $\widetilde A_E$ replaces $A_E$.
	
	\begin{proof}[Proof of Proposition~\ref{prop:random-pointwise}]
		\emph{Dynamic reduction.}  Before $\tau_1$ no review can occur, so the agent receives no
		information about the realized branch, and because the review-score law and the linear
		effort cost depend on the pre-$\tau_1$ path only through total effort, any deviation is
		summarized by $e\in[0,\tau_1]$.  If the early review occurs his utility is $U_1(e)$.  If it
		does not, he learns the date is $\tau_2$; starting from $e$, feasibility confines cumulative
		effort at $\tau_2$ to $[e,e+d]$, so his continuation value is $\widehat U_2(e)$ in
		\eqref{eq:random-U}.  Choosing $e$ is therefore worth
		$pU_1(e)+(1-p)\widehat U_2(e)$.  On path $e=\tau_1$, and obedience after non-arrival requires
		cumulative effort $\tau_2$, which is $U_2(\tau_2)=\widehat U_2(\tau_1)$.  Given that,
		full effort before $\tau_1$ is optimal if and only if
		$pU_1(\tau_1)+(1-p)U_2(\tau_2)\geq pU_1(e)+(1-p)\widehat U_2(e)$ for every $e$, which
		rearranges to \eqref{eq:random-exact-IC}.
		
		\emph{The bound.}  If $U_2(\tau_2)\geq U_2(m)$ on $[0,\tau_2]$ then
		$\widehat U_2(e)\leq U_2(\tau_2)$, so $D_2(e)\geq0$; the same condition applied at
		$m=\tau_1$ gives $U_2(\tau_2)=\widehat U_2(\tau_1)$.  Where $D_1(e)\geq0$,
		\eqref{eq:random-exact-IC} holds for every $p$.  Where $D_1(e)<0$ it is equivalent to
		$p[D_2(e)-D_1(e)]\leq D_2(e)$, that is $p\leq D_2(e)/[D_2(e)-D_1(e)]$; deterring all
		deviations is therefore $p\leq\overline p$ in \eqref{eq:random-pbar}.
		
		\emph{Dominance.}  Under full effort the lottery is worth $V(p)=P_2+p(P_1-P_2)$, affine and
		increasing in $p$, so $V(p)>V^{D,\pw}$ if and only if $p>\underline p$.  A feasible strict
		improvement exists exactly when $(\underline p,1]$ meets $[0,\overline p]$, which is
		\eqref{eq:random-dominance}.
	\end{proof}
	
	Under a shape condition the bound reduces to a comparison of marginal incentive returns.
	Suppose the early branch locally favors less effort, $U_1'<0$ on $[0,\tau_1]$, and that after
	non-arrival the late-branch optimum is always to work at the maximal rate,
	\begin{equation}
		\widehat U_2(e)=U_2(e+d),
		\qquad
		U_2'(e+d)>0
		\qquad\text{on }[0,\tau_1].
		\label{eq:random-upper-end-late}
	\end{equation}
	Write $\alpha=-U_1'>0$ and $\beta(\cdot)=U_2'(\cdot+d)>0$, so that
	$-D_1(e)=\int_e^{\tau_1}\alpha$ and $D_2(e)=\int_e^{\tau_1}\beta$, and hence
	\begin{equation}
		\frac{D_2(e)}{D_2(e)-D_1(e)}
		=\frac{\int_e^{\tau_1}\beta}{\int_e^{\tau_1}[\alpha+\beta]}
		=\frac{\int_e^{\tau_1}\rho\,w}{\int_e^{\tau_1}w},
		\qquad
		\rho\equiv\frac{\beta}{\alpha+\beta},
		\quad w\equiv\alpha+\beta,
		\label{eq:random-weighted-rho}
	\end{equation}
	a weighted average of $\rho$ on $[e,\tau_1]$.
	
	\begin{corollary}
		\label{oa:cor:random-local-global}
		Under \eqref{eq:random-upper-end-late}, if $\rho$ is weakly increasing the binding
		deviation is $e=0$ and $\overline p=D_2(0)/[D_2(0)-D_1(0)]$; if $\rho$ is weakly decreasing
		the binding deviation is infinitesimal coasting just before $\tau_1$ and
		$\overline p=\rho(\tau_1)$.  In the latter case \eqref{eq:random-dominance} reads
		\begin{equation}
			\beta(\tau_1)\bigl(P_1-V^{D,\pw}\bigr)>\alpha(\tau_1)\bigl(V^{D,\pw}-P_2\bigr).
			\label{eq:random-marginal-dominance}
		\end{equation}
	\end{corollary}
	
	\begin{proof}
		Deleting the lowest part of $[e,\tau_1]$ raises a weighted average of a weakly increasing
		$\rho$, so the infimum in \eqref{eq:random-weighted-rho} is at $e=0$; if $\rho$ is weakly
		decreasing the average falls as the left end is deleted, so the infimum is approached as
		$e\uparrow\tau_1$, and continuity gives $\rho(\tau_1)$.  Substituting into
		$\underline p<\overline p$ and clearing positive denominators gives
		\eqref{eq:random-marginal-dominance}.
	\end{proof}
	
	Randomization is thus profitable when the payoff gained per unit of incentive deficit on the
	attractive branch exceeds the payoff sacrificed per unit of incentive capacity on the
	supporting branch.
	
	\subsection{Numerical construction}
	\label{oa:sr:numerics}
	
	All figures are direct evaluations of the Gaussian formulas; no simulation is used.  Set
	$T=1.5$, $c=0.72$, $M=2.60$ throughout.
	
	\paragraph{Phase-level example.}
	Optimizing first in $\lambda$ and then in $\tau$ gives
	$\tau^D=0.861662$, $\lambda^D=1.824076$ and $V^D=0.07850435$.  Reversing the order gives the
	minimizing common multiplier $\lambda^*=1.969842$, at which the maximized Lagrangian has two
	global maximizers in the review date:
	
	\begin{center}
		\begin{tabular}{lccc}
			& review date & principal payoff $P$ & incentive surplus $J$ \\
			\hline
			early branch ($b_-$) & $0.094868$ & $\phantom{-}0.12577784$ & $-0.02182519$ \\
			later branch ($b_+$) & $0.816614$ & $-0.02796598$ & $\phantom{-}0.05622360$
		\end{tabular}
	\end{center}
	
	The later branch carries the incentive surplus, so it is $b_+$.
	Equation~\eqref{eq:two-date-mixing-probability} places probability
	$p(b_+,b_-)=0.279635$ on it and the remaining $0.720365$ on the early branch; aggregate
	obedience binds, and $V^S=0.08278565$, so $V^S-V^D=0.00428130>0$.  At a
	candidate $(\tau,\lambda)$ these values require only the breakpoints $q_0,0,q_2$ of
	\eqref{eq:three-breakpoints}, their score cutoffs from \eqref{eq:k-formula} with
	$k_j=+\infty$ for unreached breakpoints, and the region probabilities
	\begin{equation}
		\pi_i^j(\tau,\lambda)
		=\Phi\!\left(\frac{\overline k_i-j\tau}{\sqrt\tau}\right)-\Phi\!\left(\frac{\underline k_i-j\tau}{\sqrt\tau}\right),
		\qquad j\in\{0,1\},
		\label{eq:random-region-probabilities}
	\end{equation}
	from which $P=\tau+\sum_i(s_i-r_i)\pi_i^1$ and $J=\sum_ir_i[\pi_i^1-\pi_i^0]-c\tau$.
	
	\paragraph{Pointwise example.}
	The same two dates are used, so $h_1=1.405132$ and $h_2=0.683386$.  With $z_D(h)$ the
	difficult root of $\Delta(h,z)=ch/M$ and $\zeta(h)$ the positive root of
	$\phi(\zeta)=c\sqrt h/M$ from \eqref{eq:harder-easy-bar},
	\begin{equation*}
		r_D(h)=M\Phi\bigl(z_D(h)-\sqrt h\bigr),
		\qquad
		\widetilde r_E(h)=M\Phi\bigl(\zeta(h)\bigr)-ch,
	\end{equation*}
	giving $z_D=0.03576003$, $\zeta=0.62451798$, $r_D=0.32538914$, $\widetilde r_E=0.89633126$
	on the early branch and $z_D=-0.64321880$, $\zeta=1.05396840$, $r_D=0.18406886$,
	$\widetilde r_E=1.72849548$ on the late branch.  The menus use cutoffs
	
	\begin{center}
		\begin{tabular}{lcccc}
			& $\tau_i$ & $k_{i1}$ & $k_{i2}$ & $k_{i3}$ \\
			\hline
			early & $0.094868$ & $-0.044984$ & $0.878889$ & $0.970519$ \\
			late  & $0.816614$ & $\phantom{-}0.404840$ & $1.146958$ & $1.503501$
		\end{tabular}
	\end{center}
	
	so that the rent schedule is the step function with jumps
	$\delta_{i1}=r_D(h_i)$, $\delta_{i2}=\widetilde r_E(h_i)-r_D(h_i)$ and
	$\delta_{i3}=M-\widetilde r_E(h_i)$, and
	\begin{equation}
		U_i(e)=\sum_{j=1}^{3}\delta_{ij}\Phi\!\left(\frac{e-k_{ij}}{\sqrt{\tau_i}}\right)-ce,
		\qquad
		U_i'(e)=\frac{1}{\sqrt{\tau_i}}\sum_{j=1}^{3}\delta_{ij}\phi\!\left(\frac{e-k_{ij}}{\sqrt{\tau_i}}\right)-c .
		\label{eq:random-U-closed}
	\end{equation}
	Under full effort the region probabilities are $(0.32489,0.66965,0.00322,0.00223)$ on the
	early branch and $(0.32431,0.31834,0.13375,0.22359)$ on the late branch, whence
	$P_1=0.1330054377$ and $P_2=0.0319887113$.
	
	For the sequential check, $U_1(\tau_1)=0.15828980$ and $U_1(0)=0.18420688$, so
	$D_1(0)=-0.02591708$, and $U_1'\in[-0.27499,-0.27207]$ on $[0,\tau_1]$: the early branch by
	itself strictly favors less pre-review effort throughout.  On the late branch
	$U_2(\tau_2)=0.28317018$ while $U_2'$ has a unique zero on $[0,\tau_2]$ at
	$m_0=0.35917$, a strict minimum with $U_2(m_0)=0.21524173$; hence $U_2(\tau_2)$ is the global
	maximum and the late branch is self-incentive-compatible.  Because $U_2$ has only this
	interior minimum, $\widehat U_2(e)=\max\{U_2(e),U_2(e+d)\}$, with the two terms crossing once at
	$e_c=0.00215$.  At $e=0$ the agent prefers not to add effort after learning the date, so
	$\widehat U_2(0)=U_2(0)=0.25991515$ and $D_2(0)=0.02325503$.  The deviation-specific bound
	$q(e)=D_2(e)/[D_2(e)-D_1(e)]$ is increasing on $[0,e_c]$ and on $[e_c,\tau_1)$, so the
	binding deviation is $e=0$ and
	\begin{equation*}
		\overline p=q(0)=\frac{0.02325503}{0.02325503+0.02591708}=0.47293135 .
	\end{equation*}
	Since the phase-level problem checks fewer deviations both before and after the review,
	$V^{D,\pw}\leq V^D=0.07850435$ and $\underline p\leq0.46047462<\overline p$.  At $p=0.465$
	the worst deviation still loses $0.465D_1(0)+0.535D_2(0)=0.00039>0$, and the contract is
	worth $0.465P_1+0.535P_2=0.07896149$, exceeding $V^D$ by $0.00045714$.	
	\addtocontents{toc}{\protect\endgroup}
	
\end{document}